\documentclass[a4paper,UKenglish]{article}

\usepackage[fleqn]{amsmath}
\usepackage{amssymb}
\usepackage{amsthm}
\usepackage{authblk}
\usepackage{cite}
\usepackage{comment}
\usepackage[T1]{fontenc}
\usepackage[colorlinks]{hyperref}
\usepackage{inconsolata}
\usepackage[utf8]{inputenc}
\usepackage{listings}
\usepackage{MnSymbol}
\usepackage{multicol}
\usepackage{multirow}
\usepackage{paralist}
\usepackage{slashed}
\usepackage{stmaryrd}
\usepackage{subcaption}
\usepackage{svrsymbols}
\usepackage{todonotes}
\usepackage{tikz}
\usepackage{xcolor}
\usepackage{url}
\usepackage{a4wide}

\usepackage{pgfplots}
\pgfplotsset{compat=1.7}

\usepackage{lstcoq}

\definecolor{codegreen}{rgb}{0,0.6,0}
\definecolor{codegray}{rgb}{0.5,0.5,0.5}
\definecolor{codepurple}{rgb}{0.58,0,0.82}
\definecolor{backcolour}{rgb}{0.95,0.95,0.92}

\lstdefinestyle{atomis}{
    language=java,
    backgroundcolor=\color{white},   
    commentstyle=\color{codegreen},
    keywordstyle=\color{magenta},
    numberstyle=\tiny\color{codegray},
    stringstyle=\color{codepurple},
    basicstyle=\ttfamily\footnotesize,
    breakatwhitespace=false,         
    breaklines=true,                 
    captionpos=b,                    
    keepspaces=true,                 
    numbers=left,                    
    numbersep=5pt,                  
    showspaces=false,                
    showstringspaces=false,
    showtabs=false,                  
    tabsize=2,
    numberblanklines=false,
    mathescape=true,
    classoffset=1, %
    morekeywords= [2]{
        @Atomic, @All, @AtomiSVar, @AtomiSSpec, @AtomicityOf, @NoHLDR,
        atomic, non_atomic},
    keywordstyle= [2]\color{blue},
    classoffset=0,
        morekeywords={
        Unit,  let,  in, def, finish, async},
}

\newcommand{\icode}[1]{\lstinline[basicstyle=\ttfamily\normalsize]{#1}}

\newcommand{\nope}[1]{}

\usepackage{xspace}
\usepackage{xspace}

\newtheorem{theorem}{Theorem}
\newtheorem{lemma}[theorem]{Lemma}
\newtheorem{definition}[theorem]{Definition}

\newcommand{\rulename}[1]{[\textsc{#1}]}

\newcommand{\cf}{\textit{cf.}\xspace}
\newcommand{\eg}{\textit{e.g.}\xspace}
\newcommand{\ie}{\textit{i.e.},\xspace}

\newcommand{\regular}{regular\xspace}

\newcommand{\atan}[2]{(#1) \rightarrow #2}

\newcommand{\ateq}[2]{#1 = #2}
\newcommand{\valeq}[2]{#1 = #2}

\newcommand{\preprocess}{\func{preProcess}}

\usepackage{supertabular}
\usepackage{ifthen}
\usepackage{alltt}%
\usepackage{color}
\newcommand{\ottdrule}[4][]{{\displaystyle\ottdrulename{#4}~\frac{\begin{array}{l}#2\end{array}}{#3}}}
\newcommand{\ottusedrule}[1]{$#1\qquad$}  %
\newcommand{\ottpremise}[1]{ #1 \\}
\newenvironment{ottdefnblock}[3][]{ \framebox{#2} {~~~~~#3}}{}    %

\newcommand{\ottnt}[1]{\mathit{#1}}
\newcommand{\ottmv}[1]{\mathit{#1}}
\newcommand{\ottkw}[1]{\mathbf{#1}}
\newcommand{\ottsym}[1]{#1}

\newcommand{\ottdrulename}[1]{\textsc{#1}}

\usepackage{latexsym}
\providecommand{\kwool}[1]{\texttt{\textbf{#1}%
}%
}
\providecommand{\startblock}{\texttt{\symbol{123}%
}%
}
\providecommand{\finishblock}{\texttt{\symbol{125}%
}%
}

\renewcommand{\mathbf}{\textbf}
\providecommand{\ottlinebreakhack}{}
\newcommand{\ottdrulewfXXcast}[1]{\ottdrule[#1]{%
\ottpremise{\Gamma  \vdash  \ottnt{e}  \ottsym{:}  \ottnt{t'}}%
\ottpremise{ \ottnt{t'}  <:  \ottnt{t} }%
}{
\Gamma  \vdash  \ottsym{(}  \ottnt{t}  \ottsym{)}  \ottnt{e}  \ottsym{:}  \ottnt{t}}{%
{\ottdrulename{wf\_cast}}{}%
}}

\newcommand{\ottdrulewfXXlet}[1]{\ottdrule[#1]{%
\ottpremise{\Gamma  \vdash  \ottnt{e_{{\mathrm{1}}}}  \ottsym{:}  \ottnt{t_{{\mathrm{1}}}}}%
\ottpremise{\Gamma  \ottsym{,}  \ottmv{x}  \ottsym{:}  \ottnt{t_{{\mathrm{1}}}}  \vdash  \ottnt{e_{{\mathrm{2}}}}  \ottsym{:}  \ottnt{t}}%
}{
\Gamma  \vdash   \kwool{let}\,   \ottmv{x}  =  \ottnt{e_{{\mathrm{1}}}} ~\kwool{in}~ \ottnt{e_{{\mathrm{2}}}}   \ottsym{:}  \ottnt{t}}{%
{\ottdrulename{wf\_let}}{}%
}}

\newcommand{\ottdrulewfXXcall}[1]{\ottdrule[#1]{%
\ottpremise{\Gamma  \ottsym{(}  \ottmv{x}  \ottsym{)}  \ottsym{=}  \ottnt{t_{{\mathrm{1}}}}}%
\ottpremise{ \Gamma  \vdash  \ottnt{e}  \ottsym{:}  \ottnt{t_{{\mathrm{2}}}}  \ottlinebreakhack }%
\ottpremise{\ottkw{msigs} \, \ottsym{(}  \ottnt{t_{{\mathrm{1}}}}  \ottsym{)}  \ottsym{(}  \ottmv{m}  \ottsym{)}  \ottsym{=}  \ottmv{y}  \ottsym{:}  \ottnt{t_{{\mathrm{2}}}}  \rightarrow  \ottnt{t}}%
}{
\Gamma  \vdash  \ottsym{x.m}  \ottsym{(}  \ottnt{e}  \ottsym{)}  \ottsym{:}  \ottnt{t}}{%
{\ottdrulename{wf\_call}}{}%
}}

\newcommand{\ottdrulewfXXvar}[1]{\ottdrule[#1]{%
\ottpremise{ \vdash  \Gamma  \!\!\! }%
\ottpremise{ \Gamma  \ottsym{(}  \ottmv{x}  \ottsym{)}  \ottsym{=}  \ottnt{t}  \!\!\! }%
}{
\Gamma  \vdash  \ottmv{x}  \ottsym{:}  \ottnt{t}}{%
{\ottdrulename{wf\_var}}{}%
}}

\newcommand{\ottdrulewfXXloc}[1]{\ottdrule[#1]{%
\ottpremise{\vdash  \Gamma}%
\ottpremise{\Gamma  \ottsym{(}  \iota  \ottsym{)}  \ottsym{=}  \ottmv{C}}%
\ottpremise{ \ottmv{C}  <:  \ottnt{t} }%
}{
\Gamma  \vdash  \iota  \ottsym{:}  \ottnt{t}}{%
{\ottdrulename{wf\_loc}}{}%
}}

\newcommand{\ottdrulewfXXnull}[1]{\ottdrule[#1]{%
\ottpremise{\vdash  \Gamma}%
\ottpremise{\vdash  \ottnt{t}}%
}{
\Gamma  \vdash  \kwool{null}  \ottsym{:}  \ottnt{t}}{%
{\ottdrulename{wf\_null}}{}%
}}

\newcommand{\ottdrulewfXXselect}[1]{\ottdrule[#1]{%
\ottpremise{ \Gamma  \vdash  \ottmv{x}  \ottsym{:}  \ottmv{C}  \ottlinebreakhack }%
\ottpremise{ \ottkw{fields} \, \ottsym{(}  \ottmv{C}  \ottsym{)}  ( \ottmv{f} ) =  \ottnt{t} }%
}{
\Gamma  \vdash  \ottmv{x}  \ottsym{.}  \ottmv{f}  \ottsym{:}  \ottnt{t}}{%
{\ottdrulename{wf\_select}}{}%
}}

\newcommand{\ottdrulewfXXupdate}[1]{\ottdrule[#1]{%
\ottpremise{\Gamma  \vdash  \ottmv{x}  \ottsym{:}  \ottmv{C}}%
\ottpremise{ \Gamma  \vdash  \ottnt{e}  \ottsym{:}  \ottnt{t}  \ottlinebreakhack }%
\ottpremise{ \ottkw{fields} \, \ottsym{(}  \ottmv{C}  \ottsym{)}  ( \ottmv{f} ) =  \ottnt{t} }%
}{
\Gamma  \vdash  \ottmv{x}  \ottsym{.}  \ottmv{f}  \ottsym{=}  \ottnt{e}  \ottsym{:}  \kwool{Unit}}{%
{\ottdrulename{wf\_update}}{}%
}}

\newcommand{\ottdrulewfXXnew}[1]{\ottdrule[#1]{%
\ottpremise{\vdash  \Gamma}%
\ottpremise{\vdash  \ottmv{C}}%
}{
\Gamma  \vdash  \kwool{new} \, \ottmv{C}  \ottsym{:}  \ottmv{C}}{%
{\ottdrulename{wf\_new}}{}%
}}

\newcommand{\ottdrulewfXXfj}[1]{\ottdrule[#1]{%
\ottpremise{\Gamma  \ottsym{=}  \Gamma_{{\mathrm{1}}}  +  \Gamma_{{\mathrm{2}}}}%
\ottpremise{\Gamma_{{\mathrm{1}}}  \vdash  \ottnt{e_{{\mathrm{1}}}}  \ottsym{:}  \ottnt{t_{{\mathrm{1}}}}}%
\ottpremise{\Gamma_{{\mathrm{2}}}  \vdash  \ottnt{e_{{\mathrm{2}}}}  \ottsym{:}  \ottnt{t_{{\mathrm{2}}}}}%
\ottpremise{\Gamma  \vdash  \ottnt{e}  \ottsym{:}  \ottnt{t}}%
}{
\Gamma  \vdash  \kwool{finish} \, \,\startblock\, \, \kwool{async} \, \,\startblock\,  \ottnt{e_{{\mathrm{1}}}}  \,\finishblock\, \, \kwool{async} \, \,\startblock\,  \ottnt{e_{{\mathrm{2}}}}  \,\finishblock\,  \,\finishblock\,  \ottsym{;}  \ottnt{e}  \ottsym{:}  \ottnt{t}}{%
{\ottdrulename{wf\_fj}}{}%
}}

\newcommand{\ottdrulewfXXlock}[1]{\ottdrule[#1]{%
\ottpremise{\Gamma  \vdash  \ottmv{x}  \ottsym{:}  \ottnt{t_{{\mathrm{2}}}}}%
\ottpremise{ \Gamma  \vdash  \ottnt{e}  \ottsym{:}  \ottnt{t}  \!\!\! }%
}{
\Gamma  \vdash   \kwool{lock} ( \ottmv{x} ) \kwool{\,in\,}  \ottnt{e}   \ottsym{:}  \ottnt{t}}{%
{\ottdrulename{wf\_lock}}{}%
}}

\newcommand{\ottdrulewfXXlocked}[1]{\ottdrule[#1]{%
\ottpremise{\Gamma  \vdash  \ottnt{e}  \ottsym{:}  \ottnt{t}}%
\ottpremise{\Gamma  \ottsym{(}  \iota  \ottsym{)}  \ottsym{=}  \ottnt{t_{{\mathrm{2}}}}}%
}{
\Gamma  \vdash   \kwool{locked}_{ \iota } \{  \ottnt{e}  \}   \ottsym{:}  \ottnt{t}}{%
{\ottdrulename{wf\_locked}}{}%
}}

\newcommand{\ottdefnwfXXExpression}[1]{\begin{ottdefnblock}[#1]{$\Gamma  \vdash  \ottnt{e}  \ottsym{:}  \ottnt{t}$}{}
\ottusedrule{\ottdrulewfXXcast{}}
\ottusedrule{\ottdrulewfXXlet{}}
\ottusedrule{\ottdrulewfXXcall{}}
\ottusedrule{\ottdrulewfXXvar{}}
\ottusedrule{\ottdrulewfXXloc{}}
\ottusedrule{\ottdrulewfXXnull{}}
\ottusedrule{\ottdrulewfXXselect{}}
\ottusedrule{\ottdrulewfXXupdate{}}
\ottusedrule{\ottdrulewfXXnew{}}
\ottusedrule{\ottdrulewfXXfj{}}
\end{ottdefnblock}}

\newcommand{\ottdrulewfXXheap}[1]{\ottdrule[#1]{%
\ottpremise{ \forall \, \iota  \ottsym{:}  \ottmv{C} \, \!\!\in\! \, \Gamma  \ottsym{.}  \ottnt{H}  \ottsym{(}  \iota  \ottsym{)}  \ottsym{=}   ( \ottmv{C} ,  \ottnt{F} ,  L )   \land  \Gamma  \ottsym{;}  \ottmv{C}  \vdash  \ottnt{F}  \ottlinebreakhack }%
\ottpremise{\forall \, \iota \, \!\!\in\! \, \ottkw{dom} \, \ottsym{(}  \ottnt{H}  \ottsym{)}  \ottsym{.}   \iota  \in  \ottkw{dom} ( \Gamma ) }%
\ottpremise{\vdash  \Gamma}%
}{
\Gamma  \vdash  \ottnt{H}}{%
{\ottdrulename{wf\_heap}}{}%
}}

\newcommand{\ottdefnwfXXHeap}[1]{\begin{ottdefnblock}[#1]{$\Gamma  \vdash  \ottnt{H}$}{}
\ottusedrule{\ottdrulewfXXheap{}}
\end{ottdefnblock}}

\newcommand{\ottdrulewfXXfields}[1]{\ottdrule[#1]{%
\ottpremise{ \ottkw{fields} \, \ottsym{(}  \ottmv{C}  \ottsym{)} \,  \equiv  \, \ottmv{f_{{\mathrm{1}}}}  \ottsym{:}  \ottnt{t_{{\mathrm{1}}}}  \ottsym{,} \, .. \, \ottsym{,}  \ottmv{f_{\ottmv{n}}}  \ottsym{:}  \ottnt{t_{\ottmv{n}}}  \ottlinebreakhack }%
\ottpremise{\Gamma  \vdash  \ottnt{v_{{\mathrm{1}}}}  \ottsym{:}  \ottnt{t_{{\mathrm{1}}}}  \ottsym{,} \, .. \, \ottsym{,}  \Gamma  \vdash  \ottnt{v_{\ottmv{n}}}  \ottsym{:}  \ottnt{t_{\ottmv{n}}}}%
}{
\Gamma  \ottsym{;}  \ottmv{C}  \vdash  \ottmv{f_{{\mathrm{1}}}}  \mapsto  \ottnt{v_{{\mathrm{1}}}}  \ottsym{,} \, .. \, \ottsym{,}  \ottmv{f_{\ottmv{n}}}  \mapsto  \ottnt{v_{\ottmv{n}}}}{%
{\ottdrulename{wf\_fields}}{}%
}}

\newcommand{\ottdefnwfXXFields}[1]{\begin{ottdefnblock}[#1]{$\Gamma  \ottsym{;}  \ottmv{C}  \vdash  \ottnt{F}$}{}
\ottusedrule{\ottdrulewfXXfields{}}
\end{ottdefnblock}}

\newcommand{\ottdrulewfXXvars}[1]{\ottdrule[#1]{%
\ottpremise{ \forall \, \ottmv{x}  \ottsym{:}  \ottnt{t} \, \!\!\in\! \, \Gamma  \ottsym{.}  \ottnt{V}  \ottsym{(}  \ottmv{x}  \ottsym{)}  \ottsym{=}  \ottnt{v}  \land  \Gamma  \vdash  \ottnt{v}  \ottsym{:}  \ottnt{t}  \ottlinebreakhack }%
\ottpremise{\forall \, \ottmv{x} \, \!\!\in\! \, \ottkw{dom} \, \ottsym{(}  \ottnt{V}  \ottsym{)}  \ottsym{.}   \ottmv{x}  \in  \ottkw{dom} ( \Gamma ) }%
\ottpremise{\vdash  \Gamma}%
}{
\Gamma  \vdash  \ottnt{V}}{%
{\ottdrulename{wf\_vars}}{}%
}}

\newcommand{\ottdefnwfXXVars}[1]{\begin{ottdefnblock}[#1]{$\Gamma  \vdash  \ottnt{V}$}{}
\ottusedrule{\ottdrulewfXXvars{}}
\end{ottdefnblock}}

\newcommand{\ottdrulewfXXtXXasync}[1]{\ottdrule[#1]{%
\ottpremise{\Gamma  \vdash  \ottnt{T_{{\mathrm{1}}}}  \ottsym{:}  \ottnt{t_{{\mathrm{1}}}}}%
\ottpremise{ \Gamma  \vdash  \ottnt{T_{{\mathrm{2}}}}  \ottsym{:}  \ottnt{t_{{\mathrm{2}}}}  \ottlinebreakhack }%
\ottpremise{\Gamma  \vdash  \ottnt{e}  \ottsym{:}  \ottnt{t}}%
}{
\Gamma  \vdash   \ottnt{T_{{\mathrm{1}}}}  \mathop{||}  \ottnt{T_{{\mathrm{2}}}}  \rhd  \ottnt{e}   \ottsym{:}  \ottnt{t}}{%
{\ottdrulename{wf\_t\_async}}{}%
}}

\newcommand{\ottdrulewfXXtXXthread}[1]{\ottdrule[#1]{%
\ottpremise{\Gamma  \vdash  \ottnt{e}  \ottsym{:}  \ottnt{t}}%
}{
\Gamma  \vdash  \ottsym{(}  \mathcal{L}  \ottsym{,}  \ottnt{e}  \ottsym{)}  \ottsym{:}  \ottnt{t}}{%
{\ottdrulename{wf\_t\_thread}}{}%
}}

\newcommand{\ottdrulewfXXtXXexn}[1]{\ottdrule[#1]{%
\ottpremise{\vdash  \ottnt{t}}%
\ottpremise{\vdash  \Gamma}%
}{
\Gamma  \vdash  \textbf{EXN}  \ottsym{:}  \ottnt{t}}{%
{\ottdrulename{wf\_t\_exn}}{}%
}}

\newcommand{\ottdefnwfXXThreads}[1]{\begin{ottdefnblock}[#1]{$\Gamma  \vdash  \ottnt{T}  \ottsym{:}  \ottnt{t}$}{}
\ottusedrule{\ottdrulewfXXtXXasync{}}
\ottusedrule{\ottdrulewfXXtXXthread{}}
\ottusedrule{\ottdrulewfXXtXXexn{}}
\end{ottdefnblock}}

\newcommand{\ottdrulewfXXlXXasync}[1]{\ottdrule[#1]{%
\ottpremise{ \ottkw{heldLocks} \, \ottsym{(}  \ottnt{T_{{\mathrm{1}}}}  \ottsym{)}  \cap  \ottkw{heldLocks} \, \ottsym{(}  \ottnt{T_{{\mathrm{2}}}}  \ottsym{)} \,  \equiv  \,  \emptyset   \ottlinebreakhack }%
\ottpremise{ \forall \, \iota \, \!\!\in\! \, \ottkw{locks} \, \ottsym{(}  \ottnt{e}  \ottsym{)}  \ottsym{.}  \iota \, \!\!\in\! \, \ottkw{heldLocks} \, \ottsym{(}  \ottnt{T_{{\mathrm{1}}}}  \ottsym{)}  \ottlinebreakhack }%
\ottpremise{\ottsym{distinctLocks(}  \ottnt{e}  \ottsym{)}}%
\ottpremise{\ottnt{H}  \vdash_{\mbox{\tiny lock} }  \ottnt{T_{{\mathrm{1}}}}}%
\ottpremise{\ottnt{H}  \vdash_{\mbox{\tiny lock} }  \ottnt{T_{{\mathrm{2}}}}}%
}{
\ottnt{H}  \vdash_{\mbox{\tiny lock} }   \ottnt{T_{{\mathrm{1}}}}  \mathop{||}  \ottnt{T_{{\mathrm{2}}}}  \rhd  \ottnt{e} }{%
{\ottdrulename{wf\_l\_async}}{}%
}}

\newcommand{\ottdrulewfXXlXXthread}[1]{\ottdrule[#1]{%
\ottpremise{ \forall \, \iota \, \!\!\in\! \, \mathcal{L}  \ottsym{.}  \ottnt{H}  \ottsym{(}  \iota  \ottsym{)}  \ottsym{=}   ( \ottmv{C} ,  \ottnt{F} ,  \ottkw{locked} )   \ottlinebreakhack }%
\ottpremise{\ottsym{distinctLocks(}  \ottnt{e}  \ottsym{)}}%
\ottpremise{\forall \, \iota \, \!\!\in\! \, \ottkw{locks} \, \ottsym{(}  \ottnt{e}  \ottsym{)}  \ottsym{.}  \iota \, \!\!\in\! \, \mathcal{L}}%
}{
\ottnt{H}  \vdash_{\mbox{\tiny lock} }  \ottsym{(}  \mathcal{L}  \ottsym{,}  \ottnt{e}  \ottsym{)}}{%
{\ottdrulename{wf\_l\_thread}}{}%
}}

\newcommand{\ottdrulewfXXlXXexn}[1]{\ottdrule[#1]{%
}{
\ottnt{H}  \vdash_{\mbox{\tiny lock} }  \textbf{EXN}}{%
{\ottdrulename{wf\_l\_exn}}{}%
}}

\newcommand{\ottdefnwfXXLocking}[1]{\begin{ottdefnblock}[#1]{$\ottnt{H}  \vdash_{\mbox{\tiny lock} }  \ottnt{T}$}{}
\ottusedrule{\ottdrulewfXXlXXasync{}}
\ottusedrule{\ottdrulewfXXlXXthread{}}
\ottusedrule{\ottdrulewfXXlXXexn{}}
\end{ottdefnblock}}

\newcommand{\ottdrulewfXXcfg}[1]{\ottdrule[#1]{%
\ottpremise{\Gamma  \vdash  \ottnt{H}}%
\ottpremise{ \Gamma  \vdash  \ottnt{V}  \ottlinebreakhack }%
\ottpremise{\Gamma  \vdash  \ottnt{T}  \ottsym{:}  \ottnt{t}}%
\ottpremise{\ottnt{H}  \vdash_{\mbox{\tiny lock} }  \ottnt{T}}%
}{
\Gamma  \vdash  \langle  \ottnt{H}  \ottsym{;}  \ottnt{V}  \ottsym{;}  \ottnt{T}  \rangle  \ottsym{:}  \ottnt{t}}{%
{\ottdrulename{wf\_cfg}}{}%
}}

\newcommand{\ottdefnwfXXCfg}[1]{\begin{ottdefnblock}[#1]{$\Gamma  \vdash  \ottnt{cfg}  \ottsym{:}  \ottnt{t}$}{}
\ottusedrule{\ottdrulewfXXcfg{}}
\end{ottdefnblock}}

\newcommand{\ottdrulewfXXprogram}[1]{\ottdrule[#1]{%
\ottpremise{\forall \, \ottnt{Id} \, \!\!\in\! \, \ottnt{Ids}  \ottsym{.}  \vdash  \ottnt{Id}}%
\ottpremise{\forall \, \ottnt{Cd} \, \!\!\in\! \, \ottnt{Cds}  \ottsym{.}  \vdash  \ottnt{Cd}}%
\ottpremise{ \epsilon   \vdash  \ottnt{e}  \ottsym{:}  \ottnt{t}}%
}{
\vdash  \ottnt{Ids} \, \ottnt{Cds} \, \ottnt{e}  \ottsym{:}  \ottnt{t}}{%
{\ottdrulename{wf\_program}}{}%
}}

\newcommand{\ottdefnwfXXProgram}[1]{\begin{ottdefnblock}[#1]{$\vdash  \ottnt{P}  \ottsym{:}  \ottnt{t}$}{}
\ottusedrule{\ottdrulewfXXprogram{}}
\end{ottdefnblock}}

\newcommand{\ottdrulewfXXinterface}[1]{\ottdrule[#1]{%
\ottpremise{\forall \, \ottmv{m}  \ottsym{(}  \ottmv{x}  \ottsym{:}  \ottnt{t}  \ottsym{)}  \ottsym{:}  \ottnt{t'} \, \!\!\in\! \, \ottnt{Msigs}  \ottsym{.}  \vdash  \ottnt{t}  \land  \vdash  \ottnt{t'}}%
}{
\vdash  \kwool{interface} \, \ottmv{I}  \,\startblock\,  \ottnt{Msigs}  \,\finishblock\,}{%
{\ottdrulename{wf\_interface}}{}%
}}

\newcommand{\ottdrulewfXXclass}[1]{\ottdrule[#1]{%
\ottpremise{ \forall \, \ottmv{m}  \ottsym{(}  \ottmv{x}  \ottsym{:}  \ottnt{t}  \ottsym{)}  \ottsym{:}  \ottnt{t'} \, \!\!\in\! \, \ottkw{msigs} \, \ottsym{(}  \ottmv{I}  \ottsym{)}  \ottsym{.}  \kwool{def} \, \ottmv{m}  \ottsym{(}  \ottmv{x}  \ottsym{:}  \ottnt{t}  \ottsym{)}  \ottsym{:}  \ottnt{t'}  \,\startblock\,  \ottnt{e}  \,\finishblock\, \, \!\!\in\! \, \mathit{Mds}  \ottlinebreakhack }%
\ottpremise{\forall \, \ottnt{Fd} \, \!\!\in\! \, \ottnt{Fds}  \ottsym{.}  \vdash  \ottnt{Fd}}%
\ottpremise{\forall \, \ottnt{Md} \, \!\!\in\! \, \mathit{Mds}  \ottsym{.}  \kwool{this}  \ottsym{:}  \ottmv{C}  \vdash  \ottnt{Md}}%
}{
\vdash  \kwool{class} \, \ottmv{C} \, \kwool{implements} \, \ottmv{I}  \,\startblock\,  \ottnt{Fds} \, \mathit{Mds}  \,\finishblock\,}{%
{\ottdrulename{wf\_class}}{}%
}}

\newcommand{\ottdefnwfXXClass}[1]{\begin{ottdefnblock}[#1]{$\vdash  \ottnt{Cd}$}{}
\ottusedrule{\ottdrulewfXXclass{}}
\end{ottdefnblock}}

\newcommand{\ottdrulewfXXfield}[1]{\ottdrule[#1]{%
\ottpremise{\vdash  \ottnt{t}}%
}{
\vdash  \ottmv{f}  \ottsym{:}  \ottnt{t}}{%
{\ottdrulename{wf\_field}}{}%
}}

\newcommand{\ottdefnwfXXField}[1]{\begin{ottdefnblock}[#1]{$\vdash  \ottnt{Fd}$}{}
\ottusedrule{\ottdrulewfXXfield{}}
\end{ottdefnblock}}

\newcommand{\ottdrulewfXXmethod}[1]{\ottdrule[#1]{%
\ottpremise{ \kwool{this}  \ottsym{:}  \ottmv{C}  \ottsym{,}  \ottmv{x}  \ottsym{:}  \ottnt{t}   \vdash  \ottnt{e}  \ottsym{:}  \ottnt{t'}}%
}{
\kwool{this}  \ottsym{:}  \ottmv{C}  \vdash  \kwool{def} \, \ottmv{m}  \ottsym{(}  \ottmv{x}  \ottsym{:}  \ottnt{t}  \ottsym{)}  \ottsym{:}  \ottnt{t'}  \,\startblock\,  \ottnt{e}  \,\finishblock\,}{%
{\ottdrulename{wf\_method}}{}%
}}

\newcommand{\ottdefnwfXXMethod}[1]{\begin{ottdefnblock}[#1]{$\Gamma  \vdash  \ottnt{Md}$}{}
\ottusedrule{\ottdrulewfXXmethod{}}
\end{ottdefnblock}}

\newcommand{\oolong}{OOlong\xspace}
\newcommand{\rccc}{\textsc{AtomiS}\xspace}

\newcommand{\func}[1]{\textbf{\textsf{#1}\xspace}}
\newcommand{\set}[1]{\textsf{#1}\xspace}

\newcommand{\off}[1]{}

\newcommand{\cfg}[3]{\langle #1; #2; #3\rangle}
\newcommand{\Id}{\ensuremath{\mathit{Id}}}
\newcommand{\Ids}{\ensuremath{\mathit{Ids}}}
\newcommand{\Cd}{\ensuremath{\mathit{Cd}}}
\newcommand{\Cds}{\ensuremath{\mathit{Cds}}}
\newcommand{\Msig}{\ensuremath{\mathit{Msig}}}

\newcommand{\IMsigs}{\ensuremath{\mathit{IMsigs}}}
\newcommand{\IMsig}{\ensuremath{\mathit{IMsig}}}
\newcommand{\Fd}{\ensuremath{\mathit{Fd}}}
\newcommand{\Md}{\ensuremath{\mathit{Md}}}
\newcommand{\Fds}{\ensuremath{\mathit{Fds}}}
\newcommand{\Mds}{\ensuremath{\mathit{Mds}}}
\newcommand{\Sas}{\ensuremath{\mathit{Sas}}}
\newcommand{\Sa}{\ensuremath{\mathit{Sa}}}
\newcommand{\SB}{\{}
\newcommand{\FB}{\}}

\newcommand{\vdashel}{\vdashe}
\newcommand{\vdashe}{\vdashx{\natlhs}}
\newcommand{\vdashx}[1]{\vdash_{#1}}

\newcommand{\kc}[1]{\texttt{#1}}
\newcommand{\kw}[1]{\textbf{#1}\xspace}

\newcommand{\ifk}{\kw{if}}
\newcommand{\elsek}{\kw{else}}
\newcommand{\atomicS}{\kw{atomic}}
\newcommand{\atomick}{\atom} %

\newcommand{\natomicS}{\kw{non\_atomic}}
\newcommand{\natomick}{\slashed{\atom}} %

\newcommand{\defk}{\kw{def}}
\newcommand{\invalidk}{\kw{invalid}}
\newcommand{\validk}{\kw{valid}}

\newcommand{\Type}[0]{\set{Type}}
\newcommand{\fieldN}[0]{\set{FieldName}}
\newcommand{\interfaceN}[0]{\set{InterfaceName}}
\newcommand{\varN}[0]{\set{VarName}}

\newcommand{\classN}[0]{\set{ClassName}}
\newcommand{\methodN}[0]{\set{MethodName}}
\newcommand{\FQmethodN}[0]{\set{VariantID}}
\newcommand{\FQmethodV}[0]{\set{VariantVar}}

\newcommand{\setname}{\textsf}

\newcommand{\var}{\textsf}
\newcommand{\outputeqs}[1]{\triangleright~{#1}}

\newcommand{\maine}{\texttt{mainV}}
\newcommand{\mainid}{\natomick.\kw{Unit}.\texttt{main}.\natomick.\natomick}

\newcommand{\inter}{\func{variantConstraints}\xspace}
\newcommand{\solexpr}[0]{\setname{SolExpr}}
\newcommand{\natval}[0]{\setname{AtomicityValue}}
\newcommand{\valVal}[0]{\setname{ValidityValue}}
\newcommand{\natN}[0]{\setname{AtomicityVar}}
\newcommand{\valVar}[0]{\setname{ValidityVar}}

\newcommand{\natV}[0]{\setname{Atomicity}}

\newcommand{\fqv}[2]{{#1}@{#2}}

\newcommand{\NatSys}{\setname{%
Constraint Systems}}

\newcommand{\MetSys}{\setname{Method Constraint Systems}}

\newcommand{\NCSSolution}{\nes~\setname{Solution}}

\newcommand{\nes}{\var{Acs}}
\newcommand{\mns}{\var{Vns}}

\newcommand{\sol}{\var{Sol}}
\newcommand{\mes}{\var{Mcs}}

\newcommand{\natv}[1]{\check{#1}}
\newcommand{\natx}{\natv x}
\newcommand{\nextvar}{\func{nextvar}()}
\newcommand{\natlhs}{\natv{y}}

\newcommand{\any}{\func{interfaceMsig}}

\newcommand{\variantmsigs}{\func{variantMsigs}}
\newcommand{\consistent}{\func{interfaceImpl}}

\newcommand{\vi}[1]{\func{vi#1}\xspace}
\newcommand{\vie}{\func{viE}\xspace}

\newcommand{\vifield}{\func{viField}\xspace}
\newcommand{\solve}{\func{solve}\xspace}

\newcommand{\vvMCS}{\func{vvMCS}\xspace}

\newcommand{\ootype}{\func{ootype}}
\newcommand{\oomn}{\func{oomn}}
\newcommand{\oosig}{\func{oosig}}
\newcommand{\leftSQ}{\boldsymbol{[}\kern-.20em\boldsymbol{[}\,}
\newcommand{\rightSQ}{\boldsymbol{]}\kern-.20em\boldsymbol{]}\,}
\newcommand{\encSQ}[1]{\leftSQ{#1}\,\rightSQ}
\newcommand{\leftB}{\{\kern-.22em|\,}
\newcommand{\rightB}{\,|\kern-.22em\}}

\newcommand{\keyword}[1]{\textbf{#1}\xspace}
\newcommand{\classk}{\keyword{class}}
\newcommand{\thisk}{\keyword{this}}
\newcommand{\extendsk}{\keyword{extends}}
\newcommand{\newk}{\keyword{new}}
\newcommand{\nullk}{\keyword{null}}

\newcommand{\returnk}{\keyword{return}}

\newcommand{\truek}{\keyword{true}}
\newcommand{\falsek}{\keyword{false}}

\newcommand{\letk}{\keyword{let}}
\newcommand{\ink}{\keyword{in}}

\newcommand{\fields}{\func{fields}}
\newcommand{\dom}{\func{dom}}

\newcommand{\natureof}[1]{\encSQ{#1}} %

\newcommand{\fresh}{\emph{fresh}}

\newcommand{\regularV}{\ottkw{base}}

\newcommand{\natvar}{\ottkw{natvar}}
\newcommand{\valvar}{\ottkw{valvar}}

\title{Sound Atomicity Inference for \\Data-Centric Synchronization}

\author[1]{Hervé Paulino}
\author[2]{Ana Almeida Matos}
\author[2]{Jan Cederquist}
\author[1]{Marco Giunti}
\author[2]{João Matos}
\author[1]{António Ravara}
\affil[1]{NOVA-LINCS and FCT NOVA, NOVA University Lisbon}

\affil[2]{Instituto de Telecomunicações, and IST, University of Lisbon}

\begin{document}

\maketitle

\begin{abstract}
Data-Centric Concurrency Control (DCCC) shifts the reasoning about concurrency restrictions from control structures to data declaration. 
It is a high-level declarative approach that %
abstracts away from the actual concurrency control mechanism(s) in use.
Despite its advantages, the practical use of DCCC is hindered by the fact that it may require many annotations and/or multiple implementations of the same method to cope with differently qualified parameters.
Moreover, the existing DCCC solutions do not address the use of interfaces, precluding their use in most object-oriented programs.
To overcome these limitations, in this paper we present \rccc, a new DCCC model based on a rigorously defined type-sound programming language.
Programming with \rccc requires only (\texttt{atomic})-qualifying types of parameters and return values in interface definitions, and of fields in class definitions.
From this \emph{atomicity specification}, a static analysis infers the atomicity constraints that are local to each method, considering valid only the method variants that are consistent with the specification, and performs %
code generation for all valid variants of each method.
The generated code is then the target for automatic injection of concurrency control primitives, by means of the desired automatic technique and associated atomicity and deadlock-freedom guarantees, which can be plugged-into the model's pipeline.
We present the foundations for the \rccc analysis and synthesis, with formal guarantees that the generated program is well-typed and %
that it corresponds behaviourally to the original one. The proofs are mechanised in Coq.
We also provide a Java implementation that showcases the applicability of \rccc in real-life programs.
\end{abstract}
\section{Introduction}
\label{sec:intro}

Parallelism is omnipresent in today's computer systems, from cloud infrastructures, personal computers to handheld devices. 
Concurrent programming, hence, becomes a fundamental tool for software to fully harness the processing power available in the underlying hardware. However, %
writing efficient concurrent code remains complex, and assessing its correctness difficult, resulting in execution errors caused by concurrency issues \cite{DBLP:conf/asplos/LuPSZ08,DBLP:journals/jisa/AsadollahSEH17}, %
some constituting real threats, such as the %
Therac-25~\cite{DBLP:journals/computer/LevesonT93} and the Northeast blackout~\cite{blackout} incidents.%
The \emph{shared memory} model is the \textit{de facto} approach to program concurrent applications in shared memory architectures  (message-passing %
alternatives used in languages like Erlang or Go remain underutilised).
Some modern programming languages (\eg Rust)
offer tight control of resources preventing interference statically, but inherently concurrent data structures still require dealing with accesses to shared memory~\cite{DBLP:journals/pacmpl/YanovskiDJD21}.
Expressing %
restrictions on the manipulation of shared resources usually consists of explicitly delimiting the sequences of instructions that access those resources (the critical region), and coding the necessary synchronisation actions %
to prevent interference. %
This can be done via low-level primitives---\ie hardware \textit{read–modify–write} instructions,
\textit{locks}, \textit{semaphores}, etc.---or higher level mechanisms like 
Java's \emph{synchronized} blocks, transactional memory~\cite{DBLP:conf/podc/ShavitT95}, and others~\cite{DBLP:conf/popl/EmmiFJM07,DBLP:conf/pldi/CheremCG08}.

\paragraph*{The limitations of {control-centric} approaches.}
The mechanisms described thus far are \emph{control-centric} approaches that require a distributed analysis of the code, making reasoning about correctness harder as the code grows. 
This is prone to human error and can often result in disorganised and \textit{ad hoc} code~\cite{DBLP:conf/osdi/XiongPZZM10}. A single missing or ill-placed locking operation, \emph{synchronized} block, atomic region, etc. is enough to compromise an application's behaviour.
To mitigate the impact of such methodology, relevant work has addressed correctness 
of concurrent accesses to shared resources, with focus on the absence of
deadlocks~\cite{DBLP:conf/popl/McCloskeyZGB06,lockinference,DBLP:conf/popl/EmmiFJM07,DBLP:conf/oopsla/BoyapatiLR02,DBLP:conf/sosp/EnglerA03,DBLP:conf/pldi/CheremCG08},
data-races~\cite{DBLP:conf/oopsla/BoyapatiLR02,DBLP:conf/sosp/EnglerA03}
and atomicity violations~\cite{DBLP:conf/popl/FlanaganF04,DBLP:conf/asplos/LuTQZ06,DBLP:conf/oopsla/SchneiderMSA08,DBLP:journals/toplas/AbadiBHI11}.
Despite these efforts, the study in \cite{DBLP:conf/asplos/LuPSZ08} identifies that 97\% of the non-deadlock errors stem from either atomicity or protocol violations.

Although there are  successful static analysis tools that deal with these kinds of problems, like Infer%
\footnote{\url{https://fbinfer.com/}} 
or ThreadSafe%
\footnote{\url{http://www.contemplateltd.com/threadsafe-1-2}}, 
as well as ongoing research~\cite{DBLP:conf/asplos/ZhangJL17,DBLP:conf/osdi/CuiGKNSWY18,DBLP:conf/pldi/RoemerGB20},
the generalised use of such solutions is still distant, as the associated effort and learning curve 
are not negligible. Most companies and programmers could greatly benefit from a language-based approach that would help them to reason about the code they are developing.

\paragraph*{Data-Centric Concurrency Control (DCCC)}
shifts the expression of  %
concurrency-related constraints away from control structures and into data declaration. 
A simple example are atomic types in C++~\cite{atomic-types}, whose scope includes only individual accesses to atomic-qualified types.
By centralising all concurrency control management onto data declaration, DCCC promotes local rather than distributed reasoning, a key change to achieve simpler reasoning on interference supervision.
Several proposals have adopted DCCC for either shared memory concurrent programming \cite{DBLP:conf/asplos/CezePCMT08,atomicsets-POPL2006,DBLP:journals/toplas/DolbyHMTVV12,DBLP:conf/pppj/0001DB16,rc3-sac16}, 
extending the concept to groups of resources that share consistency requirements, 
or to distributed programming~\cite{DBLP:conf/usenix/0001LCPRV14,DBLP:conf/ecoop/ZazaN16,DBLP:journals/pvldb/KraskaHAK09}. 

\emph{Atomic sets}~\cite{atomicsets-POPL2006,DBLP:journals/toplas/DolbyHMTVV12} laid the foundations for DCCC in the shared memory context  we are targeting.
Although innovative, \emph{Atomic sets} have the problem of requiring multiple keywords, with non-trivial semantics, that render the specification burdensome and even complex at times. Moreover, progress guarantees requires extra work from the programmer and 
is not guaranteed in all situations~\cite{deadlocks-atomicsets-ICSE2013}.
Resource-Centric Concurrency Control (RC$^3$)~\cite{rc3-sac16} is
a proposal that builds from a single keyword (\atomicS), presenting itself as a simpler solution that combines the concepts of \textit{atomically qualified type} and \textit{unit of work}.
Units of work comprise a method's sequence of instructions from the first to the last access to a variable of atomic type and which are to be executed atomically.
In DCCC approaches, %
the task of the programmer is to simply {atomic}-qualify the declaration of the variables that must be accessed atomically in a unit of work.
If the annotations are consistent, a compilation process %
should produce code %
free of data-races, atomicity violations, and deadlocks. 

\paragraph*{Problem.}
RC$^3$ %
is a promising solution. %
However, its practical use (and the use of type-based concurrency control solutions in general) is hindered by the fact that the code (method bodies, for instance) is not agnostic to the {atomic}-qualification of the local variables and parameters.
Consequently, the programmer has to write several versions of the same code to account for the supported combinations of {atomic}-qualified and not {atomic}-qualified types.
Consider the following example of a Java class \texttt{ArrayList} that implements interface \texttt{List}, including method \texttt{equals} that checks if the contents of the current collection are the same of another received as argument:
\pagebreak
\begin{lstlisting}[numbers=none]
class ArrayList implements List {
  ...
  
  // Returns true if 'this' and 'other' have the same elements
  Boolean equals(List other) { ... }  
}
\end{lstlisting}
In order to account for all combinations of {atomic}-qualified and not {atomic}-qualified  for \texttt{\thisk} and parameter \texttt{other}, the programmer must  provide 4 implementations of the method.
The number of combinations naturally doubles with each parameter whose atomicity must be considered (non-primitive types).
Moreover, if the method's result is a non-primitive value that may be assigned to both an atomic or a non-atomic recipient, \textit{e.g.}
\[\text{\icode{Boolean b = x.equals(y)} or 
\icode{@Atomic Boolean b = x.equals(y)}}\]
the return type's atomicity must also be considered, doubling the number of combinations.

Furthermore, providing atomic and regular class implementations of a given functionality would result in the same code, modulo the DCCC annotations.
This goes much in the vein of what is currently done (with other constructs) in many programming languages to implement sequential and concurrent data structures\footnote{A known example are Java's synchronized collections: \url{https://github.com/openjdk/jdk/blob/master/src/java.base/share/classes/java/util/Collections.java}.}.
The problem shares similarities with the duplication  of  \texttt{const} and \texttt{non-const} member functions in C++.
However, unlike such problem~\cite{10.5555/1036281,10.5555/1051335}, it cannot be efficiently solved by simply having the non-atomic version call the atomic one, 
since this would imply non-negligible synchronisation overhead for all accesses, independently of the {atomic}-qualification of the types.

In truth, from the programmer's perspective, a method's actual implementation is independent 
of the {atomic}-qualification  of \texttt{\thisk} or of the method's parameters and, hence, should be coded in way agnostic to these atomicity concerns.
Conversely, the compiler-generated code is considerably dependent of those same concerns, since they are the ones guiding the generation of the required low-level synchronisation instructions.
In the example above, concurrency control has to be generated to guarantee that no elements are concurrently added to  {atomic}-qualified lists during the execution of \texttt{equals}, otherwise the result of \texttt{equals} could be erroneous.

Therefore, the challenge is to \emph{infer and synthesise} the use of the \texttt{atomic} type qualifier in method definitions, so that they are removed from the programmer's concerns but
are available for the compiler to generate the necessary synchronisation code. 
Our proposal provides a framework to write methods in a concurrency-agnostic way (\ie without worrying about the concurrency annotations):
the programmer %
{atomic}-qualifies some class fields to implement their concurrency restrictions.
In short, what differentiates our type-based DCCC from other approaches is that the concurrency control information is encoded on the \emph{type} of the resource.
Our insight is to leverage on type parameters to abstract away the {atomic}-qualification of fields from class implementation and move the concern to class instantiation. This approach will allow to use \emph{the same} base class implementations for both concurrent and sequential settings.

\paragraph*{Our proposal.}
We build on DCCC atomicity annotations and on the \emph{unit of work} concepts of Atomic Sets and RC$^3$ to propose \rccc, a novel concurrency control model that only requires from the programmer a \emph{high-level atomicity specification} -- the {atomic}-qualification of class fields and,
in each interface method signature, the enumeration of which combinations of atomic parameters and return types must be supported by the method's implementations.

The \texttt{List} interface depicted in Listing~\ref{lst:baselist} features such specification by using annotation \icode{@All} that we will present in \S\ref{sec:language}.
From this specifications our approach:
\begin{enumerate}
\item 
infers and synthesises the {atomic}-qualifications of the type of method parameters, local variables and return values;
\item allows for the {atomic}-qualification of class field types to be abstracted into class type parameters.%
\end{enumerate}
 
\rccc decouples the reasoning about functionality from the reasoning about concurrency.
Methods are implemented in a concurrency-agnostic way and, when convenient, the \texttt{atomic}-qualification of types can also be encoded in type parameters. 
As a result, very specific optimisations aside, with \rccc it is no longer needed to code concurrent and sequential versions of classes, as it is usual in programming languages, such as Java and C++.
Listing~\ref{lst:baselist} presents a sketch of a simple implementation of a generic list in Java.
Besides the usual type parameter (\texttt{T}) for the elements' type, we add a new one (\texttt{N}) for the type of the node to be used. 
Listing~\ref{lst:list} showcases how a concurrent list and  a sequential may be derived from the base list. 
Either of these lists may be used to store collections of both atomic and non-atomic values.
For instance, a sequential list of atomic bank accounts may be created as
\icode{new SeqList<@Atomic Account>()}.

\begin{table}[t]
\hrule
\begin{lstlisting}[multicols=2, escapechar=|]
interface List<T> { |\label{line:begin:list}|
 void add(T element);
 T get(int pos); 
 @All Boolean equals(@All List<T> other); 
} 

class Node<T> { 
 T value;
 Node<T> next, prev;
}

class BaseList<T, N extends Node<T>> implements List<T> {
 N head, tail;
 Supplier<N> factory;

 BaseList(Supplier<N> factory) {
  this.factory = factory; }

 void add(T element) {
  N node = this.factory.get();
  node.value = element;
  node.next = null;
  if (this.head == null) {
   node.prev = null;
   this.head = node;
  } else {
   node.prev = this.tail;
   this.tail.next = node;
  }
  this.tail = node;
 }

 T get(int pos) {...}
 Boolean equals(List<T> other) {...}
}
\end{lstlisting}
\vspace{5pt}
\hrule
\captionof{lstlisting}{The base list. We make use of the \protect\texttt{Supplier} functional interface\protect\footnotemark{} to overcome Java's limitations on the creation of new objects from  type parameters.}
\label{lst:baselist}
\hrule
\begin{lstlisting}[morekeywords=@Atomic, escapechar=|]
class ConcurrentList<T> extends BaseList<T, @Atomic Node<T>> {
  ConcurrentList() { super(Node<T>::new); }
}
class SeqList<T> extends BaseList<T, Node<T>> {
  SeqList() { super(Node<T>::new); }
}
\end{lstlisting}
\hrule
\captionof{lstlisting}{Concurrent and sequential list class definitions.}
\label{lst:list}
\end{table}

\footnotetext{\url{https://docs.oracle.com/javase/8/docs/api/java/util/function/Supplier.html}}

\begin{figure}
    \centering
    \includegraphics[width=\linewidth]{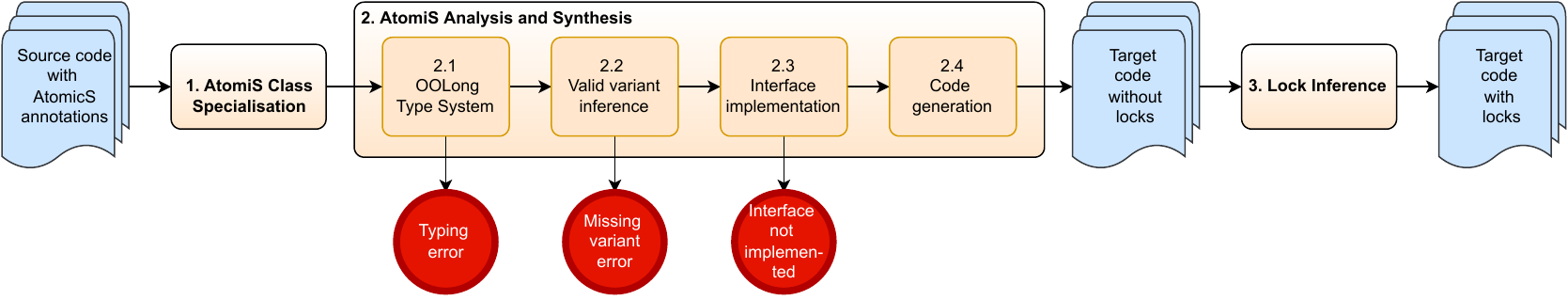}
    \caption{\rccc compilation stages}
    \label{fig:stages}
\end{figure}

The \rccc compiler generates: 
\begin{inparaenum}
\item specialisations of all classes with type parameters, one version for each combination of \{{atomic}-qualified, not {atomic}-qualified\} for each type parameter (stage \rccc \emph{Class Specialisation} in Fig.~\ref{fig:stages}),
\item the type-safe  atomicity-related overloaded versions of each method (stage \rccc \emph{Analysis} in Fig.~\ref{fig:stages}); and, from this code, 
\item the low-level concurrency control (locks) necessary to guarantee %
key properties like thread-safety (stage \emph{Lock Inference} in Fig.~\ref{fig:stages}).
\end{inparaenum}

Although we are addressing the problem in the context of \rccc, in this paper we focus mainly on the first two of the three stages depicted in Fig.~\ref{fig:stages}, where the main novelty of our work resides.
In what concerns step 3, 
\rccc is a high-level data-centric specification that 
may be encoded in virtually any base synchronisation mechanism, an advantage of data-centric synchronisation pointed out by the authors of Atomic Sets~\cite{atomicsets-POPL2006,DBLP:journals/toplas/DolbyHMTVV12}; \rccc’ units of work simply define a sequence of instructions whose execution has to be perceived as an atomic operation; such behaviour can be obtained via an encoding in many control-centric synchronisation mechanisms, such as mutual exclusion locks, read-write locks, transactional memory, and so on and so forth. 
We %
resort to the inference of either mutual exclusion or read-write locks.

Regarding step 2, we
propose an approach based on a static analysis %
with three steps (2.1 to 2.3 in Fig.~\ref{fig:stages}). It begins by considering all possible combinations of \{{atomic}-qualified, not {atomic}-qualified\} for each class type parameter, internally generating a different class for each combination.
Next, within each class, it considers all possible combinations for each method parameter, return type, and the instance, what we refer to as \textit{method variants}. 
For each variant, the analysis infers the atomicity of the local variables
and considers \textit{valid} only the variants that ensure the consistency of the atomicity type qualifications (atomic resources are \emph{never} assigned to non-atomic variables, and vice-versa) -- Step 2.2 in Fig.~\ref{fig:stages}.
Then, in Step 2.3, it checks if all \rccc-annotated interfaces are correctly implemented, when considering the valid variants inferred for each class.
Lastly, 
Step 2.4 launches a correct-by-construction compilation process that produces code comprising 
atomic versions of each type in the source code and, for each class, the set of valid method variants.
The resulting method overloading allows programmers to continue implementing methods agnostic of atomicity concerns and to be unburdened with the task of {atomic}-qualifying method parameters and local variables.

To define a provably correct analysis and inference of the atomicity annotations, we formalise a \emph{type qualifier inference system} for a well-defined, type sound, core multi-threaded O.~O.~language (\oolong~\cite{oolong-ACR2019}). 
We choose \oolong not only due to the minimal grammar and rigorously defined (static and dynamic) semantics, but also because type soundness (progress and subject reduction) has been mechanised in the Coq proof assistant%
\footnote{\url{https://coq.inria.fr/distrib/current/refman/}},
thus being a certified language to build analyses on.
Therefore, following the high-level DCCC atomicity specification approach, \rccc is a formal model that aims to provide a basis for formally guaranteeing thread-safety properties like absence of atomicity violations, data-races and deadlocks.%
\footnote{Again, to obtain formal guarantees of thread-safety is work currently under development.}

\paragraph*{Contributions.}
The contributions of this paper are thus the following:
\begin{enumerate}
\item \rccc, a new DCCC language-based approach, %
designed to be used in the programming of real-life applications. \rccc requires only writing atomicity specifications (presented in \S\ref{sec:language}), delegating to a static analysis the inference of the atomicity concerns local to each method (in \S\ref{sec:variants}). Furthermore, the code generated includes the atomic versions of the source code's types and, for all classes, the code of all valid method variants (in \S\ref{sec:valid_variants} and \S\ref{sec:code_generation}).
 \item A formalisation of \rccc in a rigorously defined, type-sound, multi-threaded, core object-oriented programming language~\oolong (in \S\ref{sec:oolong}), illustrating the use of the language in this context with a significant example (a concurrent list).

\item Soundness results supporting the inference and code generation processes.\footnote{The Coq mechanisation can be found at \emph{AtomiS-Coq proof of type preservation (2022)}, at \url{https://zenodo.org/record/6346649} and \url{https://zenodo.org/record/6382015}}
\off{the resulting artifact\footnote{\url{https://zenodo.org/record/6382015}} has been accepted by the ECOOP 2022 Artifact Evaluation PC on
May 2, 2022.}
Concretely, this work provides 
formal guarantees (in \S\ref{sec:soundness}) that: (1) the generated program is well-typed (in \S\ref{sec:ts}), with types consistent with those in the original program, being the proof of type preservation mechanised in Coq (\S\ref{sec:mech}); and  
(2) the generated code behaviourally corresponds to the original one, \ie it does not do more nor less than the original (in \S\ref{sec:bis}).
\end{enumerate}

A companion paper presents this work from a practical perspective, focusing on a prototype Java implementation that showcases the applicability of AtomiS to real code~\cite{atomis-oopsla23}.

\off{
\begin{lstlisting}[language=java]
interface Collection {

@RC3(list of qualifiers) void add(Object element);


}
class Node { ... }
class CL implements Collection {
  @Atomic Node head;
  @Atomic Node tail;
  
  void add( ...) { ...}
}
\end{lstlisting}
}

\off{
interface I {
    m(t1) : t2 ---> (non-atomic, non-atomic)
}

class C implements I {
    m()  --- m_nn, m_aa, m_an
    m2()
}

class D implements I {
    m() -- m_nn,  m_na, m_an
}

let k = new E in
let x = new D in k.n((I) x) 

class E {
    a: atomic W 
    n(I x) {
        
        let y = x.m(a) in e --  m_aa (FAILS)
    }
    
    n(C x) {
        let y = x.m(a) in ....  m_aa (PASS)
    }

}}

\section{The \texorpdfstring{\rccc}{RC3} Model}
\label{sec:language}

\rccc is a generic data-centric synchronization model applicable to any concurrent language with shared state.
The data-centrality comes from the addition of the \atomicS type qualifier
to
variables whose values are shared across multiple concurrent execution flows.
Our Java implementation does this  by defining annotation
\icode{@Atomic}, applicable to any type use.

Mutable values assigned to variables with atomic type (\textit{atomic variables}) are referred to as \textit{atomic values}.
A process' memory is thus partitioned into \textit{atomic} and \textit{regular} (non-atomic) values.
Atomic values may only be manipulated within the scope of a \textit{unit of work} -- a sequence of instructions whose execution must be perceived as an atomic operation to the remainder computation~\cite{DBLP:journals/toplas/DolbyHMTVV12}, 
which in \rccc translates to blocks of instructions (such as method bodies)
that access at least one atomic value.
So, the execution of units of work that access the same atomic values may \emph{not} overlap in time (must occur in mutual exclusion).
The goal of \rccc is to ensure that \emph{all accesses to atomic values within each unit of work are, in effect, a single atomic operation}.

For such purpose, concurrency restrictions may be specified in both interface and class definitions.
In the case of interfaces, each method specifies what combinations of atomicities are supported in its parameter and return types.
To that end, we define the \icode{@AtomiSSpec} annotation, which takes a string with specifications of the form \[\small \texttt{(\textit{atomicity\_of\_parameter$_1$}, ...,
\textit{atomicity\_of\_parameter$_n$}) -> atomicity\_of\_return}\] where the values may be \texttt{atomic}, 
\texttt{non\_atomic} or \_ (to denote \textit{not relevant}).

Consider the following interface for lists of atomic objects. We want to specify that the methods support the insertion and retrieval of atomic values.
\begin{lstlisting}[style=atomis,numbers=none,basicstyle=\ttfamily\scriptsize]
interface ListAtomic {
 @AtomiSSpec("(atomic) -> _") void add(Object element);
 @AtomiSSpec("(_) -> atomic") Object get(int pos);  
 @AtomiSSpec("(non_atomic) -> _, (atomic) -> _") boolean containsAll(ListAtomic other);
} 
\end{lstlisting}

Accordingly,  method \texttt{add}  receives an atomic parameter, 
method \texttt{get} receives the position of the value to retrieve (a value of an immutable primitive type that requires no concurrency control) and 
returns a value of atomic type and
lastly,  as mentioned in \S1, in  method \texttt{containsAll} 
we want to be able to compare the contents of the current list with any other, atomic or not.

Note the lack of atomicity information about the object providing the methods.
Any implementation of a given interface must provide implementations for all
supported atomicity combinations of all methods, regardless of the object's atomicity (see \S\ref{sec:atomic_types_and_method_variants}).

Atomicity annotation can be made easier with the introduction of some syntactic sugar.
Parameter and return types can be left \textit{bare} or annotated with \icode{@Atomic} or \icode{@All}. These three possibilities signify that the supported values are, respectively, \texttt{non\_atomic}, \texttt{atomic} or \textit{both}.
The \texttt{ListOfAtomic} interface can now be rewritten as presented in Listing~\ref{lst:conclist}.

\begin{table}
\hrule
\begin{minipage}{.54\textwidth}
\begin{lstlisting}[escapechar=|, basicstyle=\ttfamily\scriptsize]
interface ListAtomic {
 void add(@Atomic Object element);
 @Atomic Object get(int pos);  
 boolean containsAll(@All ListAtomic other);
} 

class NodeOfAtomic {
 @Atomic Object value;
 NodeOfAtomic next, prev;
}
\end{lstlisting}
\end{minipage}
\begin{minipage}{.45\textwidth}
\begin{lstlisting}[firstnumber=last, escapechar=|,basicstyle=\ttfamily\scriptsize]
class ConcurrentListOfAtomic implements ListOfAtomic {
 @Atomic NodeOfAtomic head, tail;
 
 void add(Object element) { ... }
 ...
}
\end{lstlisting}
\end{minipage}
\hrule
\captionof{lstlisting}{A concurrent list example. \texttt{add} is the same as in Listing~\ref{lst:baselist} with \texttt{N} replaced by \texttt{NodeOfAtomic}}
\label{lst:conclist}
\end{table}

To express relations between the atomicity of certain parameter/return values, we introduce the \icode{@AtomicityOf} type annotation to express that a type must have the same atomicity of a given variable or type.
For example, an identity function that must return a value with the same atomicity as its parameter can be written as \icode{@AtomicityOf("x") Object identity(Object x)}.

 Regarding class implementations, 
 class fields may be accessed by multiple methods which may, in turn, be executed by multiple threads. 
 We thus require field atomicity to be specified in the source code.
The implementation of \icode{ListOfAtomic} in Listing~\ref{lst:conclist} supports concurrent access to list nodes, so these are qualified as atomic.
At this point we are still not dealing with type parameters, so nodes are implemented as instances of class \icode{NodeOfAtomic}, which stores atomic objects.

Arrays specify the atomicity of both the array and its contents, so \icode{@Atomic} may be used on the type parameter and the array itself. An array of \texttt{T} may thus be declared in one of four possible combinations: \icode{T[]}, \icode{@Atomic T[]} (a non-atomic array of atomic \texttt{T}), \icode{T @Atomic []} ( an atomic array of non-atomic \texttt{T}), or \icode{@Atomic T @Atomic []} (an atomic array of atomic \texttt{T}).

\subsection{Units of Work} %
\label{sec:units_of_work}

Recall that
we consider a \emph{unit of work} to be a sequence of instructions (of a method) that performs an access (read or write) to an atomic value and thus must be perceived as an atomic operation to the remainder computation.
So, we consider \emph{units of work} for every method accessing atomic values, comprising the method's code from the first to the last access to an atomic value.
More precisely, the boundaries of a unit of work begin at the first instruction~$I_1$ in the method's body, not within calls, that accesses an atomic value, and end at the last instruction~$I_2$ in the method's body, not within calls, that accesses an atomic value. Note that the unit of work includes all accesses to atomic values done in the context of calls between the execution of~$I_1$ and~$I_2$.

The concept bears similarities with transactional memory~\cite{DBLP:conf/podc/ShavitT95}, as units of work may be interpreted as transactions over atomic values, and transactional memory may even be used to implement units of work.
The main difference lies in the fact that units of work arise from the annotation of data, rather than from the explicit delimitation of transactions.
Also, in \rccc, atomicity means that the operation is perceived to be indivisible (no intermediate state is visible), rather than to be executed in its entirety or not at all. We do not have a notion of commit and rollback, as exists in transactions. In our context, atomicity is a concept closer to isolation in ACID transactional systems.

Consider class \icode{ConcurrentList} (Listing~\ref{lst:baselist}).
Given that type parameter \icode{N} was instantiated with \icode{@Atomic Node<T>}, we have that method \icode{add} features a unit of work that begins in the comparison \icode{this.head == null}, performed in line 23, and ends at the assignment \icode{this.tail = node}, in line~30. %
The same happens with methods  \icode{get} and \icode{containsAll}, that also read fields %
\icode{head} and \icode{next}.
So, the executable code of these methods will have to feature concurrency control code to ensure that
the execution of their units of work does not overlap in time, when accessing the same memory positions.
Figure~\ref{fig:units_of_work} showcases the trace of the execution of multiple methods over multiple lists. 
A unit of work may only proceed its execution if it retains exclusive access to all the resources it needs to write and read access to all the values it needs to read.

\begin{figure}
    \centering \includegraphics[width=.85\linewidth]{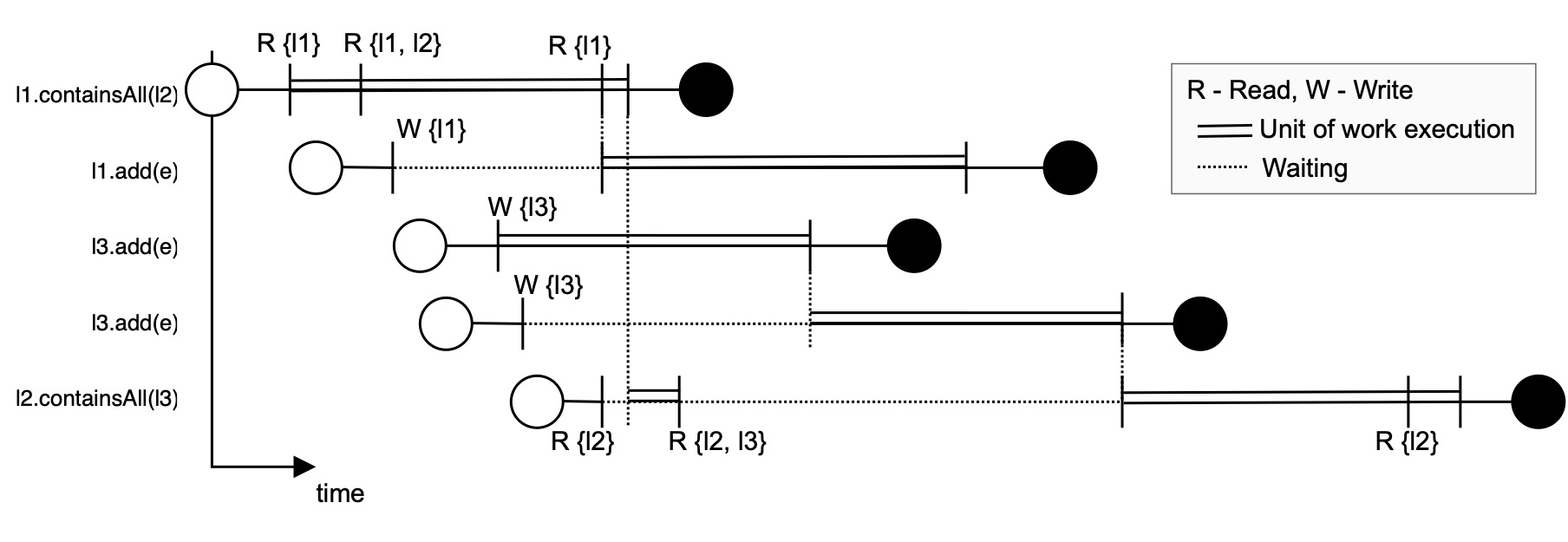} 
    \caption{Execution trace of multiple units of work working upon the state of 3 lists: $l1$, $l2$ and $l3$}
    \label{fig:units_of_work}
\end{figure}

\medskip
\textbf{Composition of units of work.}
The composition of units of work must be considered when:
\begin{enumerate}
    \item a method does not access an atomic value but calls two or more methods that do. An example is a bank transfer operation that calls existing deposit and withdraw operations that manipulate an atomic account balance. The latter operations include a unit of work, but transfer does not, since its body does not access any atomic value.
    \item a method includes a unit of work and calls methods that access atomic values before the start or after the end of its unit of work.
\end{enumerate}

To deal with unit of work composition, the compilation process must receive a \emph{unit of work composition policy}. Currently, two are supported: \emph{conservative} and \emph{standard}.
The first handles the composition automatically, creating a new unit of work for the composition of all units of work.  For case (1), it creates a unit of work for the method performing the calls and, for case~(2), it extends the unit of work to include the first and the last methods that access atomic values.
This approach is thus correct by design but may lead to extreme cases of over synchronization.
To handle such scenarios, composition may be avoided for individual methods with the annotation \icode{@NoUnitOfWork}.
In turn, the \emph{standard} policy instructs the compiler to generate new units of work %
only when all the calls made by the current method, outside %
its own unit of work, do not invoke other methods that also comprise units of work.
This excludes typical implementations of \textit{getters} and \textit{setters}, \icode{equals} and other similar methods.
By not being exhaustive, this policy might not address composition for all scenarios sought by the programmer.
To specifically determine which additional methods require composition handling, the annotation \icode{@UnitOfWork} must be used.

These two new annotations are code- rather than data-centric. 
This derives from the fact that they operate upon units of work rather than data and should be 
placed in the code as reactions to the output given by the compiler.
In our opinion, moving the annotations to data declaration would be harder to reason about because, for a given atomic value,  composition is desirable in some methods and not in others. To cope with such flexibility would require the inclusion of the lists of the methods to process (or equivalent) in the data-centric annotations. As the composition of units of work  may involve multiple atomic values, this approach could easily lead to complex and possibly inconsistent specifications.

In sum, synchronization in \rccc remains  data-centric but, when aiming for better performance, the unit-of-work-directed annotations allow for a fine grained control over which methods should feature units of work.
Given their availability, we permit the utilization of the annotations in any method, not only to address unit of work composition.

\subsection{Atomic Types and Method Variants}
\label{sec:atomic_types_and_method_variants}

Atomic-qualified (\textit{atomic}, for short) and \regular types define two different type trees, and thus are not convertible into each other.
If a class $C$ implements an interface $I$, then the atomic type \icode{@Atomic C} (that we  denote by $\atomicS~C$) will implement interface
 $\atomicS~I$. 
Equivalently, if a class $C$ extends a class $D$
then $\atomicS~C$ will extend $\atomicS~D$. 
As a result, an object with atomic type cannot be cast into a \regular type, and, by this way, give rise to uncontrollable atomicity violations. %

Concerning {atomic} class types, a decision must be made about the atomicity of unqualified fields of the same type of the hosting class, 
such as field \texttt{next} from class \texttt{NodeOfAtomic} in Listing~\ref{lst:conclist}.
The  atomicity of such fields may be inherited from the class itself,  preserving the type equality between the field and the hosting class, or simply remain unaltered, breaking this equality.
Currently, we chose to preserve the equality, causing the  atomic class type to be  slightly different from the one of its \regular  counterpart. 
Ergo, in Listing~\ref{lst:conclist},  
the type of \texttt{head.next} is  \atomicS \texttt{NodeOfAtomic}, allowing for the \texttt{NodeOfAtomic} type to be used both in 
the implementation of both concurrent and regular lists of atomic objects.
As such behavior may not always be desired, %
one may consider inferring the atomicity of these fields. %

Contrarily to interfaces, the atomicity of method parameters and return types is not explicitly conveyed in method definitions. 
This allows for the use of the same  definition in multiple atomicity contexts:
the same method name defines different methods, with different signatures and behaviors, depending of the atomic nature of its parameters.
Thus, the atomicity of parameters types must be matched, for each method call, with the atomicity of the types of the values assigned to them.
The atomicity of the types of local variables and of the return type of a method may depend on the atomicity of the types of the values that have been (or will be) assigned to a field or a parameter.
A method may hence have several different signatures, what we refer to as \textit{methods variants}, corresponding to combinations of
atomicity qualifiers
 for the object itself
(\thisk, because, for a class $D$, instances of both  $D$ and $\atomicS\ D$ may be created), 
the method's parameters and return type.
For each variant, the atomicity of  all non-qualified local variables is inferred during the  compilation process, having as base the atomicity of fields and the given atomicity combination (more details in \S\ref{sec:variants}).
Accordingly, some variants may not be type-safe.
For instance, the  validity of the variants of
method \texttt{add} from Listing~\ref{lst:conclist} depends on the {atomic}-qualification of the \texttt{element} parameter.
If the parameter is \regular it cannot be
assigned to field \texttt{value} that is of an atomic type.
 It is thus the parameter that defines the validity of the variant in this particular case. 
This is not always the case:
in method \icode{containsAll}, all variants are valid because the method does not impose any atomicity constraints on either \thisk, the parameter and the return value.

\subsection{Leveraging Type Parameters}
\label{sec:type-parameters}
Method variants provide the framework for writing atomicity-agnostic methods. To take the extra step and remove most (if not all) \rccc annotations from class implementations,
we leverage on \emph{type parameters}.
The insight is to encode the atomicity of the type as an attribute of the type parameter.
Looking again at the \icode{BaseList} class of Listing~\ref{lst:baselist},
the atomicity of type parameter \icode{N} will determine if the list is of atomic or \regular nodes. In turn, the atomicity of type parameter \texttt{T} will determine if the list stores atomic or \regular values of type \icode{T}.

The Java implementation %
 supports the {atomic}-qualification of  type parameters  at type definition (\icode{class C<@Atomic T> \{...\}}), as well as at type usage, such as on subtyping, variable declaration or instance creation (\icode{new C<@Atomic Integer>(...)}). 
For a given object, each of its type parameters is atomic if it is {atomic}-qualified at either of the alternatives. As with the atomicity of the class types themselves, the  atomic-qualification of type parameters only impacts directly on the atomicity of  fields. The one of the parameters, return values and local variables is inferred as before.

There are cases where atomicity encompasses several fields of unrelated types. For example, in an atomic data structure based on an array, the array, as well as the accounting of the number of elements in the array, must be atomic.
To avoid having to parameterize the class in several type parameters that simply convey the same atomicity information (or to have type parameters of ``kind'' atomicity), we, once again, make use of the \icode{@AtomicityOf} annotation
to express that a field must have the same atomicity of another field or type parameter. 
In the next example:
\begin{lstlisting}[basicstyle=\ttfamily\scriptsize]
class ArrayList<T, A> {
   @AtomicityOf("A") int count;
   T @AtomicityOf("A") [] array;
   ...
}  
\end{lstlisting}
fields \icode{count} and \icode{array}
have the atomicity of  type parameter \icode{A}, whose existence has the sole purpose of defining the atomicity of such fields.  
Java does not support  annotations as type arguments.
Accordingly, 
to create an instance of the class to store strings in an atomic array, one writes\linebreak \icode{new ArrayList<String, @Atomic Object>()} or simply  
\icode{new ArrayList<String, Atomic>()} without 
the \icode{@} prefix because we want to denote the type of the \icode{Atomic} annotation rather than the annotation itself. 
To create an instance of the class to store strings in an non-atomic array, we may use any (non-atomic-qualified) type other than \icode{Atomic}. 
For the sake of code readability we defined the \icode{NotAtomic} type: 
 \icode{new ArrayList<String, NotAtomic>()}.

\section{Programming with \rccc}
\label{sec:programming}

Programming with \rccc requires a specification that includes annotating interfaces and class fields.
We discuss herein the impact of such annotations. 

\textbf{Interface Annotation.}
The annotation of interfaces
augments the original specification 
with new method signatures. %
By indicating the combinations of the types'
atomicities, we only generate signatures that
have to be fulfilled by the classes that implement the interface.

The use of type parameters abstracts away most of the \icode{@Atomic} annotations, but the programmer still has to place the \icode{@All} and \icode{@AtomicityOf} annotations whenever, respectively, both atomic and \regular types are to be supported in a parameter or return type, and/or there is a need to relate the atomicities of multiple types in a signature.

\begin{table}
\centering
\caption{Valid variants and units of work resulting from different placements of  annotation \icode{@Atomic} in the subtyping of \icode{BaseList<T, N>} in Listing~\ref{lst:baselist}. For the sake of conciseness, valid variants are represented by tuples of four elements (method\_identifier, atomicity qualifiers of \thisk, atomicity qualifiers of  the parameter, atomicity qualifiers of the return value), where the atomicity qualifiers are represented by \texttt{a}, for \atomicS, and \texttt{n}, for \natomicS.}
\label{tab:atomic}
\vspace{-8pt}
\resizebox{\linewidth}{!}{
\begin{tabular}{ | l | l | p{4.8cm} | p{4cm} | }
\hline
 Type & Variables with atomic type & Valid variants  & Units of work \\
 \hline
 \hline
 \icode{SeqList<T>} 
 & 
 	--  
&  
  \multirow{2}{=}{(\icode{add}, n, n, \{a, n\}),   (\icode{get}, n, \{a, n\}, n) and 
   (\icode{containsAll}, n, \{a, n\},  \{a, n\}) }
&    
    \icode{containsAll} when \icode{other} is atomic or has atomic fields \\
     \cline{1-2}   \cline{4-4}     
  \icode{ConcurrentList<T>}  & 
		\icode{head} and \icode{tail}   
  &  &  
 \icode{add}, \icode{get} and  \icode{containsAll} \\
 \hline
  \icode{SeqList<@Atomic T>} 
  &
  \icode{Node.value}
   & 
   \multirow{2}{=}{(\icode{add}, n, a, \{a, n\}),   (\icode{get}, n, \{a, n\}, a) and 
   (\icode{containsAll}, n, \{a, n\},  \{a, n\}) }
   & 
       \icode{containsAll} when \icode{other} is atomic or has atomic fields    \\
      \cline{1-2}     \cline{4-4} 
    \icode{ConcurrentList<@Atomic T>}  
    &
     \icode{head}, \icode{tail} and \icode{Node.value} & &   \icode{add}, \icode{get} and  \icode{containsAll}   \\
 \hline 
\icode{@Atomic SeqList<T>} 
&
instance of the class (\thisk)
&
\multirow{2}{=}{(\icode{add}, a, n, \{a, n\}),   (\icode{get}, a, \{a, n\}, n) and 
   (\icode{containsAll}, a, \{a, n\},  \{a, n\}) }
&
\multirow{4}{=}{\icode{add}, \icode{get} and \icode{containsAll}} 
\\
      \cline{1-2}     
\icode{@Atomic ConcurrentList<T>} 
&
\thisk, \icode{head} and \icode{tail}  
&
&
\\
\cline{1-3}   
\icode{@Atomic SeqList<@Atomic T>}
&
\thisk and \icode{Node.value}
&
\multirow{2}{=}{(\icode{add}, a, a, \{a, n\}),   (\icode{get}, a, \{a, n\}, a) and 
   (\icode{containsAll}, a, \{a, n\},  \{a, n\}) }
&
\\
      \cline{1-2}     
\icode{@Atomic ConcurrentList<@Atomic T>} 
&
\thisk, \icode{head}, \icode{tail} 
&
&
\\
& and \icode{Node.value}  & & \\
 \hline 
\end{tabular}
}
\end{table}

\textbf{Class Implementation.}
Programming classes with \rccc
requires reasoning about which  memory objects are prone to data races
and {atomic}-qualify the ones on  which %
all accesses must be done within units of work.
This is usually done by qualifying fields but,  may also require qualifying local variables whenever these are shared by multiple threads, as may happen in Java.

The placement of the \icode{@Atomic} annotation 
allows for a more coarse or fine grained concurrency control.
In  our running list example, the \icode{@Atomic} annotation
of fields \icode{head} and \icode{tail} 
depends on the annotation of the type parameters.
The use of type parameters renders the code \rccc-annotation-free, but that does not mean that the programmer does not have to reason about which fields need to be atomic in a concurrent setting. 
It just means that the
concurrency control may be injected through subtyping, as demonstrated in Listing~\ref{lst:conclist}. The \icode{@Atomic} field annotations are replaced by the use of the type parameter.

Table~\ref{tab:atomic} showcases the impact of  \rccc-related  subtyping   by using the derived \icode{SeqList}
and \icode{ConcurrentList} types of Listing~\ref{lst:conclist}.
The use of type \icode{ConcurrentList} (cf. second row), that
{atomic}-qualifies fields \icode{head} and \icode{tail},
ensures that the list is one of atomic nodes;   
so, every access to these fields must 
ensure the atomicity of the list's operations.
Therefore, units of work emerge in methods \icode{add}, \icode{get} and \icode{containsAll}.
In the latter, the enclosed unit of work also encompasses the accesses to the atomic fields of the \icode{other} list.
The qualification of the fields has no impact on the variant validity. 
Note, however, that qualifying just one of the fields would preclude the existence of valid variants, since there would not exist a combination of atomicities that would guarantee that the type of \icode{head} and \icode{tail} are the same.
  
The {atomic}-qualification of \icode{Node}'s \icode{value} field, by qualifying the class' type parameter (\icode{T}) (cf. third and fourth rows), 
makes the list one that stores objects of atomic type.
Comparing with the remainder alternatives, we observe that the qualification has no impact on the units of work.
This happens  because the methods do not access the state of the objects they store nor modify the contents of field \icode{value}. 
The units of work will appear in methods
that access elements stored in the list.
In the table, the set of  valid variants of 
the classes with \icode{@Atomic T} is naturally adjusted to receive objects of atomic type.

Instead of {atomic}-qualifying fields of a class ($C$)'s implementation, one can simply create an object of type \atomicS $C$ (cf. rows 5 to 8).
Applying this approach to the list implementations produces types on which every access to the list falls within a unit of work.
The resulting units of work are the same as \icode{ConcurrentList} but their boundaries include the whole code of each method variant rather than
only the list of instructions that encloses all accesses to atomic values. 
Regarding variant validity, the impact is confined to the restriction of the valid variants to the ones  on which \thisk is atomic, instead of being of a \regular type.
The impact would be larger if any of the methods related one of the their parameters with a field of the class' type, as the atomicity of the parameter would have to match the one of the class.

\section{Static Analysis and Synthesis}
\label{sec:variants}
In this section, we describe the static analysis performed by \rccc to infer all atomic-qualifiers of the program's types, and to consistently synthesise the code for the valid variants of every method. 
To formally prove the soundness of the
analysis
we build on \oolong~\cite{oolong-ACR2019}, an object-oriented multi-threaded programming calculus. It is a principled language 
with minimal syntax, a formal static and dynamic semantics, and a computer-aided proof of type soundness (progress and type preservation for well-typed programs) developed in Coq and publicly available (it is thus re-usable, and the mechanisation of our results in Coq build on it).
Concretely, we developed a variant of \oolong, that we refer to as \rccc-\oolong, that is a simplified version of \rccc-Java.
Crucially, the operational semantics of the language is exactly the same of \oolong, and we thus omit its presentation here.%
\footnote{The interested reader can check Section 4, Figures 8 and 9, of \oolong main reference~\cite{oolong-ACR2019}.}

\subsection{\rccc-\oolong}
\label{sec:oolong}

\begin{table}[tp]
\caption{Syntax of \rccc-\oolong programs.  %
}
\label{tab:syntaxpos}
\centering
\resizebox{\linewidth}{!}{
\begin{tabular}{lcl}
  \hline
  $P$ & $::=$ & \Ids~\Cds~$e$ \hfill \textit{(Programs)}\\
  \Id & $::=$ & \kw{interface} $I$ \SB \IMsigs \FB  $~|$ \kw{interface} $I$ \kw{extends} $I_1, I_2$ \hfill \textit{(Interfaces)} \\
     \IMsig & $::=$ & \Msig $~|~$ 
     \Msig\ [\Sas] \hfill \textit{(Interface Signatures)} \\
     \Sa  & $::=$ &  $\atan{q_1}{q_2}$ \hfill  \textit{(Signature Annotations)}  \\
  \Cd & $::=$ & \kw{class} $C$ \kw{implements} $I$ \SB \Fds~\Mds \FB \hfill \textit{(Classes)}\\
  \Fd & $::=$ & $f : t$ $~|~$ $f :\ $\atomicS$ t $ \hfill \textit{(Fields)} \\
  \Md & $::=$ & \kw{def} \Msig \SB$e$\FB \hfill \textit{(Methods)} \\
    \Msig & $::=$ & $m(x : t_1) : t_2$ \hfill \textit{(Signatures)} \\
  $e$ & $::=$ & $v$
        $~|~$ $x$
        $~|~$ $x.f$
        $~|~$ $x.f$ \kc{=} $e$
        $~|~$ $x.m(e)$ 
        $~|~$  \kw{let} $x$ \kc{=} $e_1$ \kw{in} $e_2$
        $~|~$ \kw{new} $C$         
\qquad \qquad \hfill \textit{(Expressions)} \\
 & $~|~$ & \kw{new} \atomicS $C$
        $~|~$ $(t)~e$ 
        $~|~$   
      \kw{finish}\SB \kw{async}\SB$e_1$\FB~\kw{async}\SB$e_2$\FB \FB\kc{;} $e_3$ \\
  $v$ & $::=$ & \kw{null}
\hfill \textit{(Values)} \\
    $t$ & $::=$ & $C$ $~|~$ $I$ $~|~$ $\mathbf{Unit}$ \hfill \textit{(Base Types)} \\
             $q$ & $::=$ & \atomicS $|$ \natomicS \hfill \textit{(Atomicity Qualifiers)} \\
  \hline
\end{tabular}
}
\end{table}

\rccc-\oolong extends \oolong to support the declaration of atomic fields and the instantiation of atomic objects, via the \atomicS keyword, and signature annotations, which become the only constructs to express concurrency restrictions. Accordingly, all lock related operations of \oolong are removed.

Table~\ref{tab:syntaxpos} presents the grammar of the language. It uses the following meta-identifiers: 
$C \in \classN$ ranges over class names, $I \in \interfaceN$ over interface names,
$f \in \fieldN$ over field names, $m \in \methodN$ over method names, $x$ and $y \in \varN$ over local variables. %
Programs consist of (lists of) %
interfaces $Ids$ and classes $Cds$, and a starting expression $e$ (that we refer to as \textsf{main}). The remaining expressions are standard, or have a restricted use for the sake of simplification, as in \oolong.
Types are class or interface names, or \textbf{Unit} (used as the type of assignments). Types are assigned to fields, method return values and method parameters. The typing environment maps variables to types and abstract locations to class types.

The list of atomicity combinations   that must be supported in the implementations of each interface method is here
 expressed as signature annotations, which
are sequences of elements $\atan{q_1}{q_2}$, where $q_1$ denotes the atomicity of the parameter and $q_2$ the atomicity of the return type.
As in \rccc-Java, the programmer may also  atomic-qualify the type in class instantiation
to address the case of objects that are assigned to a local variable and then shared among threads. 
This happens in the following example, for the new \icode{BaseList_na}:
\begin{lstlisting}[ numbers=none]
let x = new atomic BaseList_na in let y = x in 
 finish { async { x.add(new Object) } async { y.add(new Object) }} ; null
\end{lstlisting}

Listing~\ref{lst:oolist} presents the \rccc-\oolong version 
of the generated \icode{BaseList_na} specialisation of 
class \icode{BaseList<T, N>}.
For the sake of code readability, in the example we use the sequential composition operator ('\textbf{;}', defined in~\cite{oolong-ACR2019}) to 
omit some of the \letk expressions. We also assume the existence of the \ifk instruction, equality (denoted by ==), and classes \icode{Boolean} and \icode{Integer}.
The code is very similar to what would be a Java version of a concurrent list of \regular objects.
The differences are mainly in interface \icode{List}.
Given that expressions of type \kw{Unit} may be assigned to variables, 
method \icode{add} must support the return of both \atomicS and \natomicS values. Also, as there are no primitive values, method \icode{get} receives an object that may be either \atomicS or \natomicS.

\begin{table}[t]
\begin{lstlisting}[multicols=2, escapechar=|]
interface Object {}

interface List_n {
 add(element : Object) : Unit [$\atan{\natomicS}{\natomicS}$, $\atan{\natomicS}{\atomicS}$]
 get(pos : Integer) : Object [$\atan{\natomicS}{\natomicS}$, $\atan{\atomicS}{\natomicS}$]
 equals(other : List) : Boolean [$\atan{\natomicS}{\natomicS}$, $\atan{\natomicS}{\atomicS}$, $\atan{\atomicS}{\natomicS}$, $\atan{\atomicS}{\atomicS}$]
} 

class Node_n implements Object { 
 value : Object
 next : Node 
 prev : Node 
}

class BaseList_na implements List_n{
 head : atomic Node_n |\label{line:atomichead}|
 tail : atomic Node_n |\label{line:atomictail}|
     
 def add(element : Object) : Unit {
  let node = new Node_n in 
    node.value = element; 
    node.next = null; 
    if (this.head == null) {
     node.prev = null; 
     this.head = node;  
    } else {
     node.prev = tail; 
     tail.next = node;
    } ; this.tail = node
 }
    
 def get(pos : Integer) : Object { ... }    
 def equals(other : List) : Boolean { ... }
}
\end{lstlisting}
\vspace{6pt}
\hrule
\captionof{lstlisting}{\rccc-\oolong version of 
generated class \icode{BaseList_na} and of its dependencies.}
\label{lst:oolist}
\end{table}

\subsection{\rccc Analysis}

The processing of an \rccc-\oolong program
$P_\textrm{orig}$, written in the language defined in Table~\ref{tab:syntaxpos} (\oolong without locks and with signature and atomicity annotations), %
produces an \oolong program (still without locks) with 
safe (\ie type-safe and preventing atomicity violations) code that makes the atomicity qualification of every type in $P_\textrm{orig}$ explicit.
For each type $t$ in $P_\textrm{orig}$, the generated code includes the \regular and atomic versions of $t$.
Moreover, the class types have their original methods unfolded into the corresponding valid variants, and all method calls are explicitly resolved into a valid variant.
The process consists of four stages, as illustrated in Fig.~\ref{fig:stages}.
The first starts by checking if, \rccc annotations aside, the input program is a well-typed \oolong program, and produces a program decorated with local types ($\preprocess$). It may fail if an \oolong typing error occurs.
The second stage computes the valid variants for each method of the program and, for each of these, the atomicities of  its local variables.
This stage may also fail if any of the variants required for the program to execute is not valid. When that %
happens, it returns an empty set 
causing the entire analysis to fail (\func{validVariants}).
Then, the third stage checks if the set of valid variants of each class includes all the signatures defined in the class' \rccc-annotated interface ($\consistent$).
Naturally, it may also  fail if the condition is not met. The first three stages guarantee the correct typing of the program, producing what we call \textit{the solution} ($\sol$) that is used by the fourth stage to generate the final \oolong program ($\vi{P}$).
Stage~1 checks if, \rccc annotations aside, the input program is a well-typed \oolong program, and may fail if an \oolong typing error occurs.
Stage~2 computes the valid variants for each method of the program and, for each of these, the atomicities of  its local variables.
The stage may also fail if any of variants required for the program to execute is not valid. When that %
happens, \func{validVariants} returns an empty set, causing the entire analysis to fail.
Lastly, Stage~3 checks if the set of valid variants of each class include all the signatures defined in the class' \rccc-annotated interface.
Naturally, it may also  fail if the condition is not met.

The first three stages guarantee the correct typing of the program, producing what we call \textit{the solution} ($\sol$) that is used by Stage~4 to generate the final \oolong program.
The complete analysis of an original program is a partial function, denoted $\func{\rccc{}Analysis}$, and defined as follows:
\begin{definition}[\rccc Analysis]\label{def:rc3-analysis}
Consider an original program $P_\textrm{orig}$, that is an annotated \rccc-\oolong program written in the language of Table~\ref{tab:syntaxpos}. If
\begin{description}
    \item[Stage 1:]  $\preprocess(P_\textrm{orig}) = P$
	\item[Stage 2:] $\func{validVariants}(P) = \sol \neq \emptyset$
	\item[Stage 3:] $\consistent(\sol, P) = \truek$
    \item[Stage 4:] $\vi{P}(\sol, P) = P^+$
\end{description}
then $\func{\rccc{}Analysis}(P_\textrm{orig}) = (P,\sol, P^+)$.
\end{definition} 
The final program is guaranteed to be well-typed and behaviourally equivalent to the program resulting of Stage 1 (see \S\ref{sec:soundness}).
The remainder of this section %
 explains each of the 4 stages.

\subsubsection{Stage 1 - The \oolong Type System}

The static analysis is only applied to programs that, \rccc annotations aside, are well-typed \oolong programs.
This property is guaranteed by a pre-processing of the original program: 
\[\footnotesize \preprocess(P_\textrm{orig}) = 
\begin{cases}
\func{dress}(P_\textrm{orig}, \func{\oolong{}TS}(\func{strip}(P_\textrm{orig}))) & \text{if } 
\func{\oolong{}TS}(\func{strip}(P_\textrm{orig})) \text{ succeeds}   \\
 \bot & \text{otherwise.}
\end{cases}
\]
The first step is to $\func{strip}$\footnote{Function defined in 
Appendix~\ref{app:aux}.
\label{note1}} the code from all \atomicS and interface signature annotations.
The resulting program is then submitted to \func{\oolong{}TS}, an instrumented version of the original \oolong type system~\cite{oolong-ACR2019}:
the original verification includes program well-formedness and the matching of the types of the various components of the program; %
the added instrumentation %
incorporates in the type system rules  %
a code %
decoration process that 
assigns types to local variables.
So, in \rccc-\oolong, the syntax of the \letk construct is actually %
\kw{let}~$x:t$~\kc{=}~$e_1$ \kw{in} $e_2$.
This decoration has no impact on the semantics of the \kw{let} instruction -- it can be trivially shown that \oolong's \emph{Progress and Preservation}  results still hold for the decorated language.
The resulting program is lastly \textit{redressed} with the original \rccc annotations to produce the program that will be passed to the subsequent stages of the analysis.%
\footnotemark [\value{footnote}]

\subsubsection{Stage 2  - Valid Method Variants}
\label{sec:valid_variants}
\begin{table}[t]
\caption{Atomicity Inference Output
}
\label{tab:eqsystem}
\begin{footnotesize}
\hrule
$\begin{array}{r c l l}
\nu \in \natval & = & \atomick ~|~  \natomick \\
\natx, \natv{y}  \in \natN && \\ %
\omega \in \valVal & = & \validk ~|~  \invalidk \\
\mu \in \FQmethodN  & ::= & \nu_1.t.m.\nu_2.\nu_3 \\ %
\sol \in \NCSSolution & = & (\natN \rightarrow \natval) \cup (\FQmethodN \rightarrow \valVal) \\
\alpha \in \natV & ::= &\nu ~|~
               \natx ~|~ 
               \fqv{\check x}{\mu}   \\
\eta \in \solexpr & ::= &  \ateq{\alpha_1}{\alpha_2} ~|~ (\natv{\mu},  \omega) ~|~   \eta_1 \vee \eta_2 ~|~ \eta_1 \wedge \eta_2 ~|~ ( \natv{\mu},  \omega)\rightarrow \eta  ~|~ (\eta) \\
\nes \in \NatSys &  =  & \wp(\solexpr) \\
\FQmethodV & ::= & \natx_1.t.m.\natx_2.\natx_3 ~|~ \natx_1.t.m.\natx_2.\natx_3@\mu\\
\mns \in \textsf{Method Call Systems} & = & \wp(\FQmethodV)\\
\mes \in \MetSys &  = & \FQmethodV \rightarrow \NatSys \times \textsf{Method Call Systems} \\
\natv{\mu}  \in \valVar &&
    \\[5pt]
\end{array}$
\end{footnotesize}
\hrule
\end{table}

The purpose of the second stage is to compute the valid variants of each method and, for each of these, obtain the atomicity of its local variables.
The process is defined in function \func{validVariants} that receives the program returned by function \preprocess~and outputs the \textit{solution} ($\sol \in \NCSSolution$, syntax is defined in Table~\ref{tab:eqsystem}) -- a map from variant identifiers to validity values and from atomicity variables to atomicity values:
\[\footnotesize \func{validVariants}(P)  
=  
\begin{cases}
\sol & \text{if } 
\solve(\nes) = (\truek, \sol)   \\
 \emptyset & \text{otherwise.}
\end{cases} \quad {\begin{array}{l}
\text{where:} ~~\vdash{P}\outputeqs{\mes} \text{ and }\\ ~~\inter(\mes) = \nes
\end{array}}
\]
The \func{validVariants} function %
is defined by the composition of three other (partial) functions. 
The first ($\vdash$) performs a type-based generation of atomicity-related information for each method in the input program.
The second (\inter) uses %
this information to build a constraint system with the restrictions that have to be met, in order for a given variant to be valid.
The constraint system is then passed to \solve\ that makes use of a SAT solver to obtain the solution for both variant validity and atomicity of local variables.
It returns a pair with the information of whether the given system has a solution, and if so, the solution itself; if there is no solution (the 'otherwise' case in the definition), \solve\ returns a pair with \falsek in the first position.

The first two steps are performed independently for each method and use only local information, allowing for separate compilation.
The last requires the information collected for the entire program and hence must be done during the linking stage.

\renewcommand{\rulename}[1]{[\textsc{#1}]}
\subsubsection{Type-Based Generation of Atomicity-Related Information}

The generated atomicity-related information is a map that, for each method of the source program, provides: 
   (a) the set of constraints imposed on the atomicity of the method's local variables; and 
   (b) the set of the method variants that the method will call in its execution.
As defined in Table~\ref{tab:eqsystem}, atomicity values $\nu \in \natV$, as used in the analysis, consist of  $\atomick$ (atomic) and $\natomick$ (non-atomic).
Many times it is not possible to infer an expression's atomicity information in a small step analysis, as the atomicity of local variables is inherited from the expression assigned to them, which may be a method parameter or the result of a method call. To refer to such atomicity in subsequent expressions we use atomicity variables $\check x, \natv y \in \natN$. %
The same holds for variants, which are identified by five components: 
atomicity of the object itself (\thisk), name of the class, name of the method, 
atomicity of the method's  parameter, and atomicity of the method's return value.
The latter two atomicities, in particular, are not always possible to infer  in a small step analysis.
Thus, to refer to method variants in method calls and to use them as keys to the generated atomicity-related information, we leverage on atomicity variables to set up variant variables of the form $\natx_1.t.m.\natx_2.\natx_3 \in \FQmethodV$.

\renewcommand{\ottdrulewfXXprogram}[1]{\ottdrule[#1]{%
\ottpremise{
\forall \, \ottnt{Cd} \, \!\!\in\! \, \ottnt{Cds}  \ottsym{.}  \vdash  \ottnt{Cd} \outputeqs{\mes_{\ottnt{Cd}}}
\quad
   \vdash_{\natx_3}	 \ottnt{e} \outputeqs  (\nes,\mns)
\quad \natx_1,  \natx_2,  \natx_3\ \fresh   
   }}%
{
\vdash  \ottnt{Ids} ~ \ottnt{Cds} ~ \ottnt{e}  %
\outputeqs{
\bigcup_{\ottnt{Cd} \in \ottnt{Cds}}.\mes_{\ottnt{Cd}} \cup \{  \natx_1.\kw{Unit}.\texttt{main}.\natx_2.\natx_3 \mapsto (\nes,\mns)\} }}{%
{\ottdrulename{ai\_program}}{}%
}}

\newcommand{\ottdrulewfXXinterfaces}[1]{\ottdrule[#1]{%
}{
	\vdash  \kwool{interface} \, \ottmv{I}  \ldots 
\outputeqs{  \{ I.m.\nu_1.\nu_2 \mapsto \emptyset\  |\ \forall\ m.\nu_1.\nu_2 : t_1 \mapsto t_2\ \in  \func{msigs}(I) \} }}{
{\ottdrulename{ai\_interfaces}}{}%
}}

\renewcommand{\ottdrulewfXXinterface}[1]{\ottdrule[#1]{%
	\ottpremise{\forall \, \ottnt{Msig} \, \!\!\in\! \, \mathit{Msigs}  \ottsym{.}\   \kwool{this}  \ottsym{:}  \ottmv{I}  \vdash  \ottnt{Msig} \outputeqs{\ottnt{\mes}_{\ottnt{Msig}}}}%
}{
	\vdash  \kwool{interface} \, \ottmv{I}  \,\startblock\,  \ottnt{Msigs}  \,\finishblock\, 
	\outputeqs{\bigcup_{\ottnt{Msig} \in \ottnt{Msigs}} \ottnt{\mes}_{\ottnt{Msig}}}}{
{\ottdrulename{ai\_interface}}{}%
}}

\newcommand{\ottdrulewfXXinterfaceextends}[1]{\ottdrule[#1]
{%
   \ottpremise{
   \vdash  \kwool{interface} \, \ottmv{I_i} \ldots \outputeqs{\mes_i}
   \quad
   \ottnt{\mes'_i}=\ottnt{\mes_i}[I_i\mapsto I]
   }%
}
	{
	\vdash  \kwool{interface} \, \ottmv{I}\  \kwool{extends} \, \ottmv{I_1}, \ottmv{I_2}
	\outputeqs{\ottnt{\mes'_1} \cup \ottnt{\mes'_2}}}{
{\ottdrulename{ai\_iextends}}{}%
}}

\newcommand{\ottdefnwfXXinterface}[1]{\begin{ottdefnblock}[#1]{$\vdash  \ottnt{Id} \outputeqs{\mes} $}{}
\ottusedrule{\ottdrulewfXXinterfaces{}}
\end{ottdefnblock}}

\newcommand{\ottdrulewfXXmsig}[1]{\ottdrule[#1]{%
  }{
 \thisk : I \vdash  \ottmv{m}.\nu_1.\nu_2  \ottsym{(}  \ottmv{x}  \ottsym{:}  \ottnt{t_1}  \ottsym{)}  \ottsym{:}  \ottnt{t_2}\, 
\outputeqs{  \{I.m.\nu_1.\nu_2(x : t_1) : t_2) \mapsto \emptyset \}  }
}{%
{\ottdrulename{ai\_msig}}{}%
}
}

\newcommand{\ottdefnwfXXmsig}[1]{\begin{ottdefnblock}[#1]{$\vdash  \ottnt{Msig} \outputeqs{\mes} $}{}
\ottusedrule{\ottdrulewfXXmsig{}}
\end{ottdefnblock}}

\renewcommand{\ottdrulewfXXclass}[1]{\ottdrule[#1]{%
\ottpremise{\forall \, \ottnt{Md} \, \!\!\in\! \, \mathit{Mds}  \ottsym{.}\   \thisk : \natv x.C  \vdash \ottnt{Md} \outputeqs{\ottnt{\mes}_{\ottnt{Md}}
\quad
\natv x\ = \natvar(\thisk) %
}
}%
  }{
\vdash  \kwool{class} \, \ottmv{C} \, \kwool{implements} \, \ottmv{I}  \,\startblock\,  \ottnt{Fds} \, \mathit{Mds}  \,\finishblock\, %
\outputeqs{ \bigcup_{\ottnt{Md} \in \ottnt{Mds}} {\ottnt{\mes}_{\ottnt{Md}}}}}{
{\ottdrulename{ai\_class}}{}%
}}

\renewcommand{\ottdefnwfXXClass}[1]{\begin{ottdefnblock}[#1]{$\vdash  \ottnt{Cd} \outputeqs{\mes} $}{}
\ottusedrule{\ottdrulewfXXclass{}}
\end{ottdefnblock}}

\renewcommand{\ottdrulewfXXfield}[1]{\ottdrule[#1]{%
\ottpremise{}}{
\kwool{this}  \ottsym{:}  \ottmv{C} \vdash  \ottmv{f}  \ottsym{:}  \ottnt{t} \outputeqs{\{(C.f , \natureof(t))\}}}%
{\ottdrulename{ai\_field}}{}%
}

\renewcommand{\ottdefnwfXXField}[1]{\begin{ottdefnblock}[#1]{$\vdash  \ottnt{Fd} \outputeqs{(\mes, \mns)} $}{}
\ottusedrule{\ottdrulewfXXfield{}}
\end{ottdefnblock}}

\renewcommand{\ottdrulewfXXmethod}[1]{\ottdrule[#1]{%
	\thisk  \ottsym{:}  \ottmv{\natx_1.C},
	\ottmv{x}  \ottsym{:}  \ottnt{\natv x_2.t_1}   \vdash_{\natv x_3}  \ottnt{e}\  \outputeqs{(\nes,\mns)} 
	\quad
	 \natx_2 = \natvar(x) 
	 \quad
 	\natv x_3\ \fresh
 }{
\thisk  \ottsym{:}  \ottmv{\natx_1.C} \vdash  \kwool{def} \, m  \ottsym{(}  \ottmv{x}  \ottsym{:}  \ottnt{t_1}  \ottsym{)}  \ottsym{:}  \ottnt{t_2}  \,\startblock\,  \ottnt{e}  \,\finishblock\,  \outputeqs{  \{ \natx_1.C.m.\natv x_2.\natv x_3 \mapsto (\nes,\mns)\} } 
 }{%
{\ottdrulename{ai\_method}}{}%
}
}

\renewcommand{\ottdefnwfXXMethod}[1]{\begin{ottdefnblock}[#1]{$\Upsilon \vdash \ottnt{Md} \outputeqs{\mes} $}{}
\ottusedrule{\ottdrulewfXXmethod{}}
\end{ottdefnblock}}

\renewcommand{\ottdefnwfXXProgram}[1]{\begin{ottdefnblock}[#1]{$\vdash  \ottnt{P}   \outputeqs{\mes}$}{}
\ottusedrule{\ottdrulewfXXprogram{}}
\end{ottdefnblock}}

\renewcommand{\ottdrulewfXXvar}[1]{\ottdrule[#1]{%
\ottpremise{\Upsilon(x)=\natx.t}}{
 \Upsilon \vdashel \ottmv{x}\ \outputeqs{(\{ \ateq{\natx}{\natlhs}\}, \emptyset)}}%
{\ottdrulename{ai\_var}}{}%
}

\renewcommand{\ottdrulewfXXlet}[1]{
\ottdrule[#1]{%
  \natv x = \natvar(x) 
    \quad
  \Upsilon \vdashx{\natv x} \ottnt{e_{{\mathrm{1}}}} \outputeqs{(\nes_1, \mns_1)} 
  \quad
  \Upsilon \ottsym{,} \ottmv{x} \ottsym{:} \ottnt{\natv x.t}  \vdashe \ottnt{e_2} \outputeqs{(\nes_2, \mns_2)} 
}{
  \Upsilon \vdashe  \kwool{let}\,  \ottmv{x : t} = \ottnt{e_{{\mathrm{1}}}} ~\kwool{in}~ \ottnt{e_2}\ \outputeqs{(\nes_1 \cup \nes_2, \mns_1 \cup \mns_2)}
}{%
{\ottdrulename{ai\_let}}{}%
}
}

\renewcommand{\ottdrulewfXXcall}[1]{
\ottdrule[#1]{%
	\Upsilon(x)=\natx_1.t
	\quad
	\natv x_2 = \nextvar 
	\quad 
	\natv x_3 = \nextvar
	\quad
    \Upsilon \vdashx{\natx_2} \ottnt{e} \outputeqs{(\nes, \mns)}
}{
\Upsilon  \vdashe  x.m(e)  \outputeqs{(\nes \cup  %
	\{ \ateq{\natv x_3}{\natlhs} \}, 
	\mns \cup \{{\natx_1.t.m.\natv x_2.\natv x_3 }\}) }}{%
{\ottdrulename{ai\_call}}{}%
}
}

\renewcommand{\ottdrulewfXXcast}[1]{\ottdrule[#1]{%
	\Upsilon \vdashe \ottnt{e}\ \outputeqs{(\nes,\mns)}%
	\quad
	\natx = \nextvar 
}{
\Upsilon \vdashe \ottsym{(} \ottnt{t} \ottsym{)}\ \ottnt{e}\ \outputeqs{(\nes \cup \{ \ateq{\natx}{\natlhs} \},\mns)}}{%
{}{\ottdrulename{ai\_cast}}%
}}

\renewcommand{\ottdrulewfXXselect}[1]{\ottdrule[#1]{%
	\Upsilon(x)=\natx.C \quad
}{
\Upsilon \vdashe \ottmv{x} \ottsym{.} \ottmv{f} \outputeqs{(\{
(\ateq{\natx}{\atomick} \wedge \ateq{\natlhs}{\func{fatom}(C, \atomick, f)}) \vee (\ateq{\natx}{\natomick} \wedge \ateq{\natlhs}{\func{fatom}(C, \natomick, f)})
\},\emptyset)}}{%
{}{\ottdrulename{ai\_select}}%
}}

\renewcommand{\ottdrulewfXXloc}[1]{\ottdrule[#1]{%
\ottpremise{\vdash \Upsilon}%
\ottpremise{\Upsilon \ottsym{(} \iota \ottsym{)} \ottsym{=} \ottnt{C}}%
\ottpremise{ \regularV(\ottnt{C}) <: \ottnt{t} }%
}{
\Upsilon \vdash \iota \ottsym{:} \ottnt{(t,\natureof(C))} \outputeqs{(\emptyset,\emptyset)}}{%
{}{\ottdrulename{ai\_loc}}%
}}

\renewcommand{\ottdrulewfXXnull}[1]{\ottdrule[#1]{%
}{
\Upsilon \vdashe \kwool{null} \outputeqs{(
\{ \ateq{\natv y}{\natv y} \}, \emptyset)}}{%
{}{\ottdrulename{ai\_null}}%
}
}

\renewcommand{\ottdrulewfXXupdate}[1]{\ottdrule[#1]{%
\ottpremise{\Upsilon(x)=\natx.C
\quad
\natx_1\ \fresh
\quad
\Upsilon \vdashx{\natx_1} \ottnt{e}\ \outputeqs (\nes,\mns) \ottlinebreakhack 
}%
}{
\Upsilon \vdashe \ottmv{x} \ottsym{.} \ottmv{f} \ottsym{=} \ottnt{e}\ 
  \outputeqs{
  (\nes \cup \{ \ateq{\natlhs}{\natlhs}, 
(\ateq{\natx}{\atomick} \wedge \ateq{\natx_1}{\func{fatom}(C, \atomick, f)}) \vee (\ateq{\natx}{\natomick} \wedge \ateq{\natx_1}{\func{fatom}(C, \natomick, f)}) \}, \mns)}}{%
{}{\ottdrulename{ai\_update}}%
}}

\renewcommand{\ottdrulewfXXnew}[1]{\ottdrule[#1]{%
\natx = \nextvar
}{
\Upsilon \vdashe \kwool{new} \, \ottmv{C}\ \outputeqs{(\{ \ateq{\natx}{\natlhs} \},  \emptyset)}}{%
{}{\ottdrulename{ai\_new}}%
}}

\newcommand{\ottdrulewfXXnewat}[1]{\ottdrule[#1]{%
}{
\Upsilon \vdashe \kwool{new}\, \atomicS \, \ottmv{C}\ \outputeqs{(\{ \ateq{\natlhs}{\atomick} \},  \emptyset)}}{%
{}{\ottdrulename{ai\_newat}}%
}}

\renewcommand{\ottdrulewfXXfj}[1]{\ottdrule[#1]{%
\ottpremise{
\Upsilon \vdashx{\natx_1} \ottnt{e_{{\mathrm{1}}}}\ \outputeqs{(\nes_{{\mathrm{1}}},\mns_{{\mathrm{1}}}})
\quad
\Upsilon \vdashx{\natx_2}  \ottnt{e_{{\mathrm{2}}}}\ \outputeqs{(\nes_{{\mathrm{2}}},\mns_{{\mathrm{2}}}})
	\natx_1,
	\natx_2\ \fresh
\quad
\Upsilon \vdashe \ottnt{e}\ \outputeqs{(\nes,\mns)}%
	}
}{
\Upsilon \vdashe \kwool{finish} \, \,\startblock\, \, \kwool{async} \, \,\startblock\, \ottnt{e_{{\mathrm{1}}}} \,\finishblock\, \, \kwool{async} \, \,\startblock\, \ottnt{e_{{\mathrm{2}}}} \,\finishblock\, \,\finishblock\, \ottsym{;} \ottnt{e} 
\outputeqs{(\nes_{{\mathrm{1}}} \cup \nes_{{\mathrm{2}}} \cup \nes,\mns_{{\mathrm{1}}} \cup \mns_{{\mathrm{2}}} \cup \mns)}}{%
{\ottdrulename{ai\_fj}}{}%
}}

\renewcommand{\ottdrulewfXXlock}[1]{\ottdrule[#1]{%
\ottpremise{\Upsilon \vdash \ottmv{x} \ottsym{:} \ottnt{(t_{{\mathrm{2}}},\natlhs_{{\mathrm{2}}})} \outputeqs{\nes_{{\mathrm{2}}}}}%
\ottpremise{ \Upsilon \vdash \ottnt{e} \ottsym{:} \ottnt{(t,\natlhs)} \outputeqs{\nes} \!\!\! }%
}{
\Upsilon \vdash  \kwool{lock} ( \ottmv{x} ) \kwool{\,in\,} \ottnt{e}  \ottsym{:} \ottnt{(t,\natlhs)} \outputeqs{\nes}}{%
{\ottdrulename{ai\_lock}}{}%
}}

\renewcommand{\ottdrulewfXXlocked}[1]{\ottdrule[#1]{%
\ottpremise{\Upsilon \vdash \ottnt{e} \ottsym{:} \ottnt{(t,\natlhs)}\outputeqs{\nes}}%
\ottpremise{\Upsilon \ottsym{(} \iota \ottsym{)} \ottsym{=} \ottnt{t_{{\mathrm{2}}}}}%
}{
\Upsilon \vdash  \kwool{locked}_{ \iota } \{ \ottnt{e} \}  \ottsym{:} \ottnt{(t,\natlhs)}\outputeqs{\nes}}{%
{\ottdrulename{ai\_locked}}{}%
}}

\newcommand{\ottdrulewfXXif}[1]{\ottdrule[#1]{%
\ottpremise{\Upsilon \vdash \ottnt{e_{{\mathrm{1}}}} \ottsym{:} \ottnt{(t_1,\natlhs_1)} \outputeqs{\nes_1}}%
\ottpremise{\Upsilon \vdash \ottnt{e_{{\mathrm{2}}}} \ottsym{:} \ottnt{(t_{{\mathrm{1}}},\natlhs_{{\mathrm{2}}})} \outputeqs{\nes_{{\mathrm{2}}}}}%
\ottpremise{\Upsilon \vdash \ottnt{e_{{\mathrm{3}}}} \ottsym{:} \ottnt{(t_{{\mathrm{3}}},\natlhs_{{\mathrm{3}}})} \outputeqs{\nes_{{\mathrm{3}}}}}%
\ottpremise{\Upsilon \vdash \ottnt{e_{{\mathrm{4}}}} \ottsym{:} \ottnt{(t_{{\mathrm{3}}},\natlhs_{{\mathrm{4}}})} \outputeqs{\nes_{{\mathrm{4}}}}}%
}{
\Upsilon \vdash  \ifk\ (e_{{\mathrm{1}}} == e_{{\mathrm{2}}})\ e_{{\mathrm{3}}}\ \elsek\ e_{{\mathrm{4}}}  \ottsym{:} \ottnt{(t_3,\natlhs_\mathrm{3})} 
\outputeqs{\nes_{{\mathrm{1}}} \cup \nes_{{\mathrm{2}}} \cup \nes_{{\mathrm{3}}} \cup \nes_{{\mathrm{4}}} \cup \{(\natlhs_1, \natlhs_2), (\natlhs_\mathrm{3},\natlhs_\mathrm{4}) \}}}{%
{\ottdrulename{ai\_if}}{}%
}}

\newcommand{\ottdrulewfXXseq}[1]{\ottdrule[#1]{%
\ottpremise{\Upsilon \vdash \ottnt{e_{{\mathrm{1}}}} \ottsym{:} \ottnt{(t_{{\mathrm{1}}},\natlhs_{{\mathrm{1}}})} \outputeqs{\nes_{{\mathrm{1}}}}}%
\ottpremise{\Upsilon \vdash \ottnt{e_{{\mathrm{2}}}} \ottsym{:} \ottnt{(t_{{\mathrm{2}}},\natlhs_{{\mathrm{2}}})} \outputeqs{\nes_{{\mathrm{2}}}}}%
}{
\Upsilon \vdash  \ottnt{e_{{\mathrm{1}}}} ; \ottnt{e_{{\mathrm{2}}}}  \ottsym{:} \ottnt{(t_{{\mathrm{2}}},\natlhs_{{\mathrm{2}}})} \outputeqs{\nes_{{\mathrm{1}}} \cup \nes_{{\mathrm{2}}}}}{%
{\ottdrulename{ai\_seq}}{}%
}}

\renewcommand{\ottdefnwfXXExpression}[1]{\begin{ottdefnblock}[#1]{$\Upsilon \vdash_{\natlhs} \ottnt{e}\ \outputeqs{(\nes,  \mns)} $}{}
\ottusedrule{\ottdrulewfXXvar{}} \quad
\ottusedrule{\ottdrulewfXXnull{}}\\[5pt]
\ottusedrule{\ottdrulewfXXlet{}}\\[5pt]
\ottusedrule{\ottdrulewfXXcall{}}\quad  \ottusedrule{\ottdrulewfXXnew{}}\\[5pt]
\ottusedrule{\ottdrulewfXXnewat{}} \quad \ottusedrule{\ottdrulewfXXcast{}}\\[5pt] 
\ottusedrule{\ottdrulewfXXupdate{}}\\[5pt] 
\ottusedrule{\ottdrulewfXXselect{}}\\[5pt] 
\ottusedrule{\ottdrulewfXXfj{}}\\[5pt]
\end{ottdefnblock}}

\begin{table}[t]
\caption{Generation of a program's atomicity information.}  %
\label{tab:ni}
\begin{rulesize}
\hrule 
\ottdefnwfXXProgram{}\\[5pt]
\ottdefnwfXXClass{}\\[5pt]
\ottdefnwfXXMethod{}\\[5pt]
\hrule 
\ottdefnwfXXExpression{}
\hrule 
$\func{fatom}(t, \nu, f) = 
\begin{cases}
\atomick & \fields(t)(f) = \atomicS\ t \vee  (\fields(t)(f) = t \wedge \nu = \atomick) \\
\natomick & \text{otherwise}  
\end{cases}$ \\
Function \func{fields}(C)(f) is defined as in \cite{oolong-ACR2019} and returns the possibly qualified type of field $f$ of class $C$.\\
Function \func{interfaceOf}(C) consults the code to return the interface implemented by class $C$. 
\hrule 
\end{rulesize}
\end{table}

The rules of the type system  are presented in Table~\ref{tab:ni}.
The variant identifiers $\natx_1.t.m.\natx_2.\natx_3$ used as keys in rule
\rulename{ai\_method} have their 
$\natx_1$ and $\natx_2$ components constructed from program variable names, respectively, \thisk (in rule \rulename{ai\_class}) and the method parameter, by using function 
$\natvar$ that generates an atomicity variable from program variable.
The identifier $\natx_3$ is a fresh variable but could also be constructed from a reserved identifier, such as \returnk.
Rule \rulename{ai\_program} unites the information gathered for each class with the one generated for the \emph{main} expression, which is also stored in the table under a special key.

Rules for expressions receive: (a) a
typing environment with elements of the form $x : \natx.t$, where $\natx$ and $t$ denote, respectively, the atomicity and type of variable $x$; and (b) an atomicity variable ($\natv y$) that carries the atomicity of the expression evaluation's recipient, be it a variable, a field or a method call argument.
The output is a pair, whose %
first element ($\nes$) is a set of expressions that impose  constraints on the values of atomicity variables, while the second ($\mns$) is the aforementioned set of method variants that the expression calls. 

The atomicity of a local variable is given by the expression assigned to it.
To that end, rule~\rulename{ai\_let} creates a new atomicity variable ($\natx$) to represent the atomicity of variable $x$ and uses it in the typing of $e_1$ (as the recipient's atomicity) and $e_2$ (in the typing environment).
The \emph{recipient's atomicity} variable is actively  used in several other rules. 
An example is  rule \rulename{ai\_var}, where it is used to ensure  that the atomicity of the recipient is the same  of the variable.
In rules where it is not relevant, 
such as \rulename{ai\_null} --  \nullk  does not have a defined atomicity and, hence, may be passed to any recipient --
we simply add the constraint $\ateq{\natv y}{\natv y}$ to ensure that all atomicity variables are represented in the constraint set. 
A simple example   of the application of these initial rules follows, where
	 $\natx_2 = \natvar(x)$,  
	 $\natv x_3\ \fresh$,
	  $\hat y = \natvar(y)$, 
	 and $\hat z = \natvar(z)$; the generated constraints ensure that the atomicity of the parameter, of the return type, and of variable~$y$ bound by the first let-expression are all equal:
\begin{align*}
  & \thisk : \natx_1.C \vdash \defk\ m(x : t) : t\ \{ \letk\ y : t = x\ \ink\ \letk\ z : t = \nullk\ \ink\ y \}  
  \outputeqs{} \\
  & \qquad \{\natx_1.C.m.\natx_2.\natx_3 \mapsto ( \{  \ateq{\natv y}{\natx_2}, \ateq{\natv z}{\natv z}, \ateq{\natv y}{\natx_3} \}, \mns)\}
\end{align*}
When calling a method, several variants may be available.
The choice of the one to call is determined by the atomicity of the target object, of the expression passed as argument, and of the recipient.
Accordingly, to compose the identifier of the variant to add to the set of calls, rule~\rulename{ai\_call}
obtains the first three components ($\natx_1$, $t$ and $m$) from the expression and the typing environment.
The atomicity of the parameter  ($\natx_2$)  and of the return value  ($\natx_3$)  require  fresh variables, which are used to match the referred atomicities with the one of the expression passed as argument
and with the one of the recipient of the call's return value.
We will need to refer to variables $\natx_2$ and $\natx_3$ in subsequent stages of our analysis and hence need to know which identifiers have been generated.
The usual approach is to store this information on a data structure, along with context information, and propagate it along the analysis.
For the sake of simplicity and readability, in this presentation we opt for a different solution.
We make use of the  \nextvar\ deterministic  variable generator that, once reset,  
always generates the same sequence of variables.
Given that we will visit the source program always by the same order, if every rule generates the same amount of variables, we will always obtain the same identifier at the same points of the analysis. %
We make use of this generator every time we need  (to remember) an atomicity variable that cannot be constructed from the code nor is reflected in the variant identifier, 
such as in rule \rulename{ai\_new}, where it used to represent the atomicity of the class being instantiated, or in rule 
\rulename{ai\_cast}, where it is used to guarantee that the atomicity of the type used is the same of the one of the receiver and of the  expression being cast.

The constraints imposed by rules~\rulename{ai\_select} and~\rulename{ai\_update} are based on the atomicity of the field under analysis, which depends on the atomicity of the target object.  
When this object is not atomic ($\ateq{\natx}{\natomick})$, the atomicity of the field is the one coded in the program.
Otherwise, as explained in \S\ref{sec:atomic_types_and_method_variants} and defined in 
function \func{fatom},
the field also becomes atomic if it is of the same type of the enclosing class.
The application of the rules to a simplified version of method \texttt{add} from Listing~\ref{lst:list} follows:

 \begin{footnotesize}
 \begin{align}
 \begin{split}
 & \thisk : \natx_1.\texttt{BaseList\_na} \vdash \defk\  \texttt{add}(\texttt{element} : \texttt{Object}) : \kw{Unit}\ \{ 
  \letk\ \texttt{node : Node} = \newk\ \texttt{Node}\ \ink \\ 
& \quad   \letk\ \_ : \kw{Unit} = \texttt{node.value = element}\ \ink\ (\dots)\ \thisk.\texttt{tail = node} \}  
  \outputeqs{} \\
&  \{\natx_1.\texttt{BaseList\_na}.\texttt{add}.\natx_2.\natx_3 \mapsto ( \{  \\
& \qquad 
\ateq{\natv{\texttt{node}}}{\natx_4},  
\qquad \qquad \qquad \qquad  \qquad \qquad  \qquad \quad \qquad \qquad 
\text{Added by } \rulename{ai\_new}, 
\text{for } \natx_4\ = \nextvar \\
&
  \qquad \qquad \qquad \qquad \qquad \qquad \qquad \qquad \qquad \quad \qquad \qquad
  \qquad \quad
 \text{ and } \natv{\texttt{node}} = \natvar(\texttt{node})\\
& \qquad 
\ateq{\natx_6}{\natx_6}, 
(\ateq{\natv{\texttt{node}}}{\atomick} \wedge \ateq{\natx_5}{\natomick})
  \vee (\ateq{\natv{\texttt{node}}}{\natomick} \wedge \ateq{\natx_5}{\natomick}),
  \quad  
  \text{Added by } \rulename{ai\_update}, \\
  &  
  \qquad \qquad \qquad \qquad \qquad \qquad \qquad \qquad \qquad \quad \qquad
  \qquad \qquad \qquad  
  \text{for } \natx_5\ \text{fresh  and } \natx_6 = \natvar(\texttt{\_})\\
&  \qquad 
\ateq{\natx_5}{\natx_2}, 
\qquad \qquad \qquad \qquad \qquad \qquad \qquad \qquad \qquad \qquad
\text{Added by } \rulename{ai\_var} \\
  & \qquad 
  \ateq{\natx_3}{\natx_3},
  (\ateq{\natx_1}{\atomick} \wedge \ateq{\natx_7}{\atomick})
  \vee (\ateq{\natx_1}{\natomick} \wedge \ateq{\natx_7}{\atomick}), 
  \qquad  \quad
  \text{Added by } \rulename{ai\_update}, \text{for } \natx_7\ \text{fresh} \\
    & \qquad  
  \ateq{\natx_7}{\natv{\texttt{node}}}
\qquad \qquad \qquad \qquad \qquad \qquad \qquad \qquad \qquad \ \quad 
  \text{Added by } \rulename{ai\_var} \\
 & \quad \}, \emptyset)\}
 \end{split} \label{eq:1}
\end{align}
 \end{footnotesize}

We thus have that the atomicity of local variable \texttt{node}
must be the same as of the type instantiated in \texttt{\newk Node} ($\natx_4)$ and of 
\texttt{\thisk.tail} ($\natx_7$), whose atomicity is always $\atomick$, independently of the atomicity of 
\thisk ($\natx_1)$. Also the atomicity of parameter \texttt{element} ($\natx_2$) must be the same of \texttt{node.value} ($\natx_5$), which is always $\natomick$, independently of the atomicity of \texttt{node}.
Given that the method's return type is \kw{Unit}, no constraints are imposed on its atomicity ($\natx_3$).

\paragraph{Variant Constraints.}
Function \inter\ takes the atomicity-information previously generated for each method (collected in $\mes$) and builds a single (global) constraint system, comprising a variant validity constraint per every possible variant of each method. 
Validity of each variants is represented by a unique validity variable, used in the validity constraints so as to be solved. %

\begin{footnotesize}
\setlength{\mathindent}{0cm}
\begin{align*}
\inter : \MetSys  \rightarrow\ & \NatSys  \\
 \inter(\mes)   =\ &
 \func{map}(\lambda x.\vvMCS(x), \mes) %
\end{align*}
\begin{align*}
 \vvMCS(\natx_1.\kw{Unit}.\texttt{main}.\natx_2.\natx_3 \mapsto (\nes,\mns))   = ( 
 \valeq{\valvar(\maine)}{\validk} \land \nes@\maine \land \func{bindCall}(\mns, \maine)) \\
  \text{ where } \maine = \mainid \\
 \vvMCS(\natx_1.C.m.{\natx}_2.{\natx}_3 \mapsto (\nes,\mns))   =\  \hspace{-0.6cm}
\bigwedge_{\begin{subarray}{l}{\nu_1,\nu_2,\nu_3 \in \{\atomick,\natomick\}}\\
\mu = \nu.C.m.\nu_1.\nu_2 
    \end{subarray}}  \hspace{-0.6cm}
 (\valeq{\valvar(\mu)}{\validk}) \rightarrow
\Big((\nes \cup \{ \ateq{\natx_1}{\nu_1}, \ateq{\natx_2}{\nu_2}, \ateq{\natx_3}{\nu_3}\})@\mu  \\[-20pt]
 {\hspace{4cm}
\wedge~ \func{bindCall}(\mns, \mu)
\Big)}
\end{align*}
\begin{align*}
\func{bindCall}(\mns,\mu) =\ & 
\hspace{-.7cm}
\bigwedge_{\natx_1.C.m.{\natx}_2.{\natx}_3 \in \mns}
\Big( \bigvee_{\nu_1,\nu_2,\nu_3 \in \{\atomick,\natomick\}}
 \hspace{-0.5cm}
 \big( \valeq{\valvar(\nu_1.C.m.\nu_2.\nu_3)}{\validk} \land \ateq{\natx_1@\mu}{\nu_1}\land\ateq{{\natx}_2@\mu}{\nu_2}\land\ateq{{\natx}_3@\mu}{\nu_3}\big)
\Big)\\
\valvar(\mu) \ & \text{generates a validity variable $\natv \mu$ from variant identifier $\mu$.} \\
\nes@\mu \ &  \text{defined in Appendix~\ref{app:aux}, qualifies all occurrences of unqualified atomicity} \\ & \text{variables ($\natx$) with the given variant identifier $\mu$: $\natx@\mu$}.
\end{align*}
 \end{footnotesize}

The \emph{main} expression has a single variant, as it does not belong to any class, has no input parameter and the result of its evaluation may simply be discarded.
We represent this variant with special identifier $\mainid$, shortened to \maine.
The validity of \maine\ is required to ensure that the program has an initial expression, meaning that all the constraints ($\nes$) imposed by the program's original initial expression 
must be satisfied, and the variants called must be valid.
For global uniqueness, the local variables in  $\nes$ are qualified with the variant's id.%
\footnote{This is not important in the case of \maine, because there is only one variant, but is needed for methods in general.}
The constraints needed to guarantee that, for each method to call,  there is a valid variant that satisfies the atomicity restrictions  
imposed by the caller are generated by function \func{bindCall}.
For each variant variable received ($\natx_1.C.m.{\natx}_2.{\natx}_3  \in \mns$), the function
adds a constraint to ensure that at least one variant (with identifier 
$\nu_1.C.m.\nu_2.\nu_3$ for $\nu_1,\nu_2$ and $\nu_3 \in \{\atomick,\natomick\}$) 
of the target method  is both valid and matches the atomicity restrictions for the object, the parameter and the return value.
The resulting constraints are conjugated to ensure satisfiability of all calls.

The case for methods is similar, having its differences grounded on the fact that methods have  $2^n$ possible variants (for $n=3$ in \rccc-\oolong) and not all of these have to be valid.
The validity of a variant implies the satisfiability of the restrictions imposed by the method's body, complemented by the constraints imposed by
the variant on the atomicity of the object, parameter and return value, as well as the  satisfiability of the restrictions imposed by \func{bindCall}.

\paragraph{Solving the Constraint System.}
The global constraint system only makes use of equality, implication, conjunction and disjunction operations over two sets of binary values.
It can thus easily be transformed into a Boolean satisfiability problem
by mapping values $\atomick$ and \validk into \truek and, $\natomick$ and  \invalidk into \falsek.
The model resulting from the system's satisfiability
assigns  Boolean values to all validity and atomicity variables, producing what we call the \emph{solution}.
For the example of Eq.~\ref{eq:1}, independently of the variant, we have that $\natv{\texttt{node}} = \natx_4 = \natx_7 = \atomick$ and that 
$\natx_5 = \natx_2 = \natomick$. So, only variants on which $\natx_2 = \natomick$ are valid. 
For those we have (omitting auxiliary variables):

\begin{footnotesize}
\begin{align*}
\hspace{-0.3cm} \{\ &
\atomick.\texttt{BaseList\_na.add}.\natomick.\atomick = \validk, 
\atomick.\texttt{BaseList\_na.add}.\natomick.\atomick.\texttt{node} = \atomick,
\atomick.\texttt{BaseList\_na.add}.\natomick.\atomick.\natx_4 = \atomick, \\
\hspace{-0.3cm}&
\natomick.\texttt{BaseList\_na.add}.\natomick.\atomick = \validk,
\natomick.\texttt{BaseList\_na.add}.\natomick.\atomick.\texttt{node} = \atomick,
\natomick.\texttt{BaseList\_na.add}.\natomick.\atomick.\natx_4  = \atomick, \\
\hspace{-0.3cm}&
\atomick.\texttt{BaseList\_na.add}.\natomick.\natomick = \validk, 
\atomick.\texttt{BaseList\_na.add}.\natomick.\natomick.\texttt{node} = \atomick,
 \atomick.\texttt{BaseList\_na.add}.\natomick.\natomick.\natx_4  = \atomick, \\
\hspace{-0.3cm}& 
\natomick.\texttt{BaseList\_na.add}.\natomick.\natomick = \validk, 
\natomick.\texttt{BaseList\_na.add}.\natomick.\natomick.\texttt{node} = \atomick,
\natomick.\texttt{BaseList\_na.add}.\natomick.\natomick.\natx_4  = \atomick\
\}
\end{align*}
\end{footnotesize}
We show the values for $\natx_4$, because it is the variable that will be used in the code generation  to known the type to instantiate in the $\newk\ \texttt{Node}$ expression, which in this case will be \atomicS \texttt{Node}.

\subsubsection{Stage 3 - Interface Implementation}

An \rccc-\oolong interface may feature signature annotations
that enable the programmer to explicitly convey the atomicity of  the methods parameter and return type.
The third stage of the \rccc analysis has the goal of guaranteeing that
the set of valid variants computed for each class $C$ includes the signatures that result from the parsing of the interface $I$ implemented by $C$.
 To that end,  we define function \consistent\ that receives the   solution obtained from \func{validVariants} and the source code.
The solution is used to retrieve the signatures of the valid variants 
of the class' methods, while the code is used  the retrieve the signatures originally defined in the program's classes and  interfaces.
These signatures are represented 
by triples of the form $ (\nu_1, \nu_2, m(x: t_1) : t_2)$,
where $\nu_1$ and $\nu_2$ denote, respectively,  the atomicity of the parameter and of the return value, and $m(x: t_1) : t_2$
denotes the method's original non-annotated signature.

\begin{footnotesize}
\begin{align*}
\hspace{-0.3cm}\consistent(\sol, \Ids\ \Cds\ e) =\ & \forall \classk\ C\ \kw{implements}\ I\ \{ \Mds \}  \in \Cds\ ~.~
\forall \nu \in \natval\ ~.~ \\
\hspace{-0.3cm}& \hspace{-0.5cm}\bigcup_{\IMsig \in \func{msigs}(I)} \hspace{-0.5cm} \any(\IMsig)  
\subseteq   \hspace{-0.5cm}
\bigcup_{\Msig \in \func{msigs}(C)} \hspace{-0.5cm} \variantmsigs(\sol, \nu, C, \Msig) \\[5pt]
\hspace{-0.3cm}\any(m(x : t_1) : t_2\ \Sas) =\ &
\{ (\func{atom}(q_1), \func{atom}(q_2), m(x: t_1) : t_2)   ~|~  \atan{q_1}{q_2} \in \Sas\} \\
\hspace{-0.3cm}\func{atom}(\atomicS) =\ & \atomick  \qquad \func{atom}(\natomicS) = \natomick \\
\hspace{-0.3cm}\variantmsigs(\sol, \nu, C, m(x : t_1) : t_2) =\ &
    \{\
(\nu_1, \nu_2, m(x_1 : t_1) : t_2) ~|~  \sol(\valvar(\nu.C.m.\nu_1.\nu_2))=\validk \}
\end{align*}
\end{footnotesize}
The set of signatures defined by interface \texttt{List} of our running example is:

\begin{footnotesize}
\begin{align*}
\hspace{-0.4cm}\{\ & (\natomick, \natomick, \texttt{add}(x : \texttt{Object}) : \kw{Unit}),  
(\natomick, \atomick, \texttt{add}(x : \texttt{Object}) : \kw{Unit}), 
(\natomick, \natomick, \texttt{get}(x : \texttt{Integer}) : \texttt{Object}),  \\
\hspace{-0.4cm}&
 (\atomick, \natomick, \texttt{get}(x : \texttt{Integer}) : \texttt{Object}),
(\natomick, \natomick, \texttt{equals}(x : \texttt{List}) : \texttt{Boolean}), 
(\natomick, \atomick, \texttt{equals}(x : \texttt{List}) : \texttt{Boolean}), \\
\hspace{-0.4cm}&
(\atomick, \natomick, \texttt{equals}(x : \texttt{List}) : \texttt{Boolean}), 
(\atomick, \atomick, \texttt{equals}(x : \texttt{List}) : \texttt{Boolean}) \
 \}
\end{align*}
\end{footnotesize}

On the other hand, in agreement with the discussion about valid variants in \S\ref{sec:atomic_types_and_method_variants}, we have that the set of signatures for methods \texttt{add}, \texttt{get} and \texttt{equald} deemed valid for either 
\texttt{BaseList\_na} and \texttt{\atomicS BaseList\_na} is also the one above.
Ergo, given that the latter set is included in the former, the \texttt{Collection} interface is correctly implemented by  both 
\texttt{BaseList\_na} and \texttt{\atomicS BaseList\_na}.

\subsubsection{Stage 4 - Code Generation}\label{sec:code_generation}

The purpose of the final code generation stage is five-fold: 
1 - replace original method signatures in interfaces
by the  ones resulting from the parsing of the atomicity annotations;
2 - replace the method definitions in classes by the  variants considered valid in the solution generated by Stage 2 (\func{validVariants});
3 - replace the method names by the right variant in all method calls;
and 
4 - types and qualifiers by new types from a set that is implicitly partitioned into atomic and no-atomically qualified types.
The resulting code  will be cleared from all ambiguities with regard to variable atomicity. %

\begin{table}[t]
\caption{Code generation}
\label{tab:rc3_variant_injection}
\begin{rulesize}	
\hrule
\begin{align*}
\vi{P}(\sol, %
\Ids\ \Cds\ e)  =\ & \bigcupplus_{\Id \in \Ids}
\vi{I}(Id) \ \uplus \bigcupplus_{\Cd \in \Cds} \vi{C}(\sol, %
Cd)\ \uplus \vi{E}(\sol, \maine, e) \\
\vi{I}(\kw{interface}\ I\ \{ \IMsigs \})  =\ & 
	\bigcupplus_{\nu \in \{\atomick,\natomick\}}
	\kw{interface}\ \ootype(\nu, I)\ \{\bigcupplus_
	{ %
	\IMsig \in \IMsigs
    }
	 \oosig (\any(\IMsig))  \} \\
\vi{I}(\kw{interface}\  I\ \extendsk\  I_1, I_2)  =\ & 
	\bigcupplus_{\nu \in \{\atomick,\natomick\}} \kw{interface}\  \ootype(\nu, I)\ \extendsk\   \ootype(\nu, I_1),  \ootype(\nu, I_2)  \\ 
\vi{C}(\sol, %
\classk\ C\  \kw{implements}\ I\ \{ \Fds\ \Mds \} )  =\ \hspace{-4em}&\hspace{4em} 
   \bigcupplus_{\nu \in \{\atomick,\natomick\}} \classk\ \ootype(\nu, C)\   \kw{implements}\ \ootype(\nu, I)\ \Big\{ \\[-7pt]
   & \hspace{10em} \bigcupplus_{ \Fd \in \Fds} \vifield(\nu,
   \Fd)  %
   ~\uplus \bigcupplus_{\Md \in \Mds} \vi{Md}(\sol, %
   \nu, C, \Md)  \Big\}   \\
\vifield(C, \nu, f : t)  =\ &
 \begin{cases}
  f: \ootype(\atomick, t)\ & \textrm{ if } \func{fatom}(C, \nu,  f)  = \atomick \\ 
   f: \ootype(\natomick, t)\ & \textrm{ otherwise }
  \end{cases} \\
\vi{Md}(\sol, %
\nu, D, \defk\ m(x : t_1) : t_2   \{ e \})  =\ & 
\bigcupplus_{\sol(\valvar(\nu.D.m.\nu_1.\nu_2))=\validk
}
\defk\ \Msig'\ \{\  \vie(\sol, \nu.D.m.\nu_1.\nu_2, e)\ \} \\
& \textrm{where } \Msig' = \oosig(\nu_1, \nu_2, m(x : t_1) : t_2) \\
\vie(\sol,  \mu, \letk\ x : t = e_1\ \ink\ e_2) =\ & \letk\ x : \ootype (\nu_1, t) =
  	\vie(\sol,  \mu, e_1)\ \ink\ \vie(\sol,  \mu, e_2)	 
  	 \\ & \textrm{where } 
  		\natx = \natvar(x) \text{ and } \nu_1 = \sol(\natx@\mu) \\
\vie(\sol,  \mu, x.m(e)) =\ & x.m'(\vie(\sol,  \mu, e)) \
	\\ &
	\begin{aligned}
	    \textrm{where } &
			\natx_1 = \nextvar; 
			\natx_2 = \nextvar, 
			\nu_1 = \sol(\natx_1@\mu),
			\nu_2 = \sol(\natx_2@\mu) \\
	    \textrm{and } &
	        m' = \oomn(m,\nu_1, \nu_2)
	\end{aligned}
\\
\vie(\sol,  \mu, \newk\ C) =\ & \newk\ C'  \quad \textrm{where } \natx = \nextvar, \nu = \sol(\natx@\mu) \text{ and } C' = \ootype(\nu, C)\\
\vie(\sol,  \mu, \newk\ \atomicS\ C) =\ & \newk\ C' \quad \textrm{where } C' = \ootype(\atomick, C)\\
\vie(\sol,  \mu, x.f = e) =\ & x.f = \vie(\sol,  \mu, e) \\
\vie(\sol,  \mu,  (t)\ e) =\ & (t')\ \vie(\sol, \mu, e) \\
 & 
    \textrm{where } \natx = \nextvar, \nu = \sol(\natx@\mu)
    \text{ and } t' = \ootype(\nu, t)
 \\
\vie(\sol,   \mu, \kw{finish} \{
    \kw{async} \{ e'_1 \} &\ 
    \kw{async} \{ e'_2 \}
    \}; e ) \\
    =\ & %
 \kw{finish} \ \{
   \kw{async} \{ \vie(\sol, \mu, e'_1) \} %
   \ \kw{async} \{ \vie(\sol, \mu, e'_2) \}
   \};\  \vie(\sol,   \mu, e) \\
 \vie(\sol,   \mu, x.f) =\ & x.f \\
 \vie(\sol,   \mu, x) =\ & x \\
 \vie(\sol, \mu, v) =\ & v 
 \end{align*}
\hrule
\end{rulesize}
\end{table}

The \vi{P} code generation function is presented in Table \ref{tab:rc3_variant_injection}.
It is parametric on the solution produced by Stage 2 and on the source program obtain in Stage 1, 
and makes use of functions \ootype, \oosig~and \oomn, defined in Appendix~\ref{app:aux}.
These functions generate \oolong identifiers for the types and method names to include in the final program, namely
\ootype\ generates  type identifiers for pairs (atomicity,  type),
\oosig\ generates  method signature for triples (atomicity, atomicity, method signature) and
\oomn\ generates  method names for triples (atomicity, atomicity, method name).

The \vi{P} function processes the definitions of every type $t$ (interface or class) from  the input program, and generates the two correspondent types \ootype($\natomick$, $t$) and  \ootype($\atomick$, $t$).
In the case of interfaces, the signatures that compose the type are 
obtained from function \any\ that parses the programmers signature annotations.
Concerning classes, the atomicity of fields is defined by function \func{fatom} (see \S\ref{sec:valid_variants})
and the method definition list is obtained by consulting the valid variants of the class' methods in \sol.

The generation of expressions, for the methods' body and the \textit{main} expression, is given by function \vi{E}.
The case for \letk requires the retrieval of the atomicity of type $t$ in the expression.
For that, we query the solution for the atomicity value assigned to  atomicity variable $\natx$, obtained (as
in Table~\ref{tab:ni}) from the \natvar\ generator. %
The cases for method calls, \newk and casts also require information stored in the solution. In these cases (also as in Table~\ref{tab:ni})
the atomicity variables are obtained from the \nextvar\ generator.
In the particular case of method calls, the identifier of the method to call is replaced by the one of the correct variant, \text{i.e.} the one with the signature compatible with the atomicities of the parameter
and of the return value. This identifier is generated from the atomicity values computed in the solution for the atomicity variables created in rule \rulename{ai\_call} from Table~\ref{tab:ni}.
The remainder rules are straightforward code translations.

We note that this final stage does not fail due to the fact that, when passing the solution \sol\ to \vi{E}, we know that \sol\ is defined on all the variables generated by \func{natVar} and \func{nextVar} and tagged with a valid variant $\nu.D.m.v_1.v_2$ (\cf \vi{Md}).

\section{Soundness}
\label{sec:soundness}

The algorithm presented in the previous section is elaborate enough to require a rigorous statement, and proof, of soundness. In this section %
we show that the $\func{\rccc{}Analysis}$ process guarantees type and behavioural soundness, \ie that (when successful) it preserves typeability and atomicity annotations, and moreover, it produces a program behaviourally equivalent to the original one.

The type soundness of the automatically generated program has been proved in Coq. The definitions and proofs weight about 18 kloc, for a total of 297 definitions, 310 lemmas, one theorem (type preservation for programs), and six axioms about unimplemented aspects; the axiomatisation of the constraint solver relies on standard assumptions relating the input to the generated solution 
(see file README in the mechanisation\footnote{AtomiS-Coq proof of type preservation (2022), URL:\url{https://zenodo.org/record/6346649} and \url{https://zenodo.org/record/6382015}} artefact for more details). The proof uses standard tactics 
from Coq's library.

The following definitions and results%
are stated for programs $P_\textrm{orig}$, \ie annotated \rccc-\oolong programs written in the language of Table~\ref{tab:syntaxpos}, for which the analysis succeeds and produces a final \oolong program, according to Definition~\ref{def:rc3-analysis}. 

\subsection{Type Soundness and Mechanisation}
We start by formalising how programs that are generated by the analysis are typeable by the \oolong type system, while attributing the same base types to fields and methods, and atomicities are consistent with those annotated in the original code. To lighten the notation, in the following we say that %
type $\ootype(\nu,t)$ is a compound of atomicity $\nu$ and type $t$.

\subsubsection{Type Soundness}
\label{sec:ts}
The Preservation of Base Types theorem states that the \rccc{} analysis transforms programs that,   when stripped of atomicity annotations, are typable with a certain type, into (\oolong) programs that are also typeable with a type that is a compound of the former.
It thus assumes that $\func{\rccc{}Analysis}$ succeeds, and typability respects the \oolong type system, which we denote by $\vdash_{\textit{Ool}}$.
\begin{theorem}[Preservation of Base Types] \label{th:preservation}
Consider an original program $P_\textrm{orig}$ %
such that

$\func{\rccc{}Analysis}(P_\textrm{orig}) = (P,\sol,P^+)$.
If $\func{strip}(P_\textrm{orig})=S$ is typeable with type $t$, then the final (\oolong) program $P^+$ is also typeable with a type that is a compound of $t$:
for all $t$ such that $\vdash_{\textit{Ool}} S : t$, there exists $\nu$, such that $\vdash_{\textit{Ool}} P^+ : \ootype(\nu, t)$.
\end{theorem}
\off{\begin{theorem}[Preservation of Base Types]\label{th:preservation} %
Consider an original program $P_\textrm{orig}$ %
such that  $\func{\rccc{}Analysis}(P_\textrm{orig}) = (P,\_,P^+)$.
If $\func{strip}(P_\textrm{orig})$ is typeable with type $t$, then the final (\oolong) program $P^+$ is also typeable with a type that is a compound of $t$.
\end{theorem}
}

The Consistency of Types and Atomicity Inference theorem states that the program transformation performed by the AtomiS analysis produces
programs whose field and signature types are consistent with those of the original \rccc annotations.  More specifically, the field types of the final program are a compound of the atomicity and types of the fields given by the original program to fields with the same name in corresponding classes, and the signature types of methods in the final program are a compound of the atomicity and types of the methods given by the original program to methods with the corresponding name in corresponding classes.
\begin{theorem}[Consistency of Types and Atomicity Inference] \label{th:consistency}
Consider an original program $P_\textrm{orig}$, that is annotated according to the \rccc model such that  $\func{\rccc{}Analysis}(P_\textrm{orig}) = (P,\sol,P^+)$.
The types of the final program $P^+$, of its signatures and of its fields are generated from a compound version of those in $P$ that is consistent with $P_\textrm{orig}$'s atomicity annotations, i.e.:
    \begin{enumerate}
    \item for all $C^+,f,t^+$ such that $\ottkw{fields}_{P^+}(C^+)(f) = t^+$, there exist $C,\nu,t$ such that $C^+=\ootype(\nu,C)$, and either
    \begin{enumerate}
        \item $\ottkw{fields}_{P_\textrm{orig}}(C)(f) = t$ and $t^+=\ootype(\natomick,t)$, or
        \item  $\ottkw{fields}_{P_\textrm{orig}}(C)(f) = \atomicS~t$ and $t^+=\ootype(\atomick,t)$;
    \end{enumerate}%
    \item for all $t^+,m^+,t_1^+,t_2^+$ s.t.  $\ottkw{msigs}_{P^+}(t^+)(m^+) = x : t_1^+ \rightarrow t_2^+$, there exist $t, m, q_1, t_1, q_2, t_2, \nu$ s.t. $\ottkw{msigs}_{P_\textrm{orig}}(t)(m) =  x : q_1~t_1 \rightarrow q_2~t_2$ and $t^+=\ootype(\nu,t)$, with $t_1^+=$ $\ootype(\func{atom}(q_1),t_1)$, and $t_2^+=\ootype(\func{atom}(q_2),t_2)$ and also $m^+ = \oomn(m,\func{atom}(q_1),\func{atom}(q_2))$.
    \end{enumerate}
\end{theorem}

\off{
\begin{theorem}[Consistency of Types with Atomicity Annotations]\label{th:consistency}
Consider an original program $P_\textrm{orig}$, that is annotated according to the \rccc model such that\linebreak  $\func{\rccc{}Analysis}(P_\textrm{orig}) = (P,\_,P^+)$.
The field and signature's types of the final program $P^+$ are consistent with those in $P_\textrm{orig}$, i.e.:
\begin{enumerate}
    \item types given by $\ottkw{fields}_{P^+}$ to fields in $P^+$ are a compound of the atomicity and types of the fields given by $\ottkw{fields}_{P_\textrm{orig}}$ to fields with the same name in corresponding classes of $P_\textrm{orig}$.
    \item types given by $\ottkw{msigs}_{P^+}$ to methods in $P^+$ are a compound of the atomicity and types of the methods given by $\ottkw{msigs}_{P_\textrm{orig}}$ to methods with the corresponding name in corresponding classes of $P_\textrm{orig}$.
\end{enumerate}
\end{theorem}
}

Note that the Progress and Preservation results that hold for the \oolong language and type system ensure that \emph{the output of the analysis never goes wrong} in what regards both base types and atomicities.

\subsubsection{Proof Mechanisation}
\label{sec:mech}

The proofs are based on a refinement of the \oolong type system, that represents an internal step of the analysis, where types are formalised as pairs of atomicity qualifiers and base types.
The definition of the source language and typing system largely 
reuses code in ~\cite{oolong-ACR2019}, while we narrowed programs %
in order to avoid repetitions: 

\begin{lstlisting}[language=coq,basicstyle=\ttfamily\footnotesize,numbers=none,breaklines=true]
Definition program := ( { cds : list classDecl | NoDup (map pclassName cds) /\ Forall NoDupMethods cds } * { ids : list interfaceDecl | NoDup (map pinterfaceName ids) } * expr) %
\end{lstlisting}

The constraint generator in Table~\ref{tab:ni} has been implemented as an inductive type, denoted \texttt {programConstraints}, mimicking a partial function: the definition relates the program and the variables received in input to a domain of Variant IDs and a partial map produced in output;  the map, or 
\textit{ID  environment} (denoted type \texttt {muvarEnv} ), associates  the domain's IDs to pairs $(\nes,\mns)$, while variables are used to implement the ``freshness'' mechanism required by \nextvar: 

\begin{lstlisting}[language=coq,basicstyle=\ttfamily\footnotesize,numbers=none,breaklines=true]
Check programConstraints: program -> list hatVar -> list nat -> MethodVarSystems -> muvarEnv -> Prop
\end{lstlisting}

The domain and the ID environment are passed to \texttt{variantConstraints} in order to produce the constraint system to be passed to the solver, which in turn relies on \texttt {vvMcs} to analyse all possible atomicity combinations  $(v_1,v_2,v_3)$ and generate entries of the form $\valvar(v_1.t.m.v_2.v_3)= \validk \rightarrow \eta$, where $t$  and $m$ are the (cast of the) type and the method name of \texttt{muv}, respectively, and $\eta$ is a conjunction of pairs qualified with $\mu$:

\begin{lstlisting}[language=coq,basicstyle=\ttfamily\footnotesize,numbers=none,breaklines=true]
Definition variantConstraints (valVar : valVarT) (domain : MethodVarSystems) (mue : muvarEnv) : list NatureConstraintSystems := map (fun muv => match mue muv with | Some (acs, vns) => vvMcs valVar muv acs vns | _ => nil end) domain.    
\end{lstlisting}

The Coq formulation of type preservation is stated below, where for short we omit some (minor) context hypotheses. It ensures that if $P$ has type $t$ (that is the type of the main expression), then the  generated program $P'$ has type (corresponding to the non-atomic version of) $t$, where \texttt{valVar} is a partial map from  IDs to variables, and \texttt{fqmMain} is the ID reserved for the main:
       
\begin{lstlisting}[language=coq,basicstyle=\ttfamily\footnotesize,numbers=none,breaklines=true]
Check T_preserve:forall..., wfProgram P t ->...->programConstraints P fv sn domain mue ->makeV P = (domV, valVar) ->variantConstraints valVar domain mue = acs ->solver (concat acs) = sol ->valVar fqmMain = Some valvar ->sol valvar = valid ->viP P sol valVar fv = Some P' ->wfProgramOO P' (setNonAtomic (castT t))
\end{lstlisting}

\subsection{Behavioural Soundness}
\label{sec:bis}
We have seen that our approach performs a program transformation that consistently fleshes out the atomicity qualifiers of every type in the program, while unfolding classes and methods according to their determined qualified types. It remains to assert that the transformation does not affect the original behaviour, \ie that the final program does everything the original one does, and nothing more, according to a notion of indistinguishability. Intuitively, we wish to consider the correspondence between types and method names occurring in the final program, and those from which they originated in the original code.

Bisimulations are often used to relate pairs of concurrent programs that exhibit the same behaviour according to some criteria, step-by-step, by establishing a full correspondence between the possible outcomes of both programs in such a way that they will also simulate each other. In our case we need a relation between thread collections that is based on the notion of heap correspondence. To establish it, we define \emph{syntactic correspondences} between 
\emph{thread collections} and between 
\emph{program contexts}, understood as the list of interfaces and classes of a program.
These correspondences are defined as pairs of
mappings between types and between method names that are used in programs and configurations:
\begin{definition}[Expression and Thread Collection Correspondence] Given a pair

$\kw{MAP} = (\kw{methodMap},\kw{typeMap})$ of maps, where\newline $\kw{methodMap}:\methodN\rightarrow\methodN$
and $\kw{typeMap}:\Type\rightarrow\Type$, and two program contexts $PC_1$ and $PC_2$,
we say that there is a syntactic correspondence between:
\begin{enumerate}
\item expressions $e_1$ and $e_2$, well-formed with respect to $PC_1$ and $PC_2$ respectively, if $\vdash_{\text{MAP}}^{\text{PC}_1,\text{PC}_2}  \ottnt{e_1} \propto \ottnt{e_2}$
    \item 
thread collections $T_1$ and $T_2$, well-formed with respect to $PC_1$ and $PC_2$ respectively, if $\vdash_{\text{MAP}}^{\text{PC}_1,\text{PC}_2}  \ottnt{T_1} \propto \ottnt{T_2}$
\end{enumerate}
according to the rules in Table~\ref{fig:type-expr}.
\end{definition}

\newcommand{\ottdrulewfSCprogramContext}[1]{
\ottdrule[#1]{
\ottpremise{PC_1 = \ottnt{Ids_1} ~ \ottnt{Cds_1} ~ \text{[]}, PC_2 = \ottnt{Ids_2} ~  \ottnt{Cds_2} ~ \text{[]}
}
\ottpremise{\exists \text{MAP}=\text{(methodMap, typeMap) such that:}
}
\ottpremise{\qquad\text{methodMap() from sets of method names in Ids2 U Cds2 to method names in Ids1 U Cds1}
}
\ottpremise{\qquad\text{typeMap() from sets of types in PC2 to types in PC1}
}
\ottpremise{
	Ids_1 = \bigcup_{Id_2 \in Ids_2} ~ \{Id_1 ~|~ 
	\vdash_{\text{MAP}}^{\text{PC}_1,\text{PC}_2}  \ottnt{Id_1} \propto \ottnt{Id_2}\}
}
\ottpremise{
	Cds_1 = \bigcup_{Cd_2 \in Cds_2} ~ \{Cd_1 ~|~ 
	\vdash_{\text{MAP}}^{\text{PC}_1,\text{PC}_2}  \ottnt{Cd_1}  \propto \ottnt{Cd_2}\}
}
}%
{
	\vdash_{\text{MAP}}  PC_1 \propto PC_2 ~ \text{[]}
}
{\ottdrulename{sc\_program\_context}}{}
}%

\newcommand{\ottdrulewfSCprogram}[1]{
\ottdrule[#1]{%
\ottpremise{\vdash_{\text{MAP}}  \ottnt{Ids_1} ~ \ottnt{Cds_1} ~ \text{[]} \propto \ottnt{Ids_2} ~ \ottnt{Cds_2} ~ \text{[]}}
	\ottpremise{
    \vdash_{\text{MAP}}^{\text{PC}_1,\text{PC}_2}  \ottnt{e_1} \propto \ottnt{e_2}
	}
}%
{    \vdash_{\text{MAP}}^{\text{PC}_1,\text{PC}_2}  \ottnt{Ids_1} ~ \ottnt{Cds_1} ~ \ottnt{e_1} \propto \ottnt{Ids_2} ~ \ottnt{Cds_2} ~ \ottnt{e_2}
}  
{\ottdrulename{sc\_program}}{}%
}

\newcommand{\ottdefnwfSCProgramContext}[1]{
\begin{ottdefnblock}[#1]{$\vdash  \ottnt{PC}_1 \propto \ottnt{PC}_2$}{}
\ottusedrule{\ottdrulewfSCprogramContext{}}
\end{ottdefnblock}
}

\newcommand{\ottdefnwfSCProgram}[1]{
\begin{ottdefnblock}[#1]{$\vdash \ottnt{P_1} \propto \ottnt{P_2}$}{}
\ottusedrule{\ottdrulewfSCprogram{}}
\end{ottdefnblock}
}

\newcommand{\ottdrulewfSCinterface}[1]{\ottdrule[#1]{%
\ottpremise{
\ottnt{Msigs_1} = \bigcup_{\forall \, \ottmv{m_2}  \ottsym{(}  \ottmv{x_2}  \ottsym{:}  \ottnt{s_2}  \ottsym{)}  \ottsym{:}  \ottnt{s_2'} \, \!\!\in\! \, 
\ottnt{Msigs_2}} 
~\{\text{methodMap}(\ottmv{m_2})  \ottsym{(} \ottmv{x_2}  \ottsym{:}  \text{typeMap}(\ottnt{s_2})  \ottsym{)}  \ottsym{:}  \text{typeMap}(\ottnt{s_2'}) \, \!\!\in\! \, 
\ottnt{Msigs_1}\} 
}
}{
\vdash_{\text{MAP}}^{\text{PC}_1,\text{PC}_2}  \kw{interface} \, \text{typeMap}(I_2)  \,\startblock\, \ottnt{Msigs_1} \,\finishblock\, \propto \kw{interface} \, I_2 \,\startblock\, \ottnt{Msigs_2} \,\finishblock\,}{%
{\ottdrulename{sc\_interface}}{}%
}
}

\newcommand{\ottdrulewfSCinterfaceORextends}[1]{\ottdrule[#1]{%
}{
\vdash_{\text{MAP}}^{\text{PC}_1,\text{PC}_2}  \kw{interface} \, \text{typeMap}(I_2) \, \kw{extends} \, \text{typeMap}(I_2')  \ottsym{,}  \text{typeMap}(\ottmv{I_2}'')
\propto
\kw{interface} \, I_2 \, \kw{extends} \, I_2'  \ottsym{,}  \ottmv{I_2''}}{%
{\ottdrulename{sc\_interface\_extends}}{}%
}
}

\newcommand{\ottdefnwfSCInterface}[1]{
\begin{ottdefnblock}[#1]{$\vdash  \ottnt{Id_1} \propto \ottnt{Id_2}$}{}
\ottusedrule{\ottdrulewfSCinterface{}}\\[5pt]
\ottusedrule{\ottdrulewfSCinterfaceORextends{}}
\end{ottdefnblock}
}

\newcommand{\ottdrulewfSCclass}[1]{\ottdrule[#1]{%
\ottpremise{ \forall \, \ottmv{m_2}  \ottsym{(}  \ottmv{x_2}  \ottsym{:}  \ottnt{s_2}  \ottsym{)}  \ottsym{:}  \ottnt{s_2'} \, \!\!\in\! \, \ottkw{msigs}_2 \, \ottsym{(}  \ottmv{J_2}  \ottsym{)}
~\ottsym{.} ~ 
 \kwool{def} \, m_2  \ottsym{(}  \ottmv{x_2}  \ottsym{:}  s_2  \ottsym{)}  \ottsym{:}  s_2'\, \,\startblock\, e_2 \,\finishblock\, \!\!\in\! \, \ottnt{Mds}_2 \ottlinebreakhack }
\ottpremise{\ottnt{Fds}_1 = \bigcup_{\ottnt{Fd}_2 \, \!\!\in\! \, \ottnt{Fds}_2}  ~ \{  \ottnt{Fd}_1 ~|~  \vdash_{\text{MAP}}^{\text{PC}_1,\text{PC}_2} \ottnt{Fd}_1 \propto \ottnt{Fd}_2 \}}
\ottpremise{\ottnt{Mds}_1 = \bigcup_{\ottnt{Md}_2 \, \!\!\in\! \, \ottnt{Mds}_2}  ~ \{  \ottnt{Md}_1 ~|~   \kwool{this}  \ottsym{:}  C_2 \vdash_{\text{MAP}}^{\text{PC}_1,\text{PC}_2} \ottnt{Md}_1 \propto \ottnt{Md}_2 \}}
}{
\vdash_{\text{MAP}}^{\text{PC}_1,\text{PC}_2}  \kwool{class} \, \text{typeMap}(C_2) \, \kwool{implements} \,  \text{typeMap}(I_2) \,\startblock\,  \ottnt{Fds}_1 \, \mathit{Mds}_1  \,\finishblock\, 
\propto
\kwool{class} \, C_2 \, \kwool{implements} \,  I_2 \,\startblock\,  \ottnt{Fds_2} \, \mathit{Mds_2}  \,\finishblock\,  }{
{\ottdrulename{sc\_class}}{}%
}
}

\newcommand{\ottdefnwfSCClass}[1]{
\begin{ottdefnblock}[#1]{$\vdash  \ottnt{Cd_1} \propto \ottnt{Cd_2}$}{}
\ottusedrule{\ottdrulewfSCclass{}}
\end{ottdefnblock}
}

\newcommand{\ottdrulewfSCfield}[1]{\ottdrule[#1]{%
}{
\vdash_{\text{MAP}}^{\text{PC}_1,\text{PC}_2}  \ottmv{f_2}  \ottsym{:}  \text{typeMap}(s) \propto f_2 \ottsym{:} s}{%
{\ottdrulename{sc\_field}}{}%
}}

\newcommand{\ottdefnwfSCField}[1]{\begin{ottdefnblock}[#1]{$\vdash_{\text{MAP}}^{\text{PC}_1,\text{PC}_2}  \ottnt{Fd}$}{}
\ottusedrule{\ottdrulewfSCfield{}}
\end{ottdefnblock}}

\newcommand{\ottdrulewfSCmethod}[1]{\ottdrule[#1]{%
 \vdash_{\text{MAP}}^{\text{PC}_1,\text{PC}_2}  \ottnt{e_1}  %
 \propto 
 \ottnt{e_2}  %
  }{
\vdash_{\text{MAP}}^{\text{PC}_1,\text{PC}_2}  \kwool{def} \, \text{methodMap}(m_2)  \ottsym{(}  \ottmv{x_2} \ottsym{)}  \ottsym{:}  \text{typeMap}(s_2)     \,\startblock\,  \ottnt{e_1} \,\finishblock\, \ottsym{:}  \text{typeMap}(s_2')
\propto
 \kwool{def} \, m_2 \ottsym{(}\ottmv{x_2}\ottsym{)}  \ottsym{:}  s_2   \,\startblock\,  \ottnt{e_2} \,\finishblock\, %
}{%
{\ottdrulename{sc\_method}}{}%
}
}

\newcommand{\ottdefnwfSCMethod}[1]{\begin{ottdefnblock}[#1]{ $%
\vdash_{\text{MAP}}^{\text{PC}_1,\text{PC}_2} \Md_1 \propto \Md_2$}{}
\ottusedrule{ \ottdrulewfSCmethod{}
}
\end{ottdefnblock}}

\off{
\begin{table*}
\caption{Syntactic Correspondence for programs.}  %
\label{fig:type-classesandinterfaces}
\begin{scriptsize}
\hrule
\ottdefnwfSCProgramContext{}\\[7pt]
\ottdefnwfSCProgram{}\\[7pt]
\ottdefnwfSCInterface{}\\[7pt]
\ottdefnwfSCClass{}\\[7pt]
\ottdefnwfSCField{}\\[7pt]
\ottdefnwfSCMethod{}\\[7pt]
\end{scriptsize}
\end{table*}
}

\newcommand{\ottdrulewfSCConfiguration}[1]{\ottdrule[#1]{%
 \vdash_{\text{MAP}}^{\text{PC}_1,\text{PC}_2}  \ottnt{H_1}  \propto \ottnt{H_2} 
 \qquad
 \vdash_{\text{MAP}}^{\text{PC}_1,\text{PC}_2}  \ottnt{T_1}  \propto \ottnt{T_2} 
  }{
\vdash_{\text{MAP}}^{\text{PC}_1,\text{PC}_2}  
 \cfg{H_1}{V}{T_1}
\propto
 \cfg{H_2}{V}{T_2}
}{%
{\ottdrulename{sc\_configuration}}{}%
}
}

\newcommand{\ottdefnwfSCConfiguration}[1]{\begin{ottdefnblock}[#1]{ $\Gamma \vdash_{\text{MAP}}^{\text{PC}_1,\text{PC}_2} \Md_1 \propto \Md_2$}{}
\ottusedrule{ \ottdrulewfSCConfiguration{}
}
\end{ottdefnblock}}

\newcommand{\ottdrulewfSCHeap}[1]{\ottdrule[#1]{%
 \dom(H_1)=\dom(H_2) 
 \qquad
 \forall \iota \in \dom(H_1) ~.~
H_2(\iota)=(C_2,F,L) \Leftrightarrow H_1(\iota)=(\text{typeMap}(C_2),F,L)

  }{
\vdash_{\text{MAP}}^{\text{PC}_1,\text{PC}_2}  
H_1
\propto
H_2
}{%
{\ottdrulename{sc\_heap}}{}%
}
}

\newcommand{\ottdefnwfSCHeap}[1]{\begin{ottdefnblock}[#1]{ $\Gamma \vdash_{\text{MAP}}^{\text{PC}_1,\text{PC}_2} H_1 \propto H_2$}{}
\ottusedrule{ \ottdrulewfSCHeap{}
}
\end{ottdefnblock}}

\newcommand{\ottdrulewfSCAsync}[1]{\ottdrule[#1]{%
\vdash_{\text{MAP}}^{\text{PC}_1,\text{PC}_2}  
\ottnt{T_{{\mathrm{1}}}}
\propto
\ottnt{T_{{\mathrm{2}}}} 
\qquad
\vdash_{\text{MAP}}^{\text{PC}_1,\text{PC}_2}  
\ottnt{T_{{\mathrm{1}}}}'
\propto
\ottnt{T_{{\mathrm{2}}}}'
\qquad
\vdash_{\text{MAP}}^{\text{PC}_1,\text{PC}_2}  
\ottnt{e_{{\mathrm{1}}}}
\propto
\ottnt{e_{{\mathrm{2}}}} 
  }{
\vdash_{\text{MAP}}^{\text{PC}_1,\text{PC}_2}  
\ottnt{T_{{\mathrm{1}}}}  \mathop{||}  \ottnt{T_{{\mathrm{1}}}}'  \rhd  \ottnt{e_1}
\propto
\ottnt{T_{{\mathrm{2}}}}  \mathop{||}  \ottnt{T_{{\mathrm{2}}}}'  \rhd  \ottnt{e_2}
}{%
{\ottdrulename{sc\_async}}{}%
}
}

\newcommand{\ottdefnwfSCAsync}[1]{\begin{ottdefnblock}[#1]{ $\Gamma \vdash_{\text{MAP}}^{\text{PC}_1,\text{PC}_2} T_1 \propto T_2$}{}
\ottusedrule{ \ottdrulewfSCAsync{}}
\ottusedrule{ \ottdrulewfSCException{} \qquad
\ottdrulewfSCThread{}
}
\end{ottdefnblock}}

\newcommand{\ottdrulewfSCThread}[1]{\ottdrule[#1]{%
\vdash_{\text{MAP}}^{\text{PC}_1,\text{PC}_2}  
\ottnt{e_{{\mathrm{1}}}}
\propto
\ottnt{e_{{\mathrm{2}}}} 
  }{
\vdash_{\text{MAP}}^{\text{PC}_1,\text{PC}_2}  
\ottsym{(}  \mathcal{L}  \ottsym{,}  \ottnt{e_1}  \ottsym{)} 
\propto
\ottsym{(}  \mathcal{L}  \ottsym{,}  \ottnt{e_2}  \ottsym{)}
}{%
{\ottdrulename{sc\_thread}}{}%
}
}

\newcommand{\ottdrulewfSCException}[1]{\ottdrule[#1]{%
}{
\vdash_{\text{MAP}}^{\text{PC}_1,\text{PC}_2}  
\textbf{EXN}
\propto
\textbf{EXN}
}{%
{\ottdrulename{sc\_exn}}{}%
}
}

\off{
\begin{table*}
\caption{Syntactic Correspondence for configurations.}  %
\label{fig:type-classesandinterfaces}
\begin{scriptsize}
\hrule
\ottdefnwfSCConfiguration{}\\[7pt]
\ottdefnwfSCHeap{}\\[7pt]
\ottdefnwfSCAsync{}\\[7pt]
\end{scriptsize}
\end{table*}
}

\newcommand{\ottdrulewfSCvar}[1]{\ottdrule[#1]{%
}{
	\vdash_{\text{MAP}}^{\text{PC}_1,\text{PC}_2}  \ottmv{x}  %
	\propto \ottmv{x}  %
}{%
{\ottdrulename{sc\_var}}{}%
}}

\newcommand{\ottdrulewfSCvarBase}[1]{\ottdrule[#1]{%
\ottpremise{ \vdash  \Gamma  \!\!\! }%
\ottpremise{ \Gamma  \ottsym{(}  \ottmv{x}  \ottsym{)}  \ottsym{=}  \ottnt{t}  \!\!\! }%
\ottpremise{ \regularV(\ottnt{t})=\intk  \!\!\! }%
}{
\Gamma  \vdash  \ottmv{x}  \ottsym{:}  \ottmv{(t,\bot)} \outputeqs{\emptyset}}{%
{\ottdrulename{sc\_varbase}}{}%
}}

\newcommand{\ottdrulewfSCbase}[1]{\ottdrule[#1]{%
\ottpremise{ \vdash  \Gamma  \!\!\! }%
\ottpremise{ n \in\intk\!\!\! }%
}{
\Gamma  \vdash  \ottmv{n}  \ottsym{:}  \ottmv{(\intk,\bot)} \outputeqs{\emptyset}}{%
{\ottdrulename{sc\_base}}{}%
}}

\newcommand{\ottdrulewfSCintOp}[1]{\ottdrule[#1]{%
\ottpremise{ \Gamma  \vdash  \ottnt{e_{{\mathrm{1}}}}  \ottsym{:}  
\ottnt{(t_{{\mathrm{1}}}, \bot)} %
\outputeqs{\nes_1}}%
\ottpremise{ \Gamma  \vdash  \ottnt{e_{{\mathrm{2}}}}  \ottsym{:}  
\ottnt{(t_{\mathrm{2}}, \bot )} %
\outputeqs{\nes_2} }
\ottpremise{ \regularV(t_{\mathrm{1}}) = \intk}
\ottpremise{ \regularV(t_{\mathrm{2}}) = \intk}
\ottpremise{\textsf{op}\in \{ +, *, -, /\}}
}{
\Gamma  \vdash  \ottmv{e_{{\mathrm{1}}}\, \textsf{op}\, e_{{\mathrm{2}}}}\  \ottsym{:}  \ottmv{(\intk, \bot)} \outputeqs{\nes_1 \cup \nes_2}}{%
{\ottdrulename{sc\_op}}{}%
}}

\newcommand{\ottdrulewfSClet}[1]{
\ottdrule[#1]{%
	\vdash_{\text{MAP}}^{\text{PC}_1,\text{PC}_2}  \ottnt{e_{{\mathrm{1}}}}  %
	\propto \ottnt{e_{{\mathrm{2}}}}  %
	\quad
	\vdash_{\text{MAP}}^{\text{PC}_1,\text{PC}_2}  \ottnt{e_1'}  %
	\propto
	\ottnt{e_2'}  %
}{
	\vdash_{\text{MAP}}^{\text{PC}_1,\text{PC}_2}   \kwool{let}\,   \ottmv{x} :\text{typeMap}(t_2)  
	=  \ottnt{e_{{\mathrm{1}}}} ~\kwool{in}~ \ottnt{e_1'}  %
	\propto
	\kwool{let}\,  \ottmv{x}:t_2 %
	=  \ottnt{e_{{\mathrm{2}}}} ~\kwool{in}~ \ottnt{e_2'}   %
}{%
{\ottdrulename{sc\_let}}{}%
}}

\newcommand{\ottdrulewfSCcall}[1]{\ottdrule[#1]{%
  \vdash_{\text{MAP}}^{\text{PC}_1,\text{PC}_2}  \ottnt{e_1}  %
  \propto \ottnt{e_2}  %
	\\
}{
	 \vdash_{\text{MAP}}^{\text{PC}_1,\text{PC}_2}  
	 x.\text{methodMap}(m_2)  \ottsym{(}  \ottnt{e_1}  \ottsym{)}
	 \propto
	 x.m_2  \ottsym{(}  \ottnt{e_2}  \ottsym{)}  %
 }{%
{\ottdrulename{sc\_call}}{}%
}}

\newcommand{\ottdrulewfSCcast}[1]{\ottdrule[#1]{%
	\vdash_{\text{MAP}}^{\text{PC}_1,\text{PC}_2}  \ottnt{e_1}  %
	\propto \ottnt{e_2} %
}{
 \vdash_{\text{MAP}}^{\text{PC}_1,\text{PC}_2}  
 \ottsym{(} \text{typeMap}(t_2) \ottsym{)}  \ottnt{e_1} 
 \propto
 \ottsym{(}  t_2  \ottsym{)}  \ottnt{e_2}  %
}{%
{\ottdrulename{sc\_cast}}{}%
}}

\newcommand{\ottdrulewfSCselect}[1]{\ottdrule[#1]{%
}{
	\vdash_{\text{MAP}}^{\text{PC}_1,\text{PC}_2}  \ottmv{x}  \ottsym{.}  \ottmv{f}  %
	\propto
	\ottmv{x}  \ottsym{.}  \ottmv{f}  %
}{%
{\ottdrulename{sc\_select}}{}%
}}

\newcommand{\ottdrulewfSCloc}[1]{\ottdrule[#1]{%
}{
	\Gamma  \vdash_{\text{MAP}}^{\text{PC}_1,\text{PC}_2}  \iota  %
	\propto
	\iota  %
}{%
{\ottdrulename{sc\_loc}}{}%
}}

\newcommand{\ottdrulewfSCsequence}[1]{\ottdrule[#1]{%
\ottpremise{ \Gamma  \vdash  \ottnt{e_{{\mathrm{1}}}}  \ottsym{:}  
\ottnt{(\kwool{Unit}}, \bot)
\outputeqs{\nes_1}}%
\ottpremise{ \Gamma  \vdash  \ottnt{e_{{\mathrm{2}}}}  \ottsym{:}  
\ottnt{(t, \eta)} 
\outputeqs{\nes_2} }
}{
\Gamma  \vdash  \ottmv{e_{{\mathrm{1}}}\, \ottsym{;}\,  e_{{\mathrm{2}}}}\  \ottsym{:}  \ottmv{(t, \eta)} \outputeqs{\nes_1 \cup \nes_2}}{%
{\ottdrulename{sc\_seq}}{}%
}}

\newcommand{\ottdrulewfSCvoid}[1]{\ottdrule[#1]{%
\ottpremise{\vdash  \Gamma}%
}{
\Gamma  \vdash  \kwool{void}  \ottsym{:}  \ottnt{(\kwool{Unit},\bot)} \outputeqs{\emptyset}}{%
{\ottdrulename{sc\_void}}{}%
}}

\newcommand{\ottdrulewfSCnull}[1]{\ottdrule[#1]{%
}{
\vdash_{\text{MAP}}^{\text{PC}_1,\text{PC}_2}  \kwool{null}  %
\propto
\kwool{null}  %
}{%
{\ottdrulename{sc\_null}}{}%
}}

\newcommand{\ottdrulewfSCupdate}[1]{\ottdrule[#1]{%
	\vdash_{\text{MAP}}^{\text{PC}_1,\text{PC}_2}  \ottnt{e_1}  %
	\propto
	\ottnt{e_2}  %
}{
\Gamma  \vdash_{\text{MAP}}^{\text{PC}_1,\text{PC}_2}  \ottmv{x}  \ottsym{.}  \ottmv{f}  \ottsym{=}  \ottnt{e_1} %
\propto
 \ottmv{x}  \ottsym{.}  \ottmv{f}  \ottsym{=}  \ottnt{e_2} %
}{%
{\ottdrulename{sc\_update}}{}%
}}

\newcommand{\ottdrulewfSCnew}[1]{\ottdrule[#1]{%
}{
	\vdash_{\text{MAP}}^{\text{PC}_1,\text{PC}_2}  
	\kwool{new} \, \text{typeMap}(C_2) %
	\propto
	\kwool{new} \, C_2 %
}{%
	{\ottdrulename{sc\_new}}{}%
}}

\newcommand{\ottdrulewfSCfj}[1]{\ottdrule[#1]{%
\ottpremise{%
\vdash_{\text{MAP}}^{\text{PC}_1,\text{PC}_2}  \ottnt{e_{{\mathrm{1}}}}' 
\propto
\ottnt{e_{{\mathrm{2}}}}'
\qquad
\vdash_{\text{MAP}}^{\text{PC}_1,\text{PC}_2}  \ottnt{e_{{\mathrm{1}}}}''
\propto
  \ottnt{e_{{\mathrm{2}}}}''
\qquad%
\vdash  \ottnt{e_1}
\propto e_2
}%
}{
\vdash_{\text{MAP}}^{\text{PC}_1,\text{PC}_2}  \kwool{finish} \, \,\startblock\, \, \kwool{async} \, \,\startblock\,  \ottnt{e_{{\mathrm{1}}}}'  \,\finishblock\, \, \kwool{async} \, \,\startblock\,  \ottnt{e_{{\mathrm{1}}}}''  \,\finishblock\,  \,\finishblock\,  \ottsym{;}  \ottnt{e_1} 
\propto
\kwool{finish} \, \,\startblock\, \, \kwool{async} \, \,\startblock\,  \ottnt{e_{{\mathrm{2}}}}'  \,\finishblock\, \, \kwool{async} \, \,\startblock\,  \ottnt{e_{{\mathrm{2}}}}''  \,\finishblock\,  \,\finishblock\,  \ottsym{;}  \ottnt{e_2} 
}{%
	{\ottdrulename{sc\_spawn}}{}%
}}

\newcommand{\ottdrulewfSClock}[1]{\ottdrule[#1]{%
\ottpremise{\Gamma  \vdash  \ottmv{x}  \ottsym{:}  \ottnt{(t_{{\mathrm{2}}},\eta_{{\mathrm{2}})}} \outputeqs{\nes_{{\mathrm{2}}}}}%
\ottpremise{ \Gamma  \vdash  \ottnt{e}  \ottsym{:}  \ottnt{(t,\eta)} \outputeqs{\nes} \!\!\! }%
}{
\Gamma  \vdash   \kwool{lock} ( \ottmv{x} ) \kwool{\,in\,}  \ottnt{e}   \ottsym{:}  \ottnt{(t,\eta)} \outputeqs{\nes}}{%
{\ottdrulename{sc\_lock}}{}%
}}

\newcommand{\ottdrulewfSClocked}[1]{\ottdrule[#1]{%
\ottpremise{\Gamma  \vdash  \ottnt{e}  \ottsym{:}  \ottnt{(t,\eta)}\outputeqs{\nes}}%
\ottpremise{\Gamma  \ottsym{(}  \iota  \ottsym{)}  \ottsym{=}  \ottnt{t_{{\mathrm{2}}}}}%
}{
\Gamma  \vdash   \kwool{locked}_{ \iota } \{  \ottnt{e}  \}   \ottsym{:}  \ottnt{(t,\eta)}\outputeqs{\nes}}{%
{\ottdrulename{sc\_locked}}{}%
}}

\newcommand{\ottdrulewfSCif}[1]{\ottdrule[#1]{%
\ottpremise{\Gamma  \vdash  \ottnt{e_{{\mathrm{1}}}}  \ottsym{:}  \ottnt{(t_1,\eta_1)} \outputeqs{\nes_1}}%
\ottpremise{\Gamma  \vdash  \ottnt{e_{{\mathrm{2}}}}  \ottsym{:}  \ottnt{(t_{{\mathrm{1}}},\eta_{{\mathrm{2}}})} \outputeqs{\nes_{{\mathrm{2}}}}}%
\ottpremise{\Gamma  \vdash  \ottnt{e_{{\mathrm{3}}}}  \ottsym{:}  \ottnt{(t_{{\mathrm{3}}},\eta_{{\mathrm{3}}})} \outputeqs{\nes_{{\mathrm{3}}}}}%
\ottpremise{\Gamma  \vdash  \ottnt{e_{{\mathrm{4}}}}  \ottsym{:}  \ottnt{(t_{{\mathrm{3}}},\eta_{{\mathrm{4}}})} \outputeqs{\nes_{{\mathrm{4}}}}}%
}{
\Gamma  \vdash   \ifk\ (e_{{\mathrm{1}}} == e_{{\mathrm{2}}})\ e_{{\mathrm{3}}}\ \elsek\ e_{{\mathrm{4}}}   \ottsym{:}  \ottnt{(t_3,\eta_\mathrm{3})} 
\outputeqs{\nes_{{\mathrm{1}}} \cup \nes_{{\mathrm{2}}} \cup \nes_{{\mathrm{3}}} \cup \nes_{{\mathrm{4}}}  \cup \{(\eta_1, \eta_2), (\eta_\mathrm{3},\eta_\mathrm{4}) \}}}{%
{\ottdrulename{sc\_if}}{}%
}}

\newcommand{\ottdrulewfSCseq}[1]{\ottdrule[#1]{%
\ottpremise{\Gamma  \vdash  \ottnt{e_{{\mathrm{1}}}}  \ottsym{:}  \ottnt{(t_{{\mathrm{1}}},\eta_{{\mathrm{1}}})} \outputeqs{\nes_{{\mathrm{1}}}}}%
\ottpremise{\Gamma  \vdash  \ottnt{e_{{\mathrm{2}}}}  \ottsym{:}  \ottnt{(t_{{\mathrm{2}}},\eta_{{\mathrm{2}}})} \outputeqs{\nes_{{\mathrm{2}}}}}%
}{
\Gamma  \vdash   \ottnt{e_{{\mathrm{1}}}} ; \ottnt{e_{{\mathrm{2}}}}   \ottsym{:}  \ottnt{(t_{{\mathrm{2}}},\eta_{{\mathrm{2}}})} \outputeqs{\nes_{{\mathrm{1}}} \cup \nes_{{\mathrm{2}}}}}{%
{\ottdrulename{sc\_seq}}{}%
}}

\newcommand{\ottdefnwfSCExpression}[1]{\begin{ottdefnblock}[#1]{$\vdash_{\text{MAP}}^{\text{PC}_1,\text{PC}_2}  \ottnt{e_1} \propto \ottnt{e_2}$ and  $\vdash_{\text{MAP}}^{\text{PC}_1,\text{PC}_2}  \ottnt{T_1} \propto \ottnt{T_2}$
}{}\\[5pt]
\ottusedrule{\ottdrulewfSCvar{}} \qquad
\ottusedrule{\ottdrulewfSClet{}}\\[5pt]
\ottusedrule{\ottdrulewfSCcall{}}\qquad
\ottusedrule{\ottdrulewfSCcast{}}\\[5pt]
\ottusedrule{\ottdrulewfSCselect{}}\qquad
\ottusedrule{\ottdrulewfSCloc{}}\qquad
\ottusedrule{\ottdrulewfSCnull{}}\\[5pt]
\ottusedrule{\ottdrulewfSCupdate{}}\qquad
\ottusedrule{\ottdrulewfSCnew{}}\\[5pt]
\ottusedrule{\ottdrulewfSCfj{}}\\[15pt]
\ottusedrule{ \ottdrulewfSCAsync{}}
\ottusedrule{ \ottdrulewfSCException{} \qquad
\ottdrulewfSCThread{}}
\end{ottdefnblock}}

\begin{table}
\caption{Syntactic Correspondence for Expressions and Thread Collections}
\label{fig:type-expr}
\begin{scriptsize}
\hrule\text{}\\[2pt]
\ottdefnwfSCExpression{}\\[5pt]
 \hrule
\end{scriptsize}
\end{table}

\renewcommand{\ottdrulewfXXtXXasync}[1]{\ottdrule[#1]{%
\ottpremise{\Gamma  \vdash  \ottnt{T_{{\mathrm{1}}}}  \ottsym{:}  \ottnt{(t_{{\mathrm{1}}},\eta_{{\mathrm{1}}})}\outputeqs{\nes_{{\mathrm{1}}}}}%
\ottpremise{ \Gamma  \vdash  \ottnt{T_{{\mathrm{2}}}}  \ottsym{:} \ottnt{(t_{{\mathrm{2}}},\eta_{{\mathrm{2}}})}\outputeqs{\nes_{{\mathrm{2}}}}}%
\ottpremise{\Gamma  \vdash  \ottnt{e}  \ottsym{:}  \ottnt{(t,\eta)}\outputeqs{\nes}}%
}{
\Gamma  \vdash   \ottnt{T_{{\mathrm{1}}}}  \mathop{||}  \ottnt{T_{{\mathrm{2}}}}  \rhd  \ottnt{e}   \ottsym{:}  \ottnt{(t,\eta)}\outputeqs{\nes_{{\mathrm{1}}} \cup \nes_{{\mathrm{2}}} \cup \nes}}{%
{\ottdrulename{wf\_t\_async}}{}%
}}

\renewcommand{\ottdrulewfXXtXXthread}[1]{\ottdrule[#1]{%
\ottpremise{\Gamma  \vdash  \ottnt{e}  \ottsym{:}  \ottnt{(t,\eta)}\outputeqs{\nes}}%
}{
\Gamma  \vdash  \ottsym{(}  \mathcal{L}  \ottsym{,}  \ottnt{e}  \ottsym{)}  \ottsym{:}  \ottnt{(t,\eta)}\outputeqs{\nes}}{%
{\ottdrulename{wf\_t\_thread}}{}%
}}

\renewcommand{\ottdrulewfXXtXXexn}[1]{\ottdrule[#1]{%
\ottpremise{\vdash  \ottnt{t}}%
\ottpremise{\vdash  \Gamma}%
}{
\Gamma  \vdash  \textbf{EXN}  \ottsym{:}  \ottnt{(t,\bot)} \outputeqs{\emptyset} }{%
{\ottdrulename{wf\_t\_exn}}{}%
}}

\renewcommand{\ottdrulewfXXcfg}[1]{\ottdrule[#1]{%
\ottpremise{\Gamma  \vdash  \ottnt{H}}%
\ottpremise{ \Gamma  \vdash  \ottnt{V}  \ottlinebreakhack }%
\ottpremise{\Gamma  \vdash  \ottnt{T}  \ottsym{:}  \ottnt{(t,\eta)}\outputeqs{\nes}}%
\ottpremise{\ottnt{H}  \vdash_{\mbox{\tiny lock} }  \ottnt{T}}%
}{
\Gamma  \vdash  \langle  \ottnt{H}  \ottsym{;}  \ottnt{V}  \ottsym{;}  \ottnt{T}  \rangle  \ottsym{:}  \ottnt{(t,\eta)}\outputeqs{\nes}}{%
{\ottdrulename{wf\_cfg}}{}%
}}

\off{
\begin{table}
\caption{Typing of \rccc configurations}
\label{fig:rc3_typing_rules_configs}
\begin{scriptsize}
\hrule
\ottdefnwfXXHeap{}\ottdefnwfXXFields{}\ottdefnwfXXVars{}\ottdefnwfXXThreads{}\ottdefnwfXXLocking{}\ottdefnwfXXCfg{}
 \hrule
\end{scriptsize}
\end{table}
}

A Program Context Correspondence is defined as a pair of maps, the first relating method names, and the second relating types, which establishes a correspondence from a program context (which can be thought of as the final one) into another (which can be thought of as the original one) in such a way that: the types given to fields of each class in the original program correspond to those given to fields in  corresponding classes of the final program; and  the name, types and method bodies of methods in the original program correspond to those given to methods in corresponding classes of the final program.
\begin{definition}[Program Context Correspondence] \label{def:program-context-correspondence}
A pair of maps\newline $\kw{MAP} = (\kw{methodMap},\kw{typeMap})$, where $\kw{methodMap}:\methodN\rightarrow\methodN$
and $\kw{typeMap}:\Type\rightarrow\Type$, is said to establish a correspondence from program context $PC_2$ (and implicit $\ottkw{fields}_2$, $\ottkw{methods}_2$) to program context $PC_1$ (and implicit $\ottkw{fields}_1$, $\ottkw{methods}_1$), written $PC_1\leq_{\kw{MAP}} PC_2$, 
if:
\begin{enumerate}
    \item for all $C_2,f$, we have that\\ $\ottkw{fields}_1(\kw{typeMap}(C_2))(f) = \kw{typeMap}(\ottkw{fields}_2(C_2)(f))$;
    \item for all $t_2,m_2$ such that $\ottkw{methods}_2(t_2)(m_2) = x : t_2' \rightarrow t_2'',e_2$, for some $t_2',t_2''$, we have that\\
    $\ottkw{methods}_1(\kw{typeMap}(t_2))(\kw{methodMap}(m_2)) =  x : \kw{typeMap}(t_2') \rightarrow \kw{typeMap}(t_2''), e_1$ and $\vdash_{\kw{MAP}}^{\text{PC}_1,\text{PC}_2} e_1 \propto e_2$.
\end{enumerate}
\end{definition}

\off{
\begin{definition}[Program Context Correspondence]\label{def:program-context-correspondence} %

A pair
$\kw{MAP}=(\kw{methodMap},$ $\kw{typeMap})$ of maps, where
$\kw{methodMap}:\methodN\rightarrow\methodN$
and
$\kw{typeMap}:\Type\rightarrow\Type$,
\noindent is said to establish a correspondence from program context $PC_2$ to program context $PC_1$, written $PC_1\leq_{\kw{MAP}} PC_2$, 
if:
\begin{enumerate}
    \item the types given to fields of each class in the original program ($PC_1$) correspond to those given to fields in corresponding classes of the final program ($PC_2$);
    \item the name, types and method bodies of methods in the original program correspond to those given to methods in corresponding classes of the final program.
\end{enumerate}
\end{definition}
}

\off{
\begin{definition}[Thread Correspondence]
We say that $MAP=(\kw{methodMap},\kw{typeMap})$ establishes a correspondence from thread collection $T_1$ to $T_2$ if $\vdash_{\kw{MAP}}^{\text{PC}_1,\text{PC}_2} T_1 \propto T_2$.
\end{definition}
}

We consider herein the dynamic semantics of \oolong, and its run-time constructs, which we assume well-formed (Figure 6 and~7 in~\cite{oolong-ACR2019}). Oolong's \emph{configurations} $\cfg HVT$ include heaps $H$ mapping abstract locations to objects, variable maps $V$, and collections of threads~$T$.

A program context correspondence is further said to define a correspondence from a heap and its implicit program context (which can be thought of as the final one) to another heap and its implicit program context (which can be thought of as the original one) if the two heaps have the same domain, and the heaps assign corresponding classes, the same field map and lock status to all locations.
\begin{definition}[Heap Correspondence] \label{def:heap-correspondence}
A program context correspondence $\kw{MAP}$ from $PC_2$ to $PC_1$ is said to define a correspondence from heap $H_2$ (and implicit program context $\text{PC}_2$) to heap $H_1$ (and implicit program context $\text{PC}_1$), written $H_1\leq_{\kw{MAP}}^{\text{PC}_1,\text{PC}_2} H_2$, iff:
\begin{enumerate}
    \item $\dom(H_1)=\dom(H_2)$, and
    \item for all $\iota \in \dom(H_1)$:
$H_1(\iota)=(\kw{typeMap}(C_2),F,L)$ iff $H_2(\iota) = (C_2,F,L)$.
\end{enumerate}
\end{definition}
\off{
\begin{definition}[Heap Correspondence] %
\label{def:heap-correspondence}
A program context correspondence $\kw{MAP}$ from $PC_2$ to $PC_1$ is said to define a correspondence from heap $H_2$ (and implicit program context $\text{PC}_2$) to heap $H_1$ (and implicit program context $\text{PC}_1$), written $H_1\leq_{\kw{MAP}}^{\text{PC}_1,\text{PC}_2} H_2$, iff:
\begin{enumerate}
    \item the two heaps have the same domain, and
    \item the heaps assign corresponding classes, the same field map and lock status to all locations.
\end{enumerate}
\end{definition}
}

\noindent This means that, if $\kw{MAP}=(\kw{methodMap},\kw{typeMap})$ is a program context correspondence, then $\kw{typeMap}$ is homomorphic with respect to $\ottkw{fields}$, and $\kw{methodMap}$ is homomorphic with respect to $\ottkw{msigs}$. Furthermore, if $\kw{MAP}$ is a heap correspondence, then it is homomorphic with respect to the structure of the heap. 
The following bisimulation relation
is designed to relate thread collections that derive from programs between which there is a program context correspondence, and which preserve a heap correspondence throughout all possible execution paths.
It is inspired by the notion of abstraction homomorphism, which relates processes with structurally simpler but semantically equivalent ones~\cite{Castellani87}. The configuration steps in the definition below derive from the small-step dynamic semantics of \oolong (\cf Figures 8 and 9 in~\cite{oolong-ACR2019}).
\begin{definition}[Bisimulation]
\label{def:almost-bisimulation}
Given a program context correspondence $\kw{MAP}$ from $PC_2$ to $PC_1$, a $(\kw{MAP},PC_1,PC_2)$-bisimulation is a binary relation $\mathcal{B}$ on well-formed %
thread collections, that satisfies, for all $T_1,T_2,H_1,H_2,V$:
\begin{itemize}
    \item if $T_1~\mathcal{B}~T_2$  and $H_1 \leq_{\kw{MAP}}^{\text{PC}_1,\text{PC}_2} H_2$ then:
    \begin{enumerate}
        \item $\cfg{H_1}{V}{T_1} \rightarrow \cfg{H_1'}{V'}{T_1'}$ implies that exist $H_2',T_2'$ such that\\
$\cfg{H_2}{V}{T_2} \rightarrow \cfg{H_2'}{V'}{T_2'}$ and $H_1' \leq_{\kw{MAP}}^{\text{PC}_1,\text{PC}_2} H_2'$ and $T_1'~\mathcal{B}~T_2'$;
\item $\cfg{H_2}{V}{T_2} \rightarrow \cfg{H_2'}{V'}{T_2'}$ implies that exist $H_1',T_1'$ such that\\
$\cfg{H_1}{V}{T_1} \rightarrow \cfg{H_1'}{V'}{T_1'}$ and $H_1' \leq_{\kw{MAP}}^{\text{PC}_1,\text{PC}_2} H_2'$ and $T_1'~\mathcal{B}~T_2'$.
    \end{enumerate}
\end{itemize}
\end{definition}
The set of pairs of values is a $(\kw{MAP},PC_1,PC_2)$-bisimulation. Furthermore, the union of a family of $(\kw{MAP},PC_1,PC_2)$-bisimulations is a $(\kw{MAP},PC_1,PC_2)$-bisimulation. Consequently, the union of all $(\kw{MAP},PC_1,PC_2)$-bisimulations is a $(\kw{MAP},PC_1,PC_2)$-bisimulation, the largest $(\kw{MAP},PC_1,PC_2)$-bisimulation, which we denote by  ${\simeq_{\kw{MAP}}^{\text{PC}_1,\text{PC}_2}}$.
\off{One can show that the largest $(\kw{MAP},PC_1,PC_2)$-bisimulation, denoted by ${\simeq_{\kw{MAP}}^{\text{PC}_1,\text{PC}_2}}$, exists.
}

The final result states that when a program undergoes the \func{\rccc{}Analysis}, the initial and final programs simulate each other by means of a $(\kw{MAP},PC_1,PC_2)$-bisimulation as defined above.
To prove it we will define a concrete bisimulation and use an auxiliary Replacement lemma.

\begin{lemma}[Replacement w.r.t. $\propto$]\label{anexo:replacement}\text{}
\begin{enumerate}
    \item If $\vdash_{\kw{MAP}}^{\text{PC}_1,\text{PC}_2} eE_1 \propto E_2[e_2]$, then exist $E_1[\bullet],e_1$ such that $eE_1 = E_1[e_1]$ and $\vdash_{\kw{MAP}}^{\text{PC}_1,\text{PC}_2} e_1 \propto e_2$.
    Furthermore, if $\vdash_{\kw{MAP}}^{\text{PC}_1,\text{PC}_2} \hat{e}_1 \propto \hat{e}_2$, then $\vdash_{\kw{MAP}}^{\text{PC}_1,\text{PC}_2} E_1[\hat{e}_1] \propto E_2[\hat{e}_2]$.
    \item If $\vdash_{\kw{MAP}}^{\text{PC}_1,\text{PC}_2} E_1[e_1] \propto eE_2$, then exist $E_2[\bullet],e_2$ such that $eE_2 = E_2[e_2]$ and $\vdash_{\kw{MAP}}^{\text{PC}_1,\text{PC}_2} e_1 \propto e_2$.
    Furthermore, if $\vdash_{\kw{MAP}}^{\text{PC}_1,\text{PC}_2} \hat{e}_1 \propto \hat{e}_2$, then $\vdash_{\kw{MAP}}^{\text{PC}_1,\text{PC}_2} E_1[\hat{e}_1] \propto E_2[\hat{e}_2]$.
\end{enumerate}
\end{lemma}
\begin{proof} By induction on the structure of $E_2$ and $E_1$, respectively. 
\end{proof}

\off{\begin{lemma}[$\vie$ contained in $\propto$]\label{lemma:syntactic-correspondence}
For all $\sol$, $\maine$, $e$, we have that $\vdash_{\kw{MAP}}^{\text{PC}_1,\text{PC}_2} e \propto \vie(\sol,\maine,e)$ where $\kw{MAP}=(\kw{methodMap},\kw{typeMap})$, where $\kw{methodMap}=\kw{fst}\circ\oomn^{-1}$ and $\kw{typeMap}=\kw{snd}\circ\ootype^{-1}$.
\end{lemma}
\begin{proof}
By induction on the definition of $\vie(\sol,\nu,e)$:
\end{proof}
}

We prove that, for typable programs, $\propto$ is a Bisimulation:
\begin{lemma}[$\propto$ is a Bisimulation]Given a program context correspondence $\kw{MAP}$ from $PC_1$ to $PC_2$, $\propto$, when restricted to typable programs, %
is a $(\kw{MAP},PC_1,PC_2)$-bisimulation\label{lemma:concrete-bisimulation}.
\end{lemma}
\begin{proof}
Since typability is preserved by reduction, it is enough t prove that if $PC_1\leq_{\kw{MAP}} PC_2$, and $H_1\leq_{\kw{MAP}}^{\text{PC}_1,\text{PC}_2} H_2$, and $\vdash_{\kw{MAP}}^{\text{PC}_1,\text{PC}_2} T1 \propto T2$, then:
\begin{enumerate}
    \item $\cfg{H_1}{V}{T_1} \rightarrow \cfg{H_1'}{V}{T_1'}$ implies\\
$\cfg{H_2}{V}{T_2} \rightarrow \cfg{H_2'}{V}{T_2'}$ and $\vdash_{\kw{MAP}}^{\text{PC}_1,\text{PC}_2} T1' \propto T2'$ and $H_1'\leq_{\kw{MAP}}^{\text{PC}_1,\text{PC}_2} H_2'$.
    \item $\cfg{H_2}{V}{T_2} \rightarrow \cfg{H_2'}{V}{T_2'}$ implies\\
$\cfg{H_1}{V}{T_1} \rightarrow \cfg{H_1'}{V}{T_1'}$ and $\vdash_{\kw{MAP}}^{\text{PC}_1,\text{PC}_2} T1' \propto T2'$ and $H_1'\leq_{\kw{MAP}}^{\text{PC}_1,\text{PC}_2} H_2'$.
\end{enumerate}
For (2) We distinguish the following cases (the cases for (1) are analogous):\begin{enumerate}
    \item If $T_2=(\mathcal{L},e_2)$, then exist $E_2[\bullet],\hat{e}_2,\hat{e}_2'$ such that $e_2=E_2[\hat{e}_2]$ and $\cfg{H_2}{V}{(\mathcal{L},\hat{e}_2)} \rightarrow \cfg{H_2'}{V'}{\hat{T}_2'}$, and:
    \begin{enumerate}
        \item %
        $\hat{T_2}'=(\mathcal{L},\hat{e}_2')$, and
        $T_2'=(\mathcal{L},E_2[\hat{e}_2'])$; %
        \item $\hat{e}_2=\kwool{finish} \, \,\startblock\, \, \kwool{async} \, \,\startblock\,  \ottnt{\hat{e}_2^1}  \,\finishblock\, \, \kwool{async} \, \,\startblock\,  \ottnt{\hat{e}_2^2}  \,\finishblock\,  \,\finishblock\, \ottsym{;}  \ottnt{\hat{e}_2^3}$, and  $\hat{T}_2'=(\mathcal{L},\hat{e}_2{}^1) \mathop{||}  (\mathcal{L},\hat{e}_2{}^2) \outputeqs {\hat{e}_2^3}$, and  $T_2'=\hat{e}_2^1 \mathop{||}  \hat{e}_2^2 \outputeqs {E_2[\hat{e}_2^3]}$.
    \end{enumerate}
    Assuming, without loss of generality, that $\hat{e}_2$ is the smallest in the sense that it cannot be written in terms of a non-empty context, we prove these cases using the Replacement Lemma~\ref{anexo:replacement}, and by case analysis on the transitions $\cfg{H_2}{V}{\hat{e}_2} \rightarrow \cfg{H_2'}{V'}{\hat{e}_2'}$.
    \item If $T_2=\hat{T}_2{}^1 \mathop{||} \hat{T}_2{}^2 \outputeqs{ \hat{e}_2 }$, and
    \begin{enumerate}
    \item $\cfg{H_2}{V}{\hat{T}_2{}^1} \rightarrow \cfg{H_2'}{V}{\hat{T}_2{}^1{}'}$
 and $T_2'=\hat{T}_2{}^1{}' \mathop{||} \hat{T}_2{}^2 \outputeqs{ \hat{e}_2 }$
    \item $\cfg{H_2}{V}{\hat{T}_2{}^2} \rightarrow \cfg{H_2'}{V}{\hat{T}_2{}^2{}'}$
 and $T_2'=\hat{T}_2{}^1 \mathop{||} \hat{T}_2{}^2{}' \outputeqs{ \hat{e}_2 }$
    \item $\hat{T}_2{}^1=(\mathcal{L},\hat{v}_2{}^1)$, $\hat{T}_2{}^2=(\mathcal{L}',\hat{v}_2{}^2)$ and $T_2' = (\mathcal{L},\hat{e}_2)$.
    \end{enumerate}
    These cases are easy to prove, using $\ottdrulename{sc\_async}$ and \ottdrulename{ dyn\_eval\_async\_*}.
    \end{enumerate} 
\end{proof}

Finally, to formulate the Behavioural Correspondence theorem we denote by $P = PC[e]$ a program $P$ that is composed of a program context $PC$, \ie its list of interfaces and classes, and a main expression $e$.
\begin{theorem}[Behavioural Correspondence] \label{th:soundness}
Consider an original program $P_\textrm{orig}$, that is annotated according to the \rccc model such that $\func{\rccc{}Analysis}(P_\textrm{orig}) = (P,Sol,P^+)$, where there are program contexts $PC$ and $PC^+$, and main expressions $e$ and $e^+$, such that $P=PC[e]$ and $P^+=PC^+[e^+]$.
Then, for $\kw{MAP}=(\kw{methodMap},\kw{typeMap})$, where $\kw{methodMap}=\kw{fst}\circ\oomn^{-1}$ and $\kw{typeMap}=\kw{snd}\circ\ootype^{-1}$, 
we have that $PC\leq_{\kw{MAP}} PC^+$ and ${e} {\simeq_{\kw{MAP}}^{\text{PC},\text{PC}^+}} {e^+}$.
\end{theorem}
\begin{proof}
By Lemma~\ref{lemma:concrete-bisimulation}, we have that the relation
\[B = \{(T_1, T_2) |  \vdash_{\kw{MAP}}^{\text{PC},\text{PC}^+} T_1 \propto T_2\}\]
is a $(\kw{MAP},PC,PC^+)$-bisimulation. 
We then show that $({e},e^+) \in B$. %
According to Stage 4 of the analysis, $P^+ = \vi{P}(\sol, P)$, with $e^+ = \vi{E}(\sol, \maine, e)$. We show by induction on the definition of $ \vi{E}$ that for all $\sol$, $\maine$, $e$, we have that $\vdash_{\kw{MAP}}^{\text{PC},\text{PC}^+} e \propto \vi{E}(\sol, \maine, e)$.
\end{proof}

\off{
\begin{theorem}[Behavioural Correspondence] %
\label{th:soundness}
Consider an original program $P_\textrm{orig}$, that is annotated according to the \rccc model such that $\func{\rccc{}Analysis}(P_\textrm{orig}) = (P,Sol,P^+)$, where there exist program contexts $PC$ and $PC^+$, and main expressions $e$ and $e^+$, such that $P=PC[e]$ and $P^+=PC^+[e^+]$.
Then, for $\kw{MAP}=(\kw{methodMap},\kw{typeMap})$, where $\kw{methodMap}=\kw{fst}\circ\oomn^{-1}$ and $\kw{typeMap}=\kw{snd}\circ\ootype^{-1}$, 
we have that ${e} {\simeq_{\kw{MAP}}^{\text{PC},\text{PC}^+}} {e^+}$.
\end{theorem}
}
\newcommand{\dnc}{\mathbb{X}}

\section{The Java Implementation}
\label{sec:java}

In this section we  
overview  how the \func{\rccc{}Analysis} has been implemented in \rccc-Java.
The \rccc-Java static analysis is built on top of the Eclipse Java Development Tools (\href{https://www.eclipse.org/jdt/}{\color{gray}JDT})%
,
but are independent of the Eclipse IDE, being easily incorporated in any Java project via the \href{https://gradle.org}{\color{gray}Gradle Build Tool}%
. 
The compilation receives Java source code with \rccc annotations and generates Java bytecode with the \texttt{@Atomic} annotations forwarded, which will be important for the subsequent lock inference analysis that is performed over the Java bytecode.

To implement \func{\rccc{}Analysis} in Java 
we need to accommodate primitive types 
that do not exist in \oolong.
Their impact on the atomicity inference rules is minimal.
Atomic fields of any type have to be accessed in units of work to avoid data races and atomicity violations.
However, fields of primitive types store values that are immutable and, hence, the \textit{atomicity property} of such fields does not have to be propagated to 
parameters, return values or local variables.
For convenience, we associate a new  \textbf{do\_not\_care}  atomicity value 
to primitive types, but these are filtered out when added to a constraint system.

To implement the \func{solve} function presented in \S\ref{sec:valid_variants}, we resorted to the Z3 solver~\cite{DBLP:conf/tacas/MouraB08}.
Both atomicity values
and validity values  are mapped to Booleans:
 $\atomick$ and \validk to \truek, and   $\natomick$ and \invalidk to \falsek.
For atomicity values, we favour the $\natomick$ value to reduce the overhead of the concurrency control mechanisms in the execution of the generated code.
In turn, for validity values, we favour the \validk value to have as much valid variants as possible.
Accordingly, we want to obtain a model that is optimal relatively to the number of valid variants and the number of variables with atomicity value $\natomick$.
For that purpose, we use Z3's optimising solver over a set of boolean formulas and added soft constraints for every validity and atomicity variable. 

The challenge of solving a program's constraint system is dealing with 
scale, as the number and complexity of classes increases.
When moving from theory to practice, we designed our approach to support separate compilation.
Each class is therefore processed separately, producing a set of constraints that indicate the validity 
value computed for each variant of its non-private methods. 
The process features two main steps.
The \textbf{first} step (\textit{Locally Valid})  guarantees that only variants that are \textit{locally valid} make their way to the class' constraint system.
Remembering the definition of the \vvMCS function in \S\ref{sec:valid_variants}, 
we have that\[\footnotesize (\valeq{\valvar(\mu)}{\validk}) \rightarrow
\Big((\nes \cup \{ \ateq{\natx_1}{\nu_1}, \ateq{\natx_2}{\nu_2}, \ateq{\natx_3}{\nu_3}\})@\mu ~\wedge \func{bindCall}(\mns, \mu)\Big)\]
We say that a variant is \textit{locally valid} if and only if the left term of the conjunction has a solution.
The existence of this solution depends only on the constraints collected from method's code, being completely independent of  any method calls. 
Accordingly, all non \textit{locally valid} variants are  removed from the set of variants to consider from such point forward.
 
The \textbf{second} step (\textit{BindCalls})  creates and solves the class' constraint system. 
However, to improve the compilation time, instead of building a single system and solving it all at once (like in \S\ref{sec:valid_variants}),
we build and solve this system incrementally, by making use of Z3's ability to organise constraints as a stack.
We first generate the class' call graph and from there guarantee that, before being processed, 
every method only progresses if all its dependencies have already been processed, down to the nodes with  out-degree $0$.
Naturally, precautions have to be made to deal with recursion.

The constraint system is initially populated with equations on the validity of the variants with  out-degree $0$ (computed on the \textit{Locally Valid} step).
From then onward, the incremental solving  works as follows:
we select a method (with out-degree > 1) from the call-graph and, having all its dependencies processed, we push the set of constraints for the validity of each of its variants onto the constraint system (stack).%
\footnote{The incremental solving is done at method level, meaning that we solve the system once  per  method. It could also be done at variant level, but the granularity is too small, yielding longer compilation times.}
We then solve the system and, from the resulting model (if any), 
retrieve the valid variants and the corresponding atomicity values for their local variables.
Lastly, we pop the constraints previously pushed and complement the constraint system with validity of the variants freshly processed.  
Hence, at any point, the constraints system is composed of
the constraints concerning the validity of the variants of the methods that have already been processed, plus the constraints for the validity of the  variants of the method being processed.
The algorithm stops when there are no more methods to process.

The result of the analysis of each class is cached on disk, in a user-configurable folder. 
Prior to the analyses of any class, a
lookup in the cache is made. 
If information is present and its date is more recent than the
class' file date, it is retrieved and no new analysis is performed.
This means that the analysis' time
only focuses on types that have been modified since the previous compilation. 

\section{Related Work}
\label{sec:related}

We review approaches to static analysis of concurrent object calculi/languages to prevent atomicity violations. We concentrate on those analysing source code and either requiring or inferring annotations/types to control the critical resources. Therefore, we briefly present two of the main approaches: access permissions and behavioural types;
we review briefly reference works on semantic correctness criteria for language encodings;
the main part is dedicated to closely related work on data-centric concurrency control.

\paragraph*{DCCC in shared memory programming.} %
\textit{Atomic Sets} (aka AJ)~\cite{atomicsets-POPL2006,DBLP:journals/toplas/DolbyHMTVV12,DBLP:conf/ecoop/VaziriTDHV10} is one of the reference works in the area.
 Variables holding values that share consistency properties
have to be placed in an \textit{atomic} set  
by prefixing the variable's declaration with \texttt{atomic(a)}, where \texttt{a} denotes the set's identifier.
Sets may have multiple units of work associated, which can be explicitly augment to, for instance, account for multiple method parameters. 
Likewise,  alias annotations enable the union of sets from distinct classes at object creation.
Although a seminal work, Atomic Sets propose several annotations that hampers reasoning and is error-prone, as some annotations may be easily forgotten.
Moreover, progress is not guaranteed. To avoid deadlocks, the programmer may have to intervene to explicitly define a partial order between sets when the  static analysis is not able to infer it
~\cite{deadlocks-atomicsets-ICSE2013}.

AWJS~\cite{DBLP:conf/pppj/0001DB16} combines Atomic Sets with work-stealing-based task parallelism in  the Java language.  
Other works have addressed the automatic inference of 
atomic sets.
AJ-lite~\cite{aj-lite-Fool2012} is a lighter version  that assumes a single atomic set per Java class and is only applicable to libraries and not entire programs.
The number of annotations  is reduced to three, the ones needed to relate the class' atomic set  to  the ones of its fields.
In \cite{inference-atomicsets-PASTE2013}, the authors propose to
automatically infer  atomic sets from patterns recognised in execution traces.
Although being able to infer most of the required annotations, the quality of the result is sensitive to the quality of the input traces, and 
may generate more annotations than necessary.

Ceze \emph{at al.}~\cite{DBLP:conf/asplos/CezePCMT08} associate variables sharing consistency proprieties to a colour, defining a consistency domain. However, concurrency control is not centralised on data declaration. Code annotations have to be added to the methods' implementations to signal when a thread concludes its work on a domain. Moreover, atomic-region-like control-centric concurrency control is needed to handle \emph{high-level data races}~\cite{DBLP:conf/nddl/ArthoHB03} in composite operations. 
In~\cite{DBLP:conf/hpca/CezeMPT07}, some of the same authors proposed hardware support for data-centric synchronisation.

RC$^3$~\cite{rc3-sac16} proposes a DCCC model that uses a single keyword (\atomicS), being a source of inspiration for \rccc. Most of the concerns delegated on the programmer in AJ are shifted to static analyses. However, as mentioned in \S\ref{sec:intro}, the methods' implementations are atomicity-aware, requiring the duplication of code to account for the atomicities of the parameters. Also, as in all other DCCC solutions, there is no support for interfaces.

\paragraph*{Concurrency control via access permissions and protocol compliance.}
A popular approach to provide static control of data-races is the use of explicit or implicit access permissions. The key ideas were adopted by the Rust programming language which ensures memory-safety via its type system.
Sadiq \emph{et al.} provide a comprehensive survey of the state-of-the-art~\cite{DBLP:journals/jss/SadiqLL20}. %
The idea is to either declare, or have an implicit policy about, when resources are readable and/or writable. To write on a resource, one needs exclusive access; reads may be shared but forbid writes.
Access permissions do not avoid the use of explicit concurrency control. They are a way of ensuring safe units of work. Our approach is "declarative" and aims to  generate locks in a correct-by-construction manner. %

Static verification of protocol compliance avoids atomicity violations by combining %
tight aliasing control and enforcing specific sequences of method calls. %
It uses behavioural types to declare object usage protocols %
and force any sequence of calls to be one of the allowed paths; it uses access permissions to control aliasing and guarantee protocol compliance and completion (see~\cite{DBLP:journals/ftpl/AnconaBB0CDGGGH16,DBLP:journals/csur/HuttelLVCCDMPRT16} for comprehensive surveys of the state-of-the-art; a book on tools is also available~\cite{GayRavara-BehavTypesTools17}).
In object-oriented languages and frameworks supporting behavioural types, one needs to provide protocol specifications for each class and define contracts for methods to protect resources. The annotation burden is significant and might not be easy to get right, though the correctness guarantees go beyond those provided by our approach.

\paragraph*{Type qualifier inference and Type-directed code synthesis.}
Generic frameworks and tools for type qualifier inference have been proposed, \eg \textsc{CQual}~\cite{DBLP:conf/pldi/FosterFA99, DBLP:journals/toplas/FosterJKA06} and \textsc{Clarity}~\cite{Chin06inferenceof} for C, 
and JQual~\cite{DBLP:conf/oopsla/GreenfieldboyceF07} for Java, and their usage has been experimented for different purposes (see the later reference for a review).
The \rccc model performs a field-based, context-sensitive, flow-insensitive, analysis, as does Greenfieldboyce and Foster's CS/FS JQual, %
but for specifically inferring atomicity type qualifiers to achieve strong concurrency control guarantees.
Imposing program-wide atomicity constraints on how data is manipulated, in a flow-insensitive approach in the presence of alias appears to be incompatible with enabling subtyping between the atomic and non-atomic qualifiers. Indeed, the same object should not be treated simultaneously as atomic and non-atomic by different parts of the program. Automatic generation of method variants by the \rccc model provides the possible flexibility that is compatible with the methods' code, up to the conservative nature of the constraint generation.

Our code generation approach bears connections with
Osera's~\cite{Osera19} constraint-based type-directed program synthesis technique for synthesising polymorphic code from types and {\it examples}. In their work, a type inference system generates constraints between types; these constraints are solved, in order to infer types and ultimately synthesise code. The main differences are that \cite{Osera19} targets a functional programming language, and that \rccc's code generation is based on an original program.

\paragraph*{Correctness criteria for language encodings.}
When performing code transformation, %
\emph{behavioural soundness}, \ie the preservation of the original program's functional behaviour, is a key property. %
To define this criteria rigorously, a standard approach is to define a notion of mutual simulation between the source and the target programs, 
for proving that the transformation in question \emph{preserves and reflects observable behaviour}
(original and generated code simulate each other, being indistinguishable for an external observer)~\cite{Mitchell96,Reynolds98}.
Bisimulations are standard relations to capture indistinguishably of two systems' observable behaviour~\cite{Park81,Milner80} and are used extensively in process calculi to reason about language encodings and expressiveness~\cite{GorlaNestmann16}. %

\off{The notion of abstraction homomorphism laid by Castellani, which relates processes with structurally simpler but semantically equivalent ones~\cite{Castellani87}, inspired our "almost bisimulation" based on a heap correspondence.
}
\section{Conclusions}
\label{sec:conclusions}

We propose herein a sound new model of data-centric concurrency control, that only requires the annotation of interfaces and of atomic class fields.
The approach has two phases --- an analysis, %
which infers missing annotations and produces type-safe atomicity-related overloaded versions of each method; followed by a lock inference phase, to manage concurrency and prevent interferences.
In this paper we rigorously define and prove correct the first phase for a foundational object calculus---\oolong---, showing that the generated code is type-safe and behaviourally corresponds to the original one.
Moreover, we apply the algorithms developed for \oolong also to a Java implementation, to show its viability in real-life code.

We reached a milestone of our approach: we provide a computer-verified mathematical proof that %
the resulting generated program is type-safe. This is a fundamental prerequisite to assess the validity of the analysis technique and in turn to propose its use in mainstream languages. We achieved that by relying on a minimal, yet expressive Object Oriented language, and by mechanising the program analysis and code generation phases in the Coq proof assistant. 

In future work, we plan to extend the setting to deal with generics. While the computation of units of work and lock inference has been implemented, we are working on its formalisation, and aim to prove its soundness.
We also plan to evaluate the use of AtomiS in the simplification of writing concurrent code, as well as its impact in teaching concurrent programming.
We are currently developing a mechanised proof of \rccc's behavioural soundness, as described in section \ref{sec:bis}.

\bibliography{paper}

\newpage
\appendix
\section{\rccc Specialisation of the Classes of Listing~\ref{lst:baselist}}
\label{app:baselist}

\begin{table}[h]
\begin{lstlisting}[multicols=2, escapechar=|]
interface List_n<T> { 
 void add(T element);
 T get(int pos); 
 @All Boolean equals(@All List_n<T> other); 
} 

interface List_a<T>  { 
 void add(@Atomic T element);
 @Atomic T get(int pos); 
 @All Boolean equals(@All List_a<T> other); 
} 

class Node_n<T> { 
 T value;
 Node_n<T> next, prev;
}

class Node_a<T> { 
 @Atomic T value;
 Node_a<T> next, prev;
}

class BaseList_nn<T, N extends Node_n<T>> implements List_n<T> {
 N head, tail;
 Supplier<N> factory;

 BaseList_nn(Supplier<N> factory) {...}
 void add(T element) {...}
 T get(int pos) {...}
 Boolean equals(List_n<T> other) {...}
}


class BaseList_an<T, N extends Node_a<T>> implements List_a<T> {
 N head, tail;
 Supplier<N> factory;

 BaseList_an(Supplier<N> factory) {...}
 void add(T element) {...}
 T get(int pos) {...}
 Boolean equals(List_a<T> other) {...}
}

class BaseList_na<T, N extends Node_n<T>> implements List_n<T> {
 @Atomic N head, tail;
 Supplier<N> factory;

 BaseList_na(Supplier<N> factory) {...}
 void add(T element) {...}
 T get(int pos) {...}
 Boolean equals(List_n<T> other) {...}
}

class BaseList_aa<T, N extends Node_a<T>> implements List_a<T> {
 @Atomic N head, tail;
 Supplier<N> factory;

 BaseList_aa(Supplier<N> factory) {...}
 void add(T element) {...}
 T get(int pos) {...}
 Boolean equals(List_a<T> other) {...}
}
\end{lstlisting}
\end{table}
\section{Auxiliary functions}
\label{app:aux}

\begin{footnotesize}
\begin{align*}
\func{strip}: \text{\rccc-\oolong Program} \rightarrow\ & \text{\oolong Program} \\
\func{strip}(\Ids\ \Cds\ e) =\ & 	\func{strip}(\Ids)\ \func{strip}(\Cds)\ \func{strip}(e) \\
\func{strip}( \kw{interface}\ I \{ \IMsigs \})  = \ &   \kw{interface}\ I \{ \func{strip}(\IMsigs) \} \\
\func{strip}(\Msig)  = \ &  \Msig \\
\func{strip}(\Msig\ [\Sas])  = \ &  \Msig \\
\func{strip}(\kw{class}\ C\ \kw{implements}\ I\ \SB \Fds~\Mds \FB) =\ & \kw{class}\ C\ \kw{implements}\ I\ \SB \func{strip}(\Fds)~\func{strip}(\Mds) \FB\\
\func{strip}(f : \atomicS\ t)  = \ &  f : t \\
\func{strip}(\kw{def}\ \Msig \SB e \FB) =\ & \kw{def}\ \Msig \SB \func{strip}(e) \FB) \\
\func{strip}(x.f \kc{=} e)  = \ & x.f \kc{=}\func{strip}(e)  \\
\func{strip}(x.m(e))  = \ & x.m(\func{strip}(e))  \\ 
\func{strip}(\kw{let}\ x \kc{=} e_1\ \kw{in}\ e_2) = \ & \kw{let}\ x \kc{=} \func{strip}(e_1)\ \kw{in}\ \func{strip}(e_2) \\
\func{strip}(\newk\ \atomicS\ C)  = \ &  \newk\ C \\
\func{strip}((t)~e) = \ & (t)~\func{strip}(e)\\
\func{strip}(\kw{finish}\SB \kw{async}\SB e_1\FB~\kw{async}\SB e_2\FB \FB\kc{;}\ e_3) = \ & \kw{finish}\SB \kw{async}\SB \func{strip}(e_1)\FB~\kw{async}\SB \func{strip}(e_2)\FB \FB\kc{;}\ \func{strip}(e_3)\\
\func{strip}(e) =\ & e \qquad \textrm{for all other cases} 
\end{align*}
\end{footnotesize}

\begin{footnotesize}
\begin{align*}
\func{dress}: \text{\rccc-\oolong Program} \times \text{\oolong Program} \rightarrow\ & \text{\rccc-\oolong Program} \\
\func{dress}(\Ids\ \Cds\ e, \Ids'\ \Cds'\ e') =\ & 	\func{dress}(\Ids, \Ids')\ \func{dress}(\Cds, \Cds')\ \func{dress}(e, e') \\
\func{dress}(\kw{interface}\ I \{ \IMsigs \},  \kw{interface}\ I \{ \IMsigs' \})  = \ &   \kw{interface}\ I \{ \func{dress}(\IMsigs, \IMsigs') \} \\
\func{dress}(\Msig\ \Sas, \Msig)  = \ &  \Msig\ \Sas\\
\func{dress}(\Msig\ , \Msig)  = \ &  \Msig \\
\func{dress}(
        \kw{class}\ C\ \kw{implements}\ I\ \SB \Fds~\Mds \FB,
        \kw{class}\ C\ \kw{implements}\ I\ \SB \Fds'~\Mds' \FB
    ) = \hspace{-12em}&\\
     = \ &\ \kw{class}\ C\ \kw{implements}\ I\ \\
     & \SB \func{dress}(\Fds, \Fds')~\func{dress}(\Mds, \Mds') \FB\\
\func{dress}(\kw{def}\ \Msig \SB e \FB, \kw{def}\ \Msig \SB e' \FB) =\ & \kw{def}\ \Msig \SB \func{dress}(e, e') \FB) \\
\func{dress}(f : \atomicS\ t, f : t)  = \ &  f : \atomicS\ t \\
\func{dress}(f : t, f : t)  = \ &  f : t \\
\func{dress}(x.f=e, x.f=e')=& \ x.f=\func{dress}(e,e')\\ 
\func{dress}(x.m(e), x.m(e')= & \ x.m(\func{dress}(e,e'))\\ \func{dress}(\kw{let}\ x \kc{=} e_1\ \kw{in}\ e_2, \kw{let}\ x : t \kc{=} e_1'\ \kw{in}\ e_2' )  = \ &  \kw{let}\ x : t \kc{=} \func{dress}(e_1, e_1')\ \kw{in}\ \func{dress}(e_2, e_2')  \\
\func{dress}(\newk\ \atomicS\ C, \newk\ C )  = \ &  \newk\ \atomicS\  C \\
\func{dress}((t)e,(t)e')= & \
(t)\func{dress}(e,e')
\\
\func{dress}(\kw{finish}\SB \kw{async}\SB e_1\FB~\kw{async}\SB e_2\FB \FB\kc{;}\ e_3),
\\
\kw{finish}\SB \kw{async}\SB e_1'\FB~\kw{async}\SB e_2'\FB \FB\kc{;}\ e_3') = &
\
\kw{finish}\SB \kw{async}\SB \func{dress}(e_1,e'_1)\FB~\kw{async}\SB \func{dress}(e_2,e'_2)\FB \FB\kc{;}\\ 
&\ \func{dress}(e_3,e'_3)
\\
\func{dress}(e, e) =\ & e \qquad \textrm{for all other cases} 
\end{align*}
\end{footnotesize}

\begin{footnotesize}
\begin{align*}
Set@\mu =\ & \bigwedge_{x\in Set} x@\mu\\
(\eta_1 \vee \eta_2)@\mu  =\ & (\eta_1@\mu  \vee \eta_2@\mu) \\
(\eta_1 \wedge \eta_2)@\mu  =\ & (\eta_1@\mu  \wedge \eta_2@\mu) \\
(\eta)@\mu  =\ & (\eta@\mu) \\
(\ateq{\alpha_1}{\alpha_2})@\mu  =\ & (\ateq{\alpha_1@\mu}{\alpha_2@\mu}) \\
\alpha@\mu = \ & \alpha \text{ if } \alpha \neq \check x\\
\check x@\mu = \ & \check x@\mu (\in \natV, \text{ overloaded operation})
\end{align*}
\end{footnotesize}

\begin{footnotesize}
\begin{align*}
\ootype: \natval  \times \Type & \rightarrow \Type \\
\ootype(\nu, s ) & =  \nu\_s \\ 
\oomn: \methodN \times \natval \times \natval & \rightarrow \methodN \\
\oomn(m, \nu_1, \nu_2) & = m\_\nu_1\_\nu_2\\ 
\oosig: \natval \times \natval \times \Msig & \rightarrow \Msig \\
\end{align*}
\begin{align*}
\oosig(\nu_1, \nu_2, m(x: s_1) : s_2) & =
\oomn(m, \nu_1, \nu_2)\ (x : \ootype(\nu_1, s_1)) : \ootype(\nu_2, s_2) %
\end{align*}
\end{footnotesize}

\end{document}